\documentclass[letterpaper,11pt]{article}
\usepackage{amsfonts}
\usepackage[pdftex]{graphicx}
\usepackage{amsmath,amssymb,amsthm}
\usepackage{natbib}
\usepackage{hyperref}
\usepackage{bm}
\usepackage{fancyhdr}
\usepackage{sectsty}
\usepackage[notcite,notref,final]{showkeys}
\usepackage{paralist} 
\usepackage{tikz}
\usepackage{stmaryrd}
\usepackage{setspace}
\usepackage{todonotes}
\usepackage{multirow}
\usepackage{subcaption}
\usepackage{floatrow}
\usepackage{booktabs}
\usepackage{threeparttable}
\usepackage{algorithm}
\usepackage{algpseudocode}
\usepackage{array, tabularx, tabulary, booktabs, arydshln}

\usepackage{comment}
\usepackage{caption}
\usepackage{subcaption}

\definecolor{light-gray}{gray}{0.8}

\newtheoremstyle{myplain}
  {9pt}
  {9pt}
  {\itshape}
  {\parindent}
  {\scshape}
  {:}
  {.4em}
  {}
\newtheoremstyle{mydefinition}
  {9pt}
  {9pt}
  {\itshape}
  {\parindent}
  {\scshape}
  {:}
  {.4em}
  {}
\newtheoremstyle{myremark}
  {9pt}
  {9pt}
  {}
  {\parindent}
  {\scshape}
  {:}
  {.4em}
  {}
\theoremstyle{myplain}
\newtheorem{theorem}{Theorem}
\newtheorem{corollary}{Corollary}
\newtheorem{lemma}{Lemma}
\newtheorem{proposition}{Proposition}
\theoremstyle{mydefinition}
\newtheorem{assumption}{Assumption}
\newtheorem{definition}{Definition}
\theoremstyle{myremark}
\newtheorem{example}{Example}
\newtheorem{remark}{Remark}
\newtheorem{myalgorithm}{Algorithm}

\usepackage{geometry}                  	   
\usepackage{setspace}   
\renewcommand{\cite}{\citet}
\def\argmax{\mathop{\rm arg\,max}}
\def\argmin{\mathop{\rm arg\,min}}

\newcommand{\cA}{\mathcal{A}}
\newcommand{\cC}{\mathcal{C}}
\newcommand{\cD}{\mathcal{D}}

\newcommand{\cI}{\mathcal{I}}

\newcommand{\cM}{\mathcal{M}}

\newcommand{\cR}{\mathcal{R}}

\newcommand{\cW}{\mathcal{W}}
\newcommand{\cX}{\mathcal{X}}
\newcommand{\cY}{\mathcal{Y}}
\newcommand{\cZ}{\mathcal{Z}}
\newcommand{\cV}{\mathcal{V}}

\newcommand{\X}{W}
\newcommand{\x}{w}

\makeatletter
\renewcommand*\env@matrix[1][\arraystretch]{%
	\edef\arraystretch{#1}%
	\hskip -\arraycolsep
	\let\@ifnextchar\new@ifnextchar
	\array{*\c@MaxMatrixCols c}}
\makeatother

\hypersetup{colorlinks=true, linkcolor=blue, citecolor=blue} 

\newcommand{\citeposs}[1]{\citeauthor{#1}'s \citeyearpar{#1}}

\numberwithin{equation}{section}
\title{Testing Exclusion and Shape Restrictions in Potential Outcomes Models\thanks{We thank Isaiah Andrews, Yuehao Bai, Guilherme Duarte, Francesca Molinari, Pepe Montiel Olea, Ismael Mourifi\'e, Max Tabord-Meehan, Alex Torgovitsky, Davide Viviano, and seminar participants at the Bristol Econometrics Studying Group, University of Amsterdam, CEME conference, Cornell, Harvard/MIT, JER/SNSF workshop in Hokkaido, Keio, Kyoto, SMU, U. Toronto, and Washington U. St Louis for comments.
We gratefully acknowledge financial support from NSF grants SES-2520364 (Ponomarev) and SES-2520365 (Kaido).}} 
\author{
Hiroaki Kaido \\
Boston University
\and
Kirill Ponomarev\\
University of Chicago
}
\begin{document}
\maketitle

\begin{abstract}

Exclusion and shape restrictions play a central role in defining causal effects and interpreting estimates in potential outcomes models. 
To date, the testable implications of such restrictions have been studied on a case-by-case basis in a limited set of models.
In this paper, we develop a general framework for deriving testable implications of support restrictions on potential response functions, based on a novel graph-based representation of the model. The framework provides a unified and constructive method for obtaining valid inequalities and establishes their sharpness under a regularity condition. We illustrate the approach in several popular settings, including instrumental variables, treatment selection, mediation, and interference. As an empirical application, we revisit the US Lung Health Study and test for
the presence of spillovers between spouses, specification of exposure maps, and persistence
of treatment effects over time.
\end{abstract}

\clearpage
\section{Introduction}

The potential outcomes framework of \citet{neyman1923applications} and \citet{rubin1974estimating} is a workhorse of empirical economics. Within this framework, exclusion and shape restrictions on the potential response functions are commonly imposed to define causal parameters, interpret estimates, and tighten identified sets. Although such restrictions concern fundamentally unobservable objects, they have non-trivial testable implications. So far, such testable implications have been studied on a case-by-case basis in a limited set of models. Examples include testing instrumental variables (IV) validity and treatment selection monotonicity in \citet{Kitagawa15}, ``judge IV'' designs in \citet{frandsen2023judging}, encouragement designs in \citet{BaiHuangTabordMeehan2026}, and full mediation in \citet{Kwon:2024aa}, among others.

In this paper, we propose a general framework for deriving testable implications of exclusion and shape restrictions, and for testing these implications against data. The framework allows practitioners to empirically assess commonly imposed modeling assumptions in a broad class of potential outcome models. Formally, the observed data consist of an \textit{exogenous input} $Z \in \mathcal{Z}$, typically representing instruments or exogenous treatments, and an \textit{endogenous response} $R \in \mathcal{R}$, which may include an endogenous treatment, mediator, or outcome. Additional observed covariates can be accommodated but are omitted for now. The model is characterized by the \textit{potential response vector} $R^* = (R^*(x))_{x\in \mathcal{X}} \in \mathcal{R}^{\mathcal{X}}$, indexed by counterfactual inputs $x \in \mathcal{X}$, which typically contain instrumental variables, treatment variables, or mediators, and a known mapping $T: \mathcal{R}^{\mathcal{X}} \times \mathcal{Z} \to \mathcal{R}$ such that the observed response satisfies $R =T(R^*, Z)$. We assume that $R^*$ and $Z$ are independent. In this notation, a wide range of restrictions studied in the literature take the form of \textit{support restrictions}, i.e., $R^* \in \mathcal{R}^*$ almost surely, for a known set $\mathcal{R}^* \subseteq \mathcal{R}^{\mathcal{X}}$; See Table~\ref{tab:support_restrictions} for the examples.

 \begin{table}[htbp]
\renewcommand{\arraystretch}{1.2}
\newcolumntype{L}[1]{>{\raggedright\arraybackslash}m{#1}}
\centering
\resizebox{\textwidth}{!}{%
\begin{threeparttable}
\caption{Examples of Support Restrictions in Potential Outcome Models}
\label{tab:support_restrictions}

\begin{tabular}{@{}L{5cm}L{8cm}L{6cm}@{}}
\toprule
\multicolumn{1}{c}{\textbf{Assumption}} & 
\multicolumn{1}{c}{\textbf{Restriction}} & 
\multicolumn{1}{c}{\textbf{Notes}} \\
\midrule
\multicolumn{3}{l}{\textbf{A: Outcome}} \\[1mm]
Monotone Response & $d\le d' \Rightarrow Y(d)\le Y(d')$ & $d:$ treatment in an ordered set \\[1mm]
Exclusion & $Y(d,z)=Y(d,z'),\forall d,~ z\ne z'$ & $d:$ treatment, $z$: IV \\[1mm]
Semi IV/Partial Exclusion & $Y(0,z_0,z_1)=Y(0,z_0,z_1'),\forall z_0, z_1\ne z_1'$ and  $Y(1,z_0,z_1)=Y(1,z_0',z_1),\forall z_1, z_0\ne z_0'$ & $z_d:$ semi IV relevant for selecting $d$ and excluded from $Y(1-d)$\\[1mm]
Full Mediation & $Y(m,d)=Y(m,d'),\forall d\ne d'$ & $m:$ mediator, $d$: treatment \\[1mm]
No Interference/SUTVA & $z_i=z'_i \Rightarrow Y_i(z)=Y_i(z')$ & $z\in \cZ^N$, $N$: \# of units \\[1mm]

No Anticipation & $d_{i,1:t}=d'_{i,1:t}\Rightarrow Y_{i,t}(d_{i,1:T}) = Y_{i,t}(d'_{i,1:T})$ & $d_{i,1:t}:$ unit $i$'s history of treatement status up to period $t$ \\[1mm]
Exposure Mapping & $D_i(z)=D_i(z') \Rightarrow Y_i(z)=Y_i(z')$ & $z\in \cZ^N:$ treatment across $N$ units, $D_i(z):$ unit $i$'s exposure level \\[1mm]
Exposure Semimonotonicity & $D_i(z)\le D_i(z') \Rightarrow Y_i(z)\le Y_i(z')$ & \\[3mm]

\multicolumn{3}{l}{\textbf{B: Treatment selection}} \\
LATE Monotonicity &  $D(z)\ge D(z')$ or $D(z')\ge D(z), \forall z,z'$ & $D:$ binary treatment, $z$: IV  \\[1mm]
Program Substitution & $D(1)\ne D(0)\Rightarrow D(1)= h$ & $D\in \{n,c,h\}$,  $z:$ binary IV \\[1mm]
Encouragement Design & $D(z)\in \argmax_{j\in\{1,\dots,J\}}\beta_j I\{z=j\}+V_j$ & $D:$ multivalued treatment, $z:$ multivalued IV \\[1mm]
Partial Monotonicity & $D(z_l,z_{-l}) \ge D(z'_l,z_{-l})$ or $D(z_l,z_{-l}) \le D(z'_l,z_{-l})$ & $z:$ vector-valued IV\\[1mm]
Unordered Monotonicity & For any $z,z',d$, $\bm{1}\{D(z)=d\} \ge \bm{1}\{D(z')=d\}$ or $\bm{1}\{D(z)=d\} \le \bm{1}\{D(z')=d\}$ & $D:$ multivalued treatment \\[1mm]
Choice Restriction  & $D(z)\in A\Rightarrow D(z')\in B$ & $A$ and $B$ are determined by a revealed preference analysis. \\[1mm]
Two-way Flow & $D(z)=0$ iff $V_1<Q_1(z)$ \& $V_2<Q_2(z)$, $D(z)=1$ iff $V_1>Q_1(z)$ \& $V_2>Q_2(z)$, $D(z)=2$ otherwise & This class includes the double hurdle model: $D(z)=1\{V_1<Q_1(z),V_2<Q_2(z)\}$.\\[3mm]

\multicolumn{3}{l}{\textbf{C: Others}} \\[1mm]
Sample Selection & $S(1)\ge S(0)$ & $S(d):$ potential sample selection indicator  \\[1mm]
        & $S(1)\ge S(0)$ if $W\in A$,  $S(1)\le S(0)$ if $W\in B$ & $A,B:$ known sets of covariate values  \\
\bottomrule
\end{tabular}
\begin{tablenotes}
\small
\item \textit{Notes:} This table contains examples of assumptions that restrict the support of potential response variables.
\item For assumptions used in outcome models, we refer to the following selected sources: \cite{Manski:1997aa} (monotone response), \cite{bruneelzupanc2025dontfullyexcludeme} (semi IV), \cite{Kwon:2024aa} (mediation), \cite{Abbring:2007aa,Bojinov:2021aa} (no anticipation), \cite{Manski:2013aa,Aronow:2017aa} (exposure-related assumptions).
\item For assumptions used in treatment selection models, we refer to the following selected sources: \cite{Imbens:1994aa} (LATE monotonicity), \cite{Kline:2016aa} (program substitution), \cite{BaiHuangTabordMeehan2026} (encouragement), \cite{Mogstad:2021aa} (partial monotonicity), \cite{heckman2018unordered} (unordered monotonicity \& choice restriction), \cite{KitamuraStoye2018} (choice restriction), \cite{Lee:2018aa} (two-way flow).
\item For sample selection models, we refer to \cite{Lee:2009aa,heiler2024treatmentevaluationintensiveextensive}.
\end{tablenotes}
\end{threeparttable}
}
\end{table}

To characterize the observable implications of support restrictions, we introduce a graph-based representation that encodes the compatibility structure of the model. This representation allows us to compute the testable implications using efficient graph algorithms and facilitates interpretation of the resulting constraints. Moreover, by exploiting the correspondence between the graph and the convex polytope of observable distributions compatible with the model, we establish that, under a regularity condition, the implied restrictions are sharp, i.e., necessary and sufficient for the observed distribution to be consistent with the imposed assumptions.

 Suppose, for the moment, that both $\mathcal{R}$ and $\mathcal{Z}$ are finite sets.  We may represent a restriction $\mathcal{R}^*$ via an undirected graph $G = (V_G, E_G)$, in which every vertex $v_{r, z} \in V_G$ corresponds to a pair $(r, z)$ of values the response and input can take, and two vertices $v_{r, z}$ and $v_{r', z'}$ are connected by a link $(v_{r, z}, v_{r', z'}) \in E_G$ if there exists a support point $r^* \in \mathcal{R}^*$ consistent with both, i.e., $T(r^*, z) = r$ and $T(r^*, z') = r'$. In other words, the vertices $v_{r, z}$ and $v_{r', z'}$ in $G$ are connected if a pair of events ``$R = r$ when $Z = z$'' and ``$R = r'$ when $Z = z'$'' is compatible with the modeling assumptions. 
 
  With this representation, the testable implications of the modeling assumptions become transparent. If two vertices $v_{r, z}$ and $v_{r', z'}$ are disconnected in $G$, the events ``$R = r$ when $Z = z$'' and ``$R = r'$ when $Z = z'$'' cannot occur jointly, under the modeling assumptions. Thus, it must be that $P(R = r \,|\, Z = z) + P(R = r' \,|\,Z = z') \leq 1$. More generally, any maximal set $I \subseteq V_G$ of  mutually disconnected vertices (called a \textit{Maximal Independent Set} or \textit{MIS}) yields an inequality
  \begin{equation} \label{eq:intro_MIS}
  \sum_{v_{r, z} \in I} P(R = r \,|\,Z = z) \leq 1. 
  \end{equation}
Our first main result (Theorem \ref{thm:testable_implication}) shows that all such MIS inequalities provide valid testable implications. These inequalities can be efficiently computed using existing algorithms available in standard software packages and tested using a wide range of established procedures, including those that accommodate continuous response or control variables when necessary.

Next, we study the sharpness of the testable implications derived in \eqref{eq:intro_MIS}. We observe that the MIS inequalities are generally a strict subset of a broader class of inequalities that yield sharp testable implications of the model. While this more general class is required in principle, many commonly used support restrictions exhibit additional structure under which the MIS inequalities alone suffice for sharpness. We refer to such restrictions as \textit{regular} and show that regularity can be verified in a straightforward manner. In particular, monotonicity restrictions, encouragement designs, and LATE-type settings, including judge-IV designs, listed in Table~\ref{tab:support_restrictions}, are all regular, as are all restrictions with binary $Z$. Our second main result (Theorem~\ref{thm:sharpness_of_mis}) establishes that, under regularity, the MIS inequalities are sharp. By contrast, certain exclusion restrictions in IV, mediation, and interference are not regular, in which case the inequalities in \eqref{eq:intro_MIS} are insufficient. 

We then extend the analysis to settings in which (part of) the response vector $R$ is continuous. We obtain general analogues of Theorems~\ref{thm:testable_implication} and \ref{thm:sharpness_of_mis}, and specialize them to the canonical IV model with endogenous treatment and instrument monotonicity. This extension generalizes existing results in \citet{Kitagawa15} and \citet{Kwon:2024aa} by allowing for multi-valued treatments and IVs.

Conceptually, our analysis builds on classical results in combinatorial optimization --- most notably the polyhedral characterization of clique polytopes for perfect graphs \citep{chvatal1975certain}.  However, in our setting, the relevant object is not the classical clique polytope of $G$. Since each response type includes exactly one observable response per input $z\in \{z_1,\dots,z_K\}$, every admissible type corresponds to a \emph{$K$-clique} in $G$. Thus, the set of conditional probability vectors $(P(R = r \,|\,Z = z))_{r \in \mathcal{R}, z \in \mathcal{Z}}$ compatible with the support restriction corresponds to the convex hull of incidence vectors of such $K$-cliques. We provide a novel explicit half-space characterization of this set and show how it yields sharp testable implications. Under a regularity
condition, the characterization
simplifies to MIS inequalities, delivering a tractable set of moment inequalities for inference.

Another useful feature of our approach is that adding or relaxing assumptions corresponds to modifying the edges of the underlying graph. This structure enables direct comparison of testable implications across different models. By systematically computing and comparing the resulting sets of testable implications, practitioners can evaluate how changes in modeling assumptions affect the implied restrictions --- an analysis that is often difficult to conduct on a case-by-case basis.

Expressing the model’s testable implications as conditional moment inequalities allows us to draw on a large literature on testing and model selection procedures that accommodate both discrete and continuous covariates. Our primary approach exploits the half-space representation of the set of conditional distributions of observables consistent with the model. We also consider an alternative procedure based on the vertex representation of this set, following \citet{bai2025inference}.
The choice between these representations depends on the structure of the support restriction $\mathcal{R}^*$ and the geometry of the associated polytope; Section~\ref{sec:computation_inference} provides practical guidance. Broadly, the half-space-based procedure is computationally attractive when the number of implied inequalities is moderate, even if the number of latent types is very large. By contrast, the vertex-based approach can be advantageous when the number of inequalities is extremely large, but the number of latent types remains manageable. To assess which regime is relevant in a given application, we recommend using the output-sensitive algorithm of \citet{tsukiyama1977new} to enumerate maximal independent sets as a preprocessing step.

We assess the finite-sample performance of the proposed tests using Monte Carlo experiments and an empirical application. In simulations, we find that conditioning on exogenous covariates and exploiting instruments with rich support can substantially improve power, although some violations of the model are inherently difficult to detect. As an empirical illustration, we revisit the U.S. Lung Health Study, a randomized controlled trial of smoking cessation interventions. We test for spillover effects from treated subjects to their spouses, monotonicity of responses to treatment intensity, and persistence of treatment effects over time. Across these hypotheses, the number of maximal independent sets ranges from 1 to 726, and all are computed within three seconds, demonstrating that the proposed tests are computationally straightforward to implement in practice. 

\medskip 

\noindent 
\textbf{Related Literature} \; This paper contributes to the vast and growing literature on identification of causal effects using instrumental variables. In a setting with binary instrument and treatment, \citet{Imbens:1994aa} demonstrated that the exclusion restriction, combined with treatment selection monotonicity (which we will call IA monotonicity), suffices to identify the Local Average Treatment Effect (LATE) for compliers. This work started a large literature studying identifying power, interpretation, extensions, and testable implications of support-type restrictions in potential outcome models; see, e.g., \citet{mogstad2024instrumental} for a detailed review and Table \ref{tab:support_restrictions} for some notable examples. Below, we discuss in more detail several papers that are most related to the present work. 	

\citet{balke1997bounds} derived the testable implications of the exclusion restriction and instrument independence, both with and without IA monotonicity, in models with binary instruments, treatments, and outcomes. 
\citet{Kitagawa15} extended this analysis to continuous outcomes, showing that the resulting testable implications under IA monotonicity are sharp and proposing a corresponding statistical test.\footnote{\citet{mourifie2017testing} propose an alternative test based on \citet{CLR13}.} 
\citet{Kwon:2024aa} further broadened the framework to accommodate multi-valued treatments while maintaining the same notion of monotonicity. 
As we discuss in Remark~\ref{rem:artstein}, with binary instruments, the sharp testable implications in the above settings (and in many related ones) follow from a theorem of \citet{artstein1983distributions}, which provides necessary and sufficient conditions for the existence of a joint distribution of two random elements given their marginals and joint support. For regular restrictions, we extend Artstein’s result to a multi-marginal setting, which allows us to handle multi-valued instruments and alternative treatment–selection models.\footnote{
With multi-valued instruments in the absence of IA monotonicity, the inequalities derived by \citet{balke1997bounds} for binary outcomes and treatments are not sharp \citep{Bonet2001Instrumentality}. 
The corresponding sharp testable implications were later established by \citet{KedagniMourifie2020}. Related work also considers relaxations of the random-assignment assumption; see \citet{RichardsonRobins2010BinaryIV} and \citet{KedagniMourifie2020}.} 
 	 
 	Several extensions of IA monotonicity to settings with multi-valued instruments have been proposed in the literature. One example is the so-called encouragement designs, where, for each potential treatment choice, there is a randomly assigned instrument that encourages an individual to take up that treatment; see, e.g., \citet{kline2016evaluating, kirkeboen2016field}. In this context, \citet{BaiHuangTabordMeehan2026} derive sharp testable implications of the modeling assumptions and discuss available testing procedures. The general approach developed here recovers the inequalities of \citet{BaiHuangTabordMeehan2026} and applies in many other settings, including those with restrictions on the outcome response beyond exclusion.

 	Another common extension of the IA monotonicity assumption is the so-called ``judge fixed effects'' design, originally motivated by the random assignment of cases to judges in the US penitentiary system \citep{kling2006incarceration} and used for identifying causal effects in a variety of settings \citep[see, e.g., Table 1 in][]{frandsen2023judging}. In addition to the standard instrument exogeneity and exclusion restrictions, the key underlying assumption is that the ``judges'' can be ordered from ``least'' to ``most lenient,'' so that ``judges'' assignment has a weakly monotone effect on each individual's treatment status. In this setting, the Two-Stage Least Squares (TSLS) estimand has a ``weakly causal'' interpretation of a non-negatively weighted average of LATEs for different complier groups. \citet{frandsen2023judging} derived (non-sharp) testable implications of these assumptions and proposed a non-parametric test. Recent work by \citet{coulibaly2024sharptestjudgeleniency} extended the analysis and derived sharp testable implications. The framework developed here allows for the derivation and testing of sharp testable implications under either class of assumptions, possibly in conjunction with other structural restrictions.
    	
In a related setting, \citet{Mogstad:2021aa} argued that IA monotonicity with multiple instruments may be overly restrictive and proposed a weaker notion, termed ``partial monotonicity,'' under which the TSLS estimand retains a weakly causal interpretation. This restriction is our running example discussed in detail throughout the paper. Recently, \citet{blandhol2022tsls} argued that in the presence of covariates, the TSLS estimand retains a weakly causal interpretation only if covariates enter non-parametrically. Our approach accommodates covariates and can be used to test the assumptions required for such causal interpretations.

 	This paper also contributes to the literature on mediation, which deals with testing causal pathways between the treatment and outcomes. We refer to \citet{vanderweele2016mediation} and \citet{huber2019review} for detailed reviews of the literature. \citet{Kwon:2024aa} developed a non-parametric test for the sharp null of full mediation --- an assumption stating that the (binary) treatment can be excluded from the potential outcomes given (multi-valued) mediators, potentially combined with additional restrictions on the effect of treatment on mediators. In this paper, we extend the analysis to accommodate any form of support restrictions on potential outcomes and mediators, which allows for testing other hypotheses of interest, accommodating multivalued treatment and continuous controls (see Example \ref{ex:mediation}).
 			
	Moreover, this paper contributes to the literature on causal inference under interference, which explores violations of the no-interference assumption of \citet{cox1958planning} and stable unit treatment value assumption in \citet{rubin1980randomization}. Following \citet{hudgens2008toward}, \citet{Manski:2013aa}, and \citet{Aronow:2017aa}, 
    researchers commonly specify the so-called exposure maps, which restrict how each individual is affected by treatments of their peers, in order to define, identify, and estimate causal parameters of interest. This requires taking a stance on the underlying structure of interactions, which may not always be known or simple.  If an exposure map is misspecified, the corresponding estimands may not have a clear causal interpretation, standard estimators may have undesirable statistical properties, and inference procedures may be invalid; see, e.g., \citet{sobel2006randomized}, \citet{savje2021average}, \citet{leung2022unconfoundedness}, \citet{savje2024causal}, \citet{auerbach2024discussion}, \citet{leung2024discussion}, \citet{menzel2025fixed}, and \citet{auerbach2021local} for detailed discussions and possible remedies. Existing tests for specification of exposure maps (applicable in a fixed population framework and based on randomization) operate on a case-by-case basis, rely on carefully constructed symmetries of the test statistic under the null hypothesis, and require a non-trivial choice of tuning parameters, including the focal units and test statistic, which greatly affect power \citep{aronow2012general, athey2018exact, basse2019randomization, puelz2022graph}.\footnote{In the same design-based setting, \cite{zhong2025unconditionalrandomizationtestsinterference} proposes an alternative unconditional testing framework to address these issues.} The approach we develop here (applicable in the super-population framework and valid in large samples) can be used to test the specification of any exposure map and many other hypotheses, such as monotonicity of the outcome response to treatment or endogeneity of spillover effects. In many settings, the resulting tests do not require any tuning parameters. We illustrate the capabilities of the proposed approach in a setting with interference in our empirical application. 
    
	Finally, this paper contributes to the literature on inference for linear systems and moment inequalities. With discrete inputs and responses, testing support restrictions amounts to testing whether a vector of conditional probabilities  $\beta(P) = (P(R = r \,|\,Z = z))_{r \in \mathcal{R}, z \in \mathcal{Z}}$ belongs to a convex polytope $\overline{\mathbb{P}} = \{A^*Q: Q \in \Delta(\mathcal{R}^*) \}$, where a matrix $A^* \in \{0, 1\}^{|\mathcal{R}||\mathcal{Z}| \times \mathcal{R}^*}$ encodes the support restriction $\mathcal{R}^*$, so that each column corresponds to a latent response type, and a vector $Q$ represents a distribution of latent types supported on $\mathcal{R}^*$. In this form, the hypothesis $P \in \overline{\mathbb{P}}$ can be tested using existing methods such as \citet{KitamuraStoye2018}, \citet{fang2023inference}, \cite{goff2025inferencevaluelinearprogram}, or \cite{bai2025inference}. However, this formulation does not accommodate continuous control variables, which may be essential in applications. The reason is essentially that the polytope $\overline{\mathbb{P}}$ is given in its vertex ($V$) rather than half-space ($H$) representation, which would correspond to testing conditional moment inequalities. In this paper, we analytically derive an $H$-representation for $\overline{\mathbb{P}}$ (and its continuous analog) that enables using existing approaches, such as \citet{andrews2013inference}, \citet{CLR13}, or \citet{armstrong2015asymptotically}, for inference. Additionally, the $H$-representation enables formal model selection as in \citet{shi2015model} and \citet{hsu2017model}.

\medskip 

\noindent \textbf{Notation}\;\; We will use the following notation thoughout. Let $V$ be a finite set with cardinality $N$. For any subset $S \subseteq V$,  $\bm{1}_S \in \{0, 1\}^N$ denotes a vector with elements $(\bm{1}(v \in S))_{v \in V}$. For any $N \in \mathbb{N}$, $\bm{1}_N \in \mathbf{R}^N$ denotes a vector of ones. The cardinality of a set $S$ is denoted $|S|$. For any set $V$, $\Delta(V)$ denotes the set of all probability distributions supported on $V$. For sets $S_1, \dots, S_K$, their cartesian product is denoted by $\prod_{k = 1}^K S_k$. The subscript $a.s.$ in $\leq_{a.s.}$ or $\forall_{a.s.}$ indicates that the corresponding statement holds almost surely. 

\medskip 
\noindent 
\textbf{Outline}\;\; The rest of the paper is organized as follows. Section 2 introduces the formal setup and motivating examples; Section 3 states the main result for regular support restrictions in discrete settings; Section 4 extends the discussion to continuous response variables and comparisons of multiple models; Section 5 discusses computation and inference; Section 6 presents Monte Carlo experiments; and Section 7 contains an empirical illustration.

\section{Basic Setup and Motivating Examples}

Recalling the setup in the introduction, the observed data consists of an exogenous input $Z \in \mathcal{Z}$, an endogenous response $R\in \cR$, and covariates $W \in \mathcal{W}$, which we omit for now. The set $\mathcal{Z} = \{z_1, \dots, z_K\}$ is assumed to be finite, and each $z\in \cZ$ realizes with positive probability. The model specifies a vector of potential responses $R^*=(R(x))_{x\in\cX}$, indexed by $x$ in a finite set $\cX$, and a mapping  $T:\cR^{\cX}\times \cZ\to \cR$ such that
 \begin{equation}
R=T(R^*,Z).\label{eq:obs}
 \end{equation}
 The potential response $R(x)$ represents the outcome value that would be realized if the input is set to $x$; this quantity is generally unobservable. Instead, the researcher observes the pair $(R,Z)$ that satisfies \eqref{eq:obs}.
 We assume that the input is exogenous in the following sense. 
\begin{assumption}[Exogeneity] \label{A:unconf}
	\begin{align}
			R^* \perp Z.
	\end{align}
\end{assumption}
Letting $Q \in \Delta(\mathcal{R}^{\mathcal{X}})$ denote the joint distribution of $R^*$, we consider general restrictions of the following type.
\begin{assumption}[Support Restrictions] \label{A:support} For some known $\mathcal{R}^* \subseteq \cR^\cX$,
	\begin{align}
	Q(R^* \in \mathcal{R}^*) = 1.\label{eq:support}
	\end{align}
\end{assumption}
The following examples illustrate.

\begin{example}[Instrumental Variables]\label{ex:iv}
Let $(Y, D, Z) \in \mathcal{Y} \times \mathcal{D} \times \mathcal{Z}$ denote the observed outcome, treatment, and instrument. Let $Y^* = (Y(d, z))_{d \in \mathcal{D}, z \in \mathcal{Z}}$ and $D^* = (D(z))_{z \in \mathcal{Z}}$ be potential outcomes and  treatments, respectively, and $R^*=(Y^*,D^*)$ be the potential response vector. In the notation of Equation \eqref{eq:obs}, $x = (d, z)$. Suppose that $Z$ is assigned randomly so that it satisfies Assumption \ref{A:unconf}. The mapping $T$ is characterized by 
\begin{align}
\underbrace{(Y,D)}_{R}=\underbrace{(Y(D(Z),Z),D(Z))}_{T(R^*,Z)}.	
\end{align}
Commonly used models restrict the support of $R^*$. For example,  \citet{Imbens:1994aa}(IA, henceforth) assume (i) exclusion: $Y(d,z)=Y(d,z'), \forall (d,z,z')\in \cD\times\cZ^2$ and (ii) IA monotonicity: for any $z,z'\in \cZ$, $D(z)\geq D(z')$ or $D(z)\leq D(z')$. \citet{machado2019instrumental} study combinations of various exogeneity and monotonicity assumptions. As a special case of this example, some papers focus on selection models. \cite{BaiHuangTabordMeehan2026} study the encouragement design; \cite{Lee:2018aa} consider selection processes characterized by multiple thresholds; and \cite{Mogstad:2021aa} consider partial monotonicity.
\qed
\end{example}

The next example will be our running example throughout the paper. 

\begin{example}[Partial Monotonicity]\label{ex:partial_monotonicity}
Let $D \in \{0, 1\}$ denote a treatment variable and $Z=(Z_1,Z_2) \in \{0,1\}^2$ instrumental variables. For example, $D$ is college attendance, $Z_1$ indicates whether the opportunity cost of college attendance is low, and $Z_2$ indicates the presence of a four-year college in a nearby area. Let $D(z_1, z_2)$ denote a potential treatment, and suppose the observed college attendance is generated by $D = D(Z_1, Z_2)$, where $Z$ is randomly assigned.
The partial monotonicity assumption of \cite{Mogstad:2021aa}\footnote{For illustration purposes, we assume the direction of the monotonicity assumption is known a priori as in \cite{Mogstad:2021aa} (Section III-B).} states
\begin{align}
D(1,z_2)\geq_{a.s.} D(0,z_2), \text{ for }z_2\in \{0, 1\}, \text{ and }D(z_1,1)\geq_{a.s.} D(z_1,0), \text{ for }z_1 \in\{0,1\}. \label{eq:pm_assumption}
\end{align} 
Unlike IA monotonicity in Example \ref{ex:iv}, this assumption does not require the ordering of $D(1,0)$ and $D(0,1)$. Denote the potential response vector by $R^*=(D(z_1,z_2))_{(z_1,z_2)\in \{0,1\}^2} \in \{0, 1\}^4.$
\cite{Mogstad:2021aa} show that the support $\cR^*$ consists of six compliance types
\begin{align}
\cR^*&=\{
 \underbrace{(0,0,0,0)}_{nt}
,\underbrace{(0,0,0,1)}_{rc}
,\underbrace{(0,0,1,1)}_{1c}
,\underbrace{(0,1,0,1)}_{2c}
,\underbrace{(0,1,1,1)}_{ec}
,\underbrace{(1,1,1,1)}_{at}\}.\label{eq:pm_support}
\end{align}
The support points are named as never taker (nt), reluctant compiler (rc), $Z_1$-complier (1c), $Z_2$-complier (2c), eager compiler (ec), and always taker (at). For example, the eager compiler takes up the treatment when at least one of the instruments takes the value one.
\qed
\end{example}

Our third example is a model of mediation. 

\begin{example}[Mediation] \label{ex:mediation}
Let $(Y, M, D) \in \mathcal{Y} \times \mathcal{M} \times \mathcal{D}$ denote an observed vector of outcomes, a vector of mediators, and a treatment variable. For example, following \citet{heckman2013understanding}, let $D \in \{0, 1\}$ represent assignment to the treatment group in an early childhood development program, $M = (M_{C}, M_{P})$ include measures of cognitive and personality skills, and $Y = (Y_{H}, Y_{L})$ include health behavior and labor market outcomes in adulthood. Let $Y^*=(Y(m,d))_{m\in\cM,d\in\cD}$ and $M^*=(M(d))_{d\in \cD}$ denote the potential outcomes and mediators. Let $R^*=(Y^*,M^*)$, and suppose the treatment $D$ is randomly assigned, i.e., statistically independent from $R^*$. Similar to Example \ref{ex:iv}, the relationship $(Y,M)=(Y(M(D),D),M(D))$ characterizes the mapping $T$. In this setting, one may be interested in testing different combinations of exclusion and shape restrictions. For example, the restriction $Y(m,d)=Y(m)$, for all $(d, m)\in \cD\times\cM$, in this context means the effect of the preschool program on both health and labor market outcomes is fully mediated by the measured skills, corresponding to the sharp null of \cite{Kwon:2024aa}. Alternatively, one may hypothesize that the labor market outcomes are affected by both cognitive skills and personality skills, while health outcomes are affected only by personality skills and some unobserved mediators. Formally, $Y_{H}(d, m_{C}, m_{P}) = Y_{H}(d, m_P)$, $Y_{L}(d, m_C, m_P) = Y_L(m_C, m_P)$, for all $d\in\mathcal{D}$, $(m_C, m_P) \in \mathcal{M}$. One may also impose that the program increases skills, $M_j(1) \geq M_j(0)$, for each $j \in \{C, P\}$, and that skills positively affect outcomes, that is, both $Y_H(d, m_p)$ and $Y_L(m_C, m_P)$ are partially monotone in their inputs, as in Example \ref{ex:partial_monotonicity}. 
\qed 
\end{example}

Our final example is a model of interference.
\begin{example}[Interference]\label{ex:interference}
Let $Y=(Y_1,\dots, Y_N)$ denote the vector of outcomes for $N$ units, and let $Z=(Z_1,\dots, Z_N)$ be the corresponding vector of treatments. Define $Y^*=(Y(z))_{z\in \mathcal Z}$ as the vector of potential outcomes. Each unit’s potential outcome may depend on the entire vector of treatment assignments, which relaxes the stable unit treatment value assumption (SUTVA).
In our empirical application, $Y_i$ indicates individual $i$’s smoking status, which may depend not only on $Z_i$, but also on the treatments received by others, thereby accommodating spillover effects through social interactions. The observed outcome is $Y = Y(Z)$, where $Z$ is randomly assigned. 
We adopt a super-population perspective in which $N$-unit groups (e.g., households) are observed repeatedly.

 One may restrict potential outcomes using \emph{exposure maps} or \emph{effective treatents} \citep{Manski:2013aa,Aronow:2017aa}, defined as nonrandom maps $d_i:\mathcal{Z}\to\mathcal{D}$ for $i=1,\dots,N$, where $d_i(z)$ represents individual $i$’s level of treatment exposure under assignment vector $z$.
A common assumption is
\begin{align}
d_i(z)=d_i(z')~\Rightarrow~Y_i(z)=Y_i(z'),	\label{eq:exposure}
\end{align}
which states that if individual $i$ experiences the same level of exposure under two assignment vectors $z$ and $z’$, their potential outcome remains the same.
Another assumption, introduced by \citet{Manski:2013aa}, is semi-monotonicity. When $\mathcal{D}$ is partially ordered, this restriction requires that outcomes be weakly monotone in exposure:
\begin{align}
d_i(z)\geq d_i(z')~\Rightarrow~Y_i(z)\geq Y_i(z').\label{eq:semimonotone}
\end{align} 
We test restrictions of this type in our empirical application in Section \ref{sec:empirical}.
\qed
\end{example}

\subsection{Graphs and related objects}
Our goal is to systematically derive testable implications of Assumptions \ref{A:unconf}--\ref{A:support}. First, we focus on the case where $R$ is discrete and obtain testable restrictions on the conditional probabilities: 
\begin{align}
\beta_0(P)=(P(R=r\,|\,Z=z),~r\in \cR_z,~ z\in \mathcal{Z})',
\end{align}
where $\cR_z$ is the conditional support of $R$ given $Z=z$. The discreteness of $R$ is relaxed in Section \ref{ssec:continuous_outcome}.
To facilitate the derivation, we introduce the following object.

\begin{definition}[Potential response graph] \label{def:graph}
A \emph{potential response graph}  $G=(V_G,E_G)$ associated with $\cR^*$ is an undirected graph with vertices $V_G = \{v_{r,z}: r \in \mathcal R_z, z \in \mathcal Z \}$ and edges $E_G=\{(v_{r,z},v_{r',z'})\in V_G\times V_G:\exists \, r^* \in \mathcal R^*: T(r^*,z) = r, T(r^*,z') = r'\}.$ 
\end{definition}

Each vertex $v_{r,z}$ represents a possible realization of $(R,Z)$, corresponding to the statement $R=r$ given $Z=z$ whose conditional probability can be evaluated. For each $r\in \cR_z$ and $z\in\cZ$, let $S(r,z)=\{r^*\in\mathcal R^*:T(r^*,z)=r\}$ be the set of latent types compatible with the observable pair $(r,z)$. The potential response graph $G$ is the \emph{intersection graph} of these sets. Its vertices are the observable pairs $((r_k,z_k)_{k=1}^K)$, and two distinct vertices are adjacent if and only if
\[
S(r_k,z_k)\cap S(r_\ell,z_\ell)\ne\varnothing.
\]
An edge therefore indicates that some latent type can generate both responses at their respective input values. 

For each $z_k \in \mathcal Z$, let
$$
V_k = \{v_{r,z_k} : r \in \mathcal R_{z_k}\}
$$
collect the vertices corresponding to $Z=z_k$. We write $V_G = \bigcup_k V_k$ throughout. Below, we illustrate the potential response graphs using examples. We will provide more details on the computation of $G$ in Section \ref{sec:computation_inference}.

For illustration, consider a special case of Example~\ref{ex:iv} in which 
$R^* = D^*$ (i.e., treatment selection only) and both $D$ and $Z$ are binary. 
The vertex set
\[
V_G = \{v_{0,0}, v_{1,0}, v_{0,1}, v_{1,1}\}
\]
enumerates all possible values of $(D, Z)$. 
Figure~\ref{fig:graph_ia_monotonicity} displays the potential response graph for the model that imposes the monotonicity restriction $D(1) \geq_{a.s.} D(0)$.
Consider the \emph{complier} type:
\[
r^* = (D(0), D(1)) = (0,1)\in \cR^*,
\]
which satisfies monotonicity. This type is consistent with observable pairs $(D, Z)=(D(Z),Z) = (0,0)$ and $(D', Z')=(D(Z'),Z') = (1,1)$, and hence induces an edge between $v_{0,0}$ and $v_{1,1}$. Equivalently, the edge exists because $S(0,0)\cap S(1,1)=\{nt,co\}\cap \{at,co\}\ne\emptyset.$\footnote{The remaining latent types $r^*=(0,0)$ and $r^*=(1,1)$ are commonly referred as the \emph{never-taker} (nt) and \emph{always-taker} (at). The other edges follow from the presence of nt and at, respectively.}

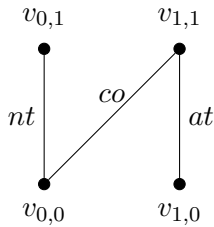
\begin{figure}
\centering
\begin{tikzpicture}[scale=1.8, every node/.style={draw, circle, fill=black, inner sep=1.5pt}]
\node (a) at (0, 0) {};
\node (b) at (1, 0) {};
\node (c) at (0, 1) {};
\node (d) at (1, 1) {};

\node[draw=none, fill=none, below] at (a) {$v_{0,0}$};
\node[draw=none, fill=none, below] at (b) {$v_{1,0}$};
\node[draw=none, fill=none, above] at (c) {$v_{0,1}$};
\node[draw=none, fill=none, above] at (d) {$v_{1,1}$};

\draw (a) -- (c) node[pos=0.5,draw=none, fill=none, left]  {$nt$};
\draw (a) -- (d) node[pos=0.5,draw=none, fill=none, above] {$co$};
\draw (b) -- (d) node[pos=0.5,draw=none, fill=none, right] {$at$};

\end{tikzpicture}
\caption{Potential Response Graph for Example \ref{ex:iv}.}
\label{fig:graph_ia_monotonicity}
\end{figure}

More generally, each support point induces a generalized notion of adjacent vertices. To see this, consider Example \ref{ex:partial_monotonicity}. 
Recall that $Z\in\{0,1\}^2$, and  $\cR_z=supp(D\,|\,Z=z)=\{0,1\}$ for any $k$.  Hence, $V_G=\{v_{r,z}\,|\,r\in \{0,1\},z\in\{0,1\}^2\}$.
Consider $r^*=(0,0,0,0)$, the never taker in \eqref{eq:pm_support}. The presence of this support point ensures the following vertices 
\begin{align}
v_{0,(0,0)}, v_{0,(0,1)},v_{0,(1,0)}, v_{0,(1,1)}
\end{align}
are mutually adjacent to each other. Indeed, each support point $r^*$ induces a \emph{maximal clique} $C$ of $G$, a set of mutually adjacent vertices.\footnote{Figure \ref{fig:graph_partial_monotonicity_1} in the Appendix show examples of such maximal cliques} \begin{definition}[Maximal Cliques]
A \emph{clique} $C$ is a subset of $V_G$ such that every two vertices in $C$ are adjacent. A \emph{maximal clique} is a clique that cannot be extended by including one more adjacent vertex.
\end{definition}
Fix a latent type $r^*\in\cR^*$ and consider the vertices $\{v_{r_k,z_k}:k=1,\dots,K\}$, where $r_k=T(r^*,z_k)$. Their associated sets $S(r_k,z_k)$ all contain $r^*$, so these vertices form a clique. The clique is maximal because it contains one vertex at each instrument value, and any additional vertex would specify a conflicting response at one of those values. Thus, every admissible latent type induces a maximal clique. The converse requires an additional condition, which is central to determining when the graph fully captures compatibility among observable responses. We discuss this condition in Section \ref{ssec:sharpness}.

The following notion plays a central role in deriving testable restrictions.
\begin{definition}[Maximal Independent Sets]
An \emph{independent set} $I$ is a subset of $V_G$ such that every two vertices in $I$ are not adjacent. A \emph{maximal independent set} is an independent set that cannot be extended by including one more independent vertex.
\end{definition}
For any $z$, vertices $v_{r,z}$ and $v_{r',z}$ are not adjacent because $R=r$ and $R=r'$ cannot occur simultaneously given $Z=z$. Hence, $V_G$ can be partitioned into $|\mathcal{Z}| =K$ independent sets $V_k,k=1,\dots,K$. Such a graph is called a \emph{$K$-partite graph}.

\begin{figure}[t]
\centering
\begin{subfigure}[t]{0.49\textwidth}
\centering
\vspace*{2mm}
\begin{tikzpicture}[scale=1.2, every node/.style={draw, circle, fill=black, inner sep=1.5pt}, baseline=(current bounding box.north)]
\node (a0) at (0, 0) {};
\node (a1) at (1, 0) {};
\node (b0) at (2, 1) {};
\node (b1) at (2, 2) {};
\node (c0) at (1, 3) {};
\node (c1) at (0, 3) {};
\node (d0) at (-1, 2) {};
\node (d1) at (-1, 1) {};

\node[draw=none, fill=none, below] at (a0) {$v_{0,(0,0)}$};
\node[draw=none, fill=none, below] at (a1) {$v_{1,(0,0)}$};
\node[draw=none, fill=none, right] at (b0) {\;\;$v_{0,(0,1)}$};
\node[draw=none, fill=none, right] at (b1) {\;\;$v_{1,(0,1)}$};
\node[draw=none, fill=none, above] at (c0) {$v_{0,(1,0)}$};
\node[draw=none, fill=none, above] at (c1) {$v_{1,(1,0)}$};
\node[draw=none, fill=none, left]  at (d0) {$v_{0,(1,1)}$\;\;};
\node[draw=none, fill=none, left]  at (d1) {$v_{1,(1,1)}$\;\;};

\draw (a0) -- (b0) -- (c0) -- (d0) -- (a0);
\draw (a0) -- (c0);
\draw (d0) -- (b0);
\draw (a0) -- (b0) -- (c0) -- (d1) -- (a0);
\draw (a0) -- (c0);
\draw (d1) -- (b0);
\draw (a0) -- (b0) -- (c1) -- (d1) -- (a0);
\draw (a0) -- (c1);
\draw (d1) -- (b0);
\draw (a0) -- (b1) -- (c0) -- (d1) -- (a0);
\draw (a0) -- (c0);
\draw (d1) -- (b1);
\draw (a0) -- (b1) -- (c1) -- (d1) -- (a0);
\draw (a0) -- (c1);
\draw (d1) -- (b1);
\draw (a1) -- (b1) -- (c1) -- (d1) -- (a1);
\draw (a1) -- (c1);
\draw (d1) -- (b1);
\end{tikzpicture} \vspace*{-2.mm}
\caption{Full graph}
\end{subfigure}
\hfill
\begin{subfigure}[t]{0.49\textwidth}
\centering
\begin{tikzpicture}[scale=1.2, every node/.style={draw, circle, fill=black, inner sep=1.5pt}, baseline=(current bounding box.north)]
\node (a0) at (0, 0) {};
\node (a1) at (1, 0) {};
\node (b0) at (2, 1) {};
\node (b1) at (2, 2) {};
\node (c0) at (1, 3) {};
\node (c1) at (0, 3) {};
\node (d0) at (-1, 2) {};
\node (d1) at (-1, 1) {};

\node[draw=none, fill=none] at (-0.45,-0.35) {$\{nt,rc,1c,2c,ec\}$};
\node[draw=none, fill=none] at (1.25,-0.35) {$\{at\}$};

\node[draw=none, fill=none] at (2.85,1) {$\{nt,rc,1c\}$};
\node[draw=none, fill=none] at (2.85,2) {$\{2c,ec,at\}$};

\node[draw=none, fill=none] at (1.25,3.35) {$\{nt,rc,2c\}$};
\node[draw=none, fill=none] at (-0.25,3.35) {$\{1c,ec,at\}$};

\node[draw=none, fill=none] at (-1.5,2) {$\{nt\}$};
\node[draw=none, fill=none] at (-1.8,0.8) {$\begin{array}{c}
    \{rc,1c,2c  \\
     \;\;\;\;\;\;\;\; ec,at\}
\end{array}$};

\draw (a0) -- (b0) -- (c0) -- (d0) -- (a0);
\draw (a0) -- (c0);
\draw (d0) -- (b0);
\draw (a0) -- (b0) -- (c0) -- (d1) -- (a0);
\draw (a0) -- (c0);
\draw (d1) -- (b0);
\draw (a0) -- (b0) -- (c1) -- (d1) -- (a0);
\draw (a0) -- (c1);
\draw (d1) -- (b0);
\draw (a0) -- (b1) -- (c0) -- (d1) -- (a0);
\draw (a0) -- (c0);
\draw (d1) -- (b1);
\draw (a0) -- (b1) -- (c1) -- (d1) -- (a0);
\draw (a0) -- (c1);
\draw (d1) -- (b1);
\draw (a1) -- (b1) -- (c1) -- (d1) -- (a1);
\draw (a1) -- (c1);
\draw (d1) -- (b1); 
\end{tikzpicture}

\vspace{-9mm}
\caption{A set representation}
\end{subfigure}
\caption{Potential Response Graph for Example \ref{ex:partial_monotonicity}.}
\label{fig:graph_partial_monotonicity}
\floatfoot{\textbf{Notes:} A vertex labeled $v_{d, z}$ corresponds to $D = d \,|\,Z = z$. Panel (a) shows the potential response graph. Panel (b) shows its set representation, where each vertex is represented by $S(r,z)$.}
\end{figure}
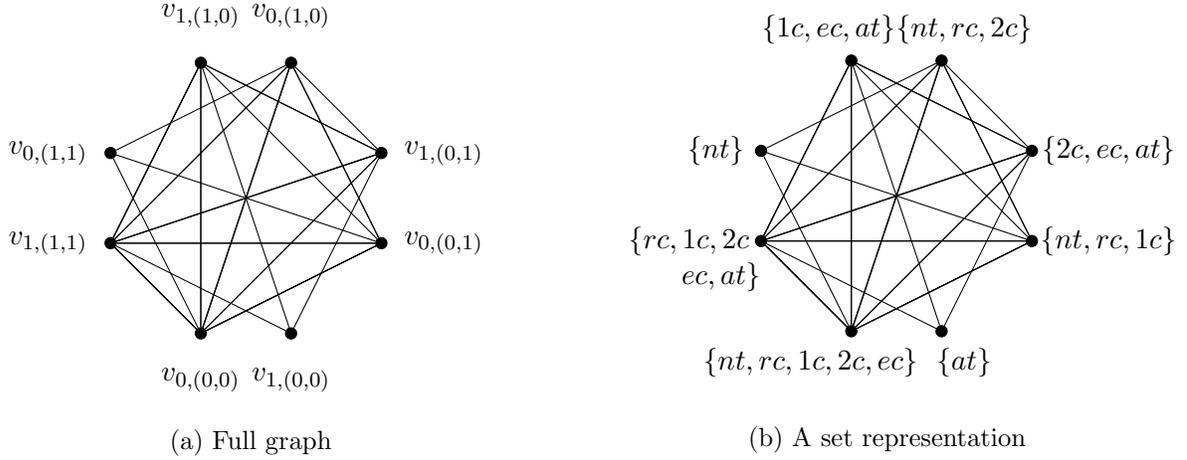

\begin{figure}[t]
\begin{center}
\includegraphics[scale=.9]{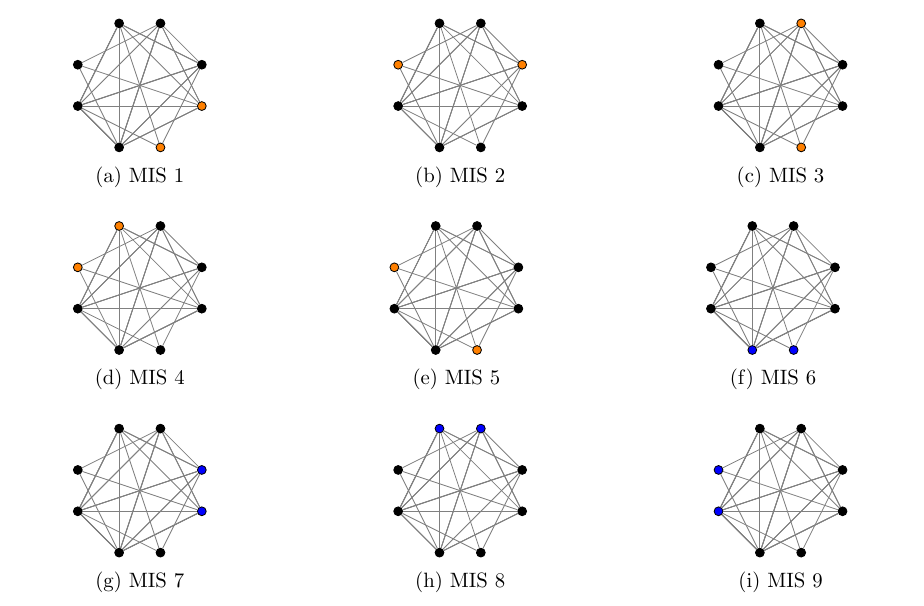}
\end{center}
\caption{Maximal Independent Sets (MISs) of $G$ in Example \ref{ex:partial_monotonicity}.}
\label{fig:pm_MIS}
\floatfoot{\textbf{Notes:}
The maximal independent sets in Panels (a)–(e) provide testable implications. The independent sets $V_{(0,0)},V_{(0,1)},V_{(1,0)},V_{(1,1)}$ in Panels (f)–(i) form a four-way partition of $V_G$.}
\end{figure}

\section{Testable Implications for Discrete Outcomes}
\subsection{Maximal Independent Set Inequalities}

This section presents the first main result: a testable implication of
support restrictions for discrete outcomes. We begin with an intuitive
argument based on Example~\ref{ex:partial_monotonicity}. Consider the set
$\{v_{1,(0,0)}, v_{0,(1,0)}\}$. These two vertices are not adjacent because
no latent type is compatible with both treatment responses: $D = 1$ when
$Z = (0,0)$ and $D = 0$ when $Z = (1,0)$. The former response is
consistent only with type $at$, whereas the latter is consistent with types
$nt$, $rc$, or $2c$; see Figure~\ref{fig:graph_partial_monotonicity}-(b).\footnote{Another equivalent way to describe this separation is to represent $D(z)$ by a structural random coefficient model and consider the sets of latent variable values consistent with either of the events; See Appendix \ref{ssec:pm_random_coeff}.} Since these two responses cannot occur jointly under
the modeling assumptions, it must be that
\[
P(D=1 \mid Z=(0,0)) + P(D=0 \mid Z=(1,0)) \leq 1.
\]
The potential response graph provides a systematic way to identify all
combinations of mutually exclusive sets of
latent types that cannot co-occur. Such collections correspond
to independent sets in the graph, with maximal independent sets generating
the tightest inequalities. Figure~\ref{fig:pm_MIS} lists all maximal
independent sets of $G$ in Example~\ref{ex:partial_monotonicity}. Theorem \ref{thm:testable_implication} below formalizes this intuition.

\begin{theorem}\label{thm:testable_implication}
Suppose Assumptions \ref{A:unconf}--\ref{A:support} hold. Let $G$ be the potential response graph associated with $\cR^*$.
Let $\mathcal I_G$ be the collection of all maximal independent sets of $G$. Then,
\begin{align} 
\sum_{v_{r, z}\in I}P(R=r\,|\,Z=z)\leq 1,\;\; \forall I\in \mathcal I_G \label{eq:MIS_inequalities}
\end{align} 
 is necessary for the existence of the distribution $Q$ of $R^*$ supported on $\cR^*$ that induces $P$.
\end{theorem}
Theorem \ref{thm:testable_implication} applies to \emph{any} support-type restriction as in Assumption \ref{A:support}, which covers most restrictions considered in the literature. One can use well-established algorithms to find all maximal independent sets of a graph (see Section \ref{ssec:computation}) and compute the inequalities above. Hence, the graph is not merely illustrative. It provides a model representation in which latent types induce $K$-cliques and testable restrictions correspond to independent sets, yielding an algorithmic route from economic assumptions to moment inequalities.

 Applying Theorem \ref{thm:testable_implication} to Example \ref{ex:partial_monotonicity} gives the following set of testable implications:
\begin{align}
&~~P(D=1|Z=(0,0))+P(D=0|Z=(0,1))\leq 1\label{eq:pm_ineq1}\\
&~~P(D=1|Z=(0,1))+P(D=0|Z=(1,1))\leq 1\\
&~~P(D=1|Z=(0,0))+P(D=0|Z=(1,0))\leq 1\\
&~~P(D=1|Z=(1,0))+P(D=0|Z=(1,1))\leq 1\\
&~~P(D=1|Z=(0,0))+P(D=0|Z=(1,1))\leq 1\label{eq:pm_ineq5}.
\end{align}
These restrictions correspond to MISs 1--5 in Figure \ref{fig:pm_MIS} Panels (a)--(e).\footnote{The independent sets, which form a $K$-partition of $V_G$, also provide restrictions, but these are trivial because each restriction is based on a single conditioning event of $Z$.  For example, MIS6 in Panel (f) of Figure \ref{fig:pm_MIS} implies that $P(D=1|Z=(0,0))+P(D=0|Z=(0,0))\leq 1$, which is trivial by the definition of a conditional probability.}
For this example, the inequalities above can be shown to be \emph{sharp}, meaning that they are also sufficient for the existence of the joint distribution of $D^*$ satisfying the partial monotonicity assumption.

More generally, in settings with discrete responses, Theorem \ref{thm:testable_implication} recovers the inequalities derived in \citet{Kitagawa15} and \citet{BaiHuangTabordMeehan2026}, and provides a simple characterization of the sharp testable implications for the setting in \citet{Kwon:2024aa}.\footnote{Continuous outcomes are discussed in Section \ref{ssec:continuous_outcome}.} Natural questions are ``Is \eqref{eq:MIS_inequalities} always sharp? If not, what are the sharp testable implications?'' The next section addresses these questions.

\subsection{Sharpness}\label{ssec:sharpness}

The potential response graph can be constructed for any support restriction $\mathcal{R}^*$, and the testable implications obtained in Theorem \ref{thm:testable_implication} are always valid. In this section, we show that these testable implications are \emph{sharp}, that is, necessary and sufficient for the model to be consistent with the data, if the support restriction satisfies two additional requirements.
\begin{assumption}[Regular Support Restriction] \hspace{1mm} \label{A:support_regularity}
\begin{enumerate}
	\item[(i)] 	For any collection $\{(r_k, z_k)\}_{k \in \mathcal{K}} \subseteq \mathcal{R} \times \mathcal{Z}$, 
	\[
	S(r_k, z_k) \cap S(r_{k'}, z_{k'}) \ne \varnothing \text{ for all } k, k' \in \mathcal{K} \;\;\; \implies \;\;\;  \bigcap_{k \in \mathcal{K}} S(r_k, z_k) \ne \varnothing.
	\] 
	\item[(ii)] The potential response graph $G$ is perfect.
\end{enumerate}
\end{assumption}
The first condition ensures that the graph-based representation of the model without loss of information. By construction, each response type $r^* \in \mathcal{R}^*$ corresponds to a maximal clique $C = \{v_{r, z} \in V_G: T(r^*, z) = r\}$ in $G$. However, the combination of edges in $G$ may produce additional maximal cliques corresponding to response types outside of $\mathcal{R}^*$. Assumption \ref{A:support_regularity} (i) requires that any set of vertices that are pairwise compatible must also be jointly compatible, thus ensuring that all maximal cliques of $G$ correspond to the response types in $\mathcal{R}^*$. Under Assumption \ref{A:support_regularity} (i), we can associate the support restriction $\mathcal{R}^*$ one-to-one with the set of all maximal cliques in $G$, denoted $\mathcal{C}_G$. Surprisingly, Assumption \ref{A:support_regularity} (i) by itself is not sufficient for sharpness of the testable implications in Theorem 1. 

Assumption \ref{A:support_regularity} (ii) requires that no induced subgraph of $G$ is a cycle with odd length at least five, or the complement of one.\footnote{ Let $\chi(G)$ denote the chromatic number of $G$, i.e., the least number of colors required to color its vertices so that no two connected vertices share the same color. Let $\omega(G)$ denote the size of the largest clique (fully connected set of vertices) in $G$. In general, $\chi(G) \geq \omega(G)$.  By definition, a graph is perfect if for each induced subgraph $G'$ of $G$, we have $\chi(G') = \omega(G')$. The characterization of perfect graphs in terms of odd holes and antiholes is equivalent to this characterization by the Strong Perfect Graph Theorem of \cite{chudnovsky2006strong}.} To see why this condition  is necessary, note that, given Assumption \ref{A:support_regularity} (i), any model-compatible probability vector $\beta_0(P)$ can be written as $\beta_0(P) = \sum_{C \in \mathcal{C}_G}Q(C)\bm{1}_C$, for some $Q$ such that $Q(C) \geq 0$, $\sum_{C \in \mathcal{C}_G} Q(C) = 1$. Thus, for any $V \subseteq V_G$, 
\begin{equation} \label{E:level_k}
\beta_0(P)'\bm{1}_V = \sum_{C \in \mathcal{C}_G}Q(C)|V \cap C| \leq \ell_G(V),
\end{equation}
where $\ell_G(V)=\max_{C \in \mathcal{C}_G}|C \cap V|$, with $C_G$ denoting the collection of all maximal cliques in $G$.  Under Assumption \ref{A:support_regularity} (i), $\ell_G(V)$ represents the maximum number of vertices in $V$ that can be ``hit'' by a single latent type represented by a maximal clique $C$.

 Now, suppose $G[V]$ is an odd cycle of length at least five, also known as an \emph{odd hole}. No latent type $r^*$ can hit three vertices in $V$: any such three vertices would be pairwise adjacent and hence form a triangle, which an odd hole does not contain. On the other hand, each edge of the cycle corresponds to a pair of vertices that a single type can generate. Hence, the maximum number of vertices in $V$ hit by a single type is exactly two (so $\max_{C \in \mathcal{C}_G}|C \cap V|=2$), implying that \eqref{E:level_k} takes the form
\[
\sum_{v_{r, z} \in V} P(R = r\,|\, Z = z) \leq 2,
\]
where the left-hand side involves a sum of conditional probabilities for at least three distinct values of $z$. This inequality non-trivially restricts $\beta_0(P)$ beyond what is implied by the partial sum equalities $\beta_0(P)'\bm{1}_{V_k} = 1$ and MIS inequalities $\beta_0(P)'\bm{1}_I \leq 1$. Graphically, since covering the vertices of an odd cycle with independent sets requires at least three independent sets, the MIS inequalities would imply $\sum_{v_{r, z} \in V} P(R = r\,|\, Z = z) \leq 3,$ which gives a loose bound. The same conclusion applies to any subset $V \subseteq V_G$ containing a further subset $V'$ corresponding to an odd cycle or its complement.  Thus, without Condition (ii), there are many additional testable implications beyond those in Theorem \ref{thm:testable_implication}. Luckily, many existing models satisfy this condition, as we will show below.

The following is the paper's second main result, which shows that the MIS inequalities are sharp for regular support restrictions.
\begin{theorem}\label{thm:sharpness_of_mis}
Suppose Assumptions \ref{A:unconf}--\ref{A:support_regularity} hold. Then, \eqref{eq:MIS_inequalities}
 is necessary and sufficient for the existence of the joint distribution $Q$ of $R^*$ supported on $\cR^*$ that induces $P$.
\end{theorem}

\subsection{Discussion} \label{subsec:discussion_main}
Several comments are in order.  First, the framework easily accommodates additional control variables. Let $W \in \mathcal{W}$ denote a vector of observed covariates such that $R^* \perp Z \mid W$. This conditional independence assumption leaves the support restriction $\mathcal{R}^*$, and hence the independent sets of the potential response graph $G$, unchanged.
Repeating the argument of Theorems \ref{thm:testable_implication}--\ref{thm:sharpness_of_mis}, and conditioning on $W = w$, we obtain a richer set of testable implications:\footnote{If $\cR^*$ varies with $w$,  let $G_w$ represent the support restriction for a given $w$. Then, $\mathcal I_G$ in \eqref{eq:MIS_cond_inequalities} can be replaced with $\mathcal I_{G_w}$. Such a restriction appears in \cite{heiler2024treatmentevaluationintensiveextensive}.}
\begin{align} 
\sum_{v_{r, z}\in I}P(R=r\,|\,Z=z, W = w)\leq 1,\;\; \forall I\in \mathcal I_G, \; \forall_{a.s.} w \in \mathcal{W}. \label{eq:MIS_cond_inequalities}
\end{align} 
When $W$ includes continuous variables, these inequalities provide a computationally tractable way to characterize sharp testable implications of the modeling assumptions. The implementation of corresponding testing procedures is discussed in Section \ref{ssec:inference_cont}.

Second, Theorem \ref{thm:sharpness_of_mis} has a geometric interpretation. Denote $K = |\mathcal{Z}|$, $N = \sum_{z \in \mathcal{Z}}|\mathcal{R}_z|$, and $M = |\mathcal{R}^*|$. The support restriction $\mathcal{R}^*$ can be represented as a binary matrix $A^* \in \{0, 1\}^{N \times M}$, which we call a \textit{support matrix}. The rows of $A^*$ are indexed by pairs $\{(r, z): r \in \mathcal{R}_z, z \in \mathcal{Z}\}$, the columns are indexed by support points $r^* \in \mathcal{R}^*$, and the entries $\bm{1}(r = T(r^*,z))$ indicate whether a value of the observables $(r,z)$ is consistent with a support point $r^*$. Any joint distribution $Q$ on $\mathcal{R}^*$ induces a vector $b = A^*Q$. Such vectors form a convex polytope,
\begin{equation} \label{E:P_bar}
\overline{\mathbb{P}} = \left\{b \in \mathbf{R}^{N}: b = A^*Q, \text{ for some } Q \in \Delta(\mathcal{R}^*) \right\},
\end{equation}
of all probability vectors consistent with the modeling assumptions. Further, for each $k \in \{1, \dots, K\}$, let $P_k =(P(R = r\,|\,Z = z_k))_{r \in \mathcal{R}_{z_k}}$ denote the observed conditional distributions, so that $\beta_0(P) = (P_1', \dots, P_K')'$. The model is consistent with the data whenever
\[
\beta_0(P) \in \overline{\mathbb{P}}.
\]
The polytope $\overline{\mathbb{P}}$ is expressed in the vertex  ($V$) representation, as the convex hull of the columns of $A^*$. For regular support restrictions $\mathcal{R}^*$, Theorem \ref{thm:sharpness_of_mis} provides the corresponding half-space ($H$) representation,
\[
\overline{\mathbb{P}} = \{b \in \mathbf{R}^N_{+}: b'\bm{1}_{V_k} = 1, \text{ for all } k = 1, \dots, K,\; b'\bm{1}_I \leq 1, \text{ for all } I \in \mathcal{I}_G\}.
\]
This polytope duality result may be of independent interest.

Third, the question we consider can also be stated as follows: When does there exist a joint distribution $Q$ with a given collection of marginals $P_1, \dots, P_K$ and a restricted support $\mathcal{R}^*$? The case of $K = 2$ has been resolved in \citet{artstein1983distributions}. In Theorem \ref{T:main} in the Appendix, we state an extension to $K > 2$ for finite and regular restrictions $\mathcal{R}^*$, which is formally equivalent to Theorem \ref{thm:sharpness_of_mis}. We also give examples of support restrictions such that $\overline{\mathbb{P}}$ is a strict subset of a polytope characterized by \eqref{eq:MIS_inequalities} when either one of the two conditions in Assumption \ref{A:support_regularity} is violated. 
Theorem \ref{T:main} may be of interest, e.g., in multi-marginal optimal transport problems; see \citet{pass2015multi}. We discuss extensions to continuous settings in Section \ref{sec:extensions}.

Fourth, we have found that most support restrictions considered in the literature are regular. A notable exception is the IV model with multi-valued instruments ($K\geq 3$), where only exclusion is assumed: $Y(d,z)=Y(d,z'), \forall (d,z,z')\in \cD\times\cZ^2$, without any additional shape restrictions.
This support restriction is non-regular because one can show its potential response graph $G$ is imperfect.\footnote{Interestingly, adding a monotonicity assumption to the exclusion restriction leads to a perfect potential response graph. Hence, the joint support restriction of exclusion and monotonicity is regular. See Section 
\ref{ssec:testing_exclusion_monotonicity}.} Hence, while regularity holds for many structured restrictions (e.g., monotonicity/exposure maps), it may fail in some economically important models. 
In those cases, additional inequalities beyond the MIS inequalities are needed to ensure sharpness.

Indeed, for the example above with binary $Y$, one can show that the inequalities in \eqref{E:level_k} with $\ell_G(V)=1$ recover those in \cite{Pearl1995Testability}, and the inequalities with $\ell_G(V)=2$ and $\ell_G(V)=3$  recover the additional sharp inequalities found by \cite{KedagniMourifie2020} (see Appendix \ref{ssec:instrumental_inequalities}).

We conclude the discussion here with two remarks pointing out connections to the literature on partial identification using random sets \citep[see, e.g.,][]{galichon2011set, chesher2017generalized, molchanov2018random}.

 \begin{remark}[Binary Instrument and Artstein's Theorem] \label{rem:artstein}
	If $Z \in \{0, 1\}$, sharp testable implications of any support restriction $\mathcal{R}^*$ follow directly from Artstein's Theorem  \citep{artstein1983distributions}. To elaborate, let $\mathcal{R}_{z}$ denote the support of $R$, conditional on $Z = z$, for $z \in \{0, 1\}$, and $\mathcal{R}^* \subseteq \mathcal{R}_0 \times \mathcal{R}_1$ be any support restriction. Let $P_0(\cdot) = P(R \in \cdot \,|\, Z = 0)$ and $P_1(\cdot) = P(R \in \cdot\,|\,Z = 1)$ denote the observed conditional distributions. Define a correspondence $F: \mathcal{R}_0 \rightrightarrows \mathcal{R}_1$ as $F(r_0) = \{r_1 \in \mathcal{R}_1: (r_0, r_1) \in \mathcal{R}^*\}$, and assume, for simplicity, that $F(r_0)$ is compact-valued.\footnote{For a more general statement, see, e.g., Theorem 2.13 in \citet{molchanov2018random}. } Then, by Theorem 3.1 in \citet{artstein1983distributions}, a joint distribution $Q \in \Delta(\mathcal{R}_0 \times \mathcal{R}_1)$ supported on $\mathcal{R}^*$ with marginals $P_0$ and $P_1$ exists if and only if
	\[
	P_1(A) \leq P_0(\{r_0 \in \mathcal{R}_0: F(r_0) \cap A \ne \varnothing\}), 
	\]
	for all compact sets $A \subseteq \mathcal{R}_1$. Re-arranging and returning to the familiar notation,
	\begin{equation} \label{E:artstein}
	P(R \in A \,|\,Z = 1) + P(F(R) \subseteq A^c\,|\,Z = 0) \leq 1.
	\end{equation}

	 To illuminate the connection with our Theorem \ref{thm:testable_implication}, suppose $\mathcal{R}$ is finite and consider any support restriction $\mathcal{R}^*$. In the corresponding bi-partite graph $G$, each link corresponds to a maximal clique, so no maximal cliques can be missing from $\mathcal{C}_G$. Moreover, any cycle in $G$ must have even length, meaning that $G$ is perfect. Thus, any $\mathcal{R}^*$ is regular, and it suffices to consider the inequalities in \eqref{eq:MIS_inequalities} corresponding to maximal independent sets. Such sets are precisely $(A, \{r_0 \in \mathcal{R}_0: F(r_0) \in A^c\})$, for all $A \subseteq \mathcal{R}_1$.
     
	With some algebra, sharp testable implications derived in \citet{Kitagawa15} and those implied by \citet{Kwon:2024aa} follow directly from Artstein's theorem via \eqref{E:artstein} (see Appendix \ref{ssec:instrumental_inequalities}).  Our results extend this characterization to multi-marginal setting. \qed
\end{remark}

\begin{remark}[Redundant Inequalities] 
Some of the inequalities derived in Theorem \ref{thm:testable_implication} may be redundant. For binary $Z$, the smallest set of non-redundant inequalities is characterized in \citet{luo_ponomarev_wan}. 
	Developing similar arguments in the present setting is more challenging. A simple implication is as follows. 
	Let $I^*, I_1, \dots, I_S \in \mathcal{I}_G$ be such that 
\[
\sum_{s = 1}^S \bm{1}_{I_s} = \sum_{k = 1}^{S-1} \bm{1}_{V_k^*} + \bm{1}_{I^*},
\] 
for some $V_1^*, \dots, V_{S-1}^*$. Then, $I^*$ is redundant given $I_1, \dots, I_S$, since  
\[
\sum_{s = 1}^S P'\bm{1}_{I_s} \leq S \iff S - 1 + P'\bm{1}_{I^*}  \leq S \iff  P'\bm{1}_{I^*} \leq 1.
\]
 For example, in Figure \ref{fig:pm_MIS}, the inequality for MIS 5 in Panel (e)  is redundant, given MIS 3 and MIS 4 in Panels (c) and (d).  On the other hand, if, in each part $V_k$, there is a vertex that does not belong to any $I \in \mathcal{I}_G$, such implications are impossible. This is the case, for example, in the encouragement design model studied in \citet{BaiHuangTabordMeehan2026}. \qed 
\end{remark}

\subsection{Regular Support Restrictions}\label{ssec:support_regularity}

Assumption \ref{A:support_regularity} is satisfied by most support restrictions considered in the literature, and can often be verified analytically. For example, if $Z$ is binary, all support restrictions are regular because the potential response graph is bipartite. Then, Assumption \ref{A:support_regularity} (i) holds trivially, because each support point (maximal clique) corresponds to a single link in a graph, and Assumption \ref{A:support_regularity} (ii) holds because bipartite graphs can only have even cycles. Below, we provide general sufficient conditions when $Z$ is multi-valued.

Consider Assumption \ref{A:support_regularity} (i).  For a vector $x = (x_1, \dots, x_d) \in \mathbf{R}^d$ and a subset of coordinates $S \subseteq \{1, \dots, d\}$, let $x[S] = (x_i)_{i \in S}$ denote the corresponding subvector. 

\begin{proposition} \label{L:helly_simple}
	Suppose Equation \eqref{eq:obs} takes the form $R = R^*[T(z)]$, for a set $T(z) \subseteq \{1, \dots, \dim(R^*)\}$. Suppose the support restriction $\mathcal{R}^*$ can be defined by pairwise linear constraints of the form $ a r^*_{i} + br^*_{j} \leq c$, for arbitrary coefficients $a, b, c \in \mathbf{R}$.  Then, Assumption \ref{A:support_regularity} (i) holds. 
\end{proposition}
This result covers all exclusion and partial-order restrictions, as we illustrate below. 

\setcounter{example}{2}

\begin{example}[Mediation (continued)]
	Let $D \in \mathcal{D}$ denote treatment, $M_j \in \mathcal{M}_j$, for $j \in \{1, \dots, J\}$, denote mediators, and $Y_{\ell} \in \mathcal{Y}_{\ell}$, for $\ell \in \{1, \dots, L\}$, denote outcomes. Let $M = (M_1, \dots, M_J) \in \mathcal{M}$ and $Y = (Y_1, \dots, Y_L) \in \mathcal{Y}$. Suppose that the effect of the treatment on each outcome operates through potentially different channels. That is, $	Y_{\ell}((m_j)_{j \in 1, \dots, J}) = Y_{\ell}((m_j)_{j \in S_\ell}),$ for some $S_{\ell} \subseteq \{1, \dots, J\}$, for all $\ell \in \{1, \dots, L\}$. For each $\ell$, let $\mathcal{P}_{\ell} = \{(m, m') \in \mathcal{M}^2: m_j = m'_j, \forall j \in S_{\ell} \}$ denote the list of pairs of mechanism assignments that are equivalent in terms of outcome $\ell$. Then,
	\[
	\mathcal{Y}^* = \{Y^* \in \mathcal{Y}^{\mathcal{M}}: Y_{\ell}(m) = Y_{\ell}(m'), \forall (m, m') \in \mathcal{P}_{\ell}, \forall \ell \in \{1, \dots, L\}\}.
	\]
	Further, the map $T$ takes the form $(y_1, \dots, y_L, m_1,\dots,m_J) = (Y_1( \{M_j(d)\}_{j \in S_1}, \dots Y_{L}(\{M_j(d)\}_{j \in S_L},$ $M_1(d),\dots,M_J(d)))$. \qed
\end{example}

\begin{example}[Interference (continued)]
	Let $R = (Y_1, Y_2, \dots, Y_N) \in \mathcal{R}$ and $Z = (Z_1, Z_2, \dots, Z_N) \in \mathcal{Z}$ denote the observed outcomes and treatments of $N$ individuals. Let $Y_i(z)$ denote the potential outcome of individual $i$, which may depend on treatments received by all of the individuals, and $R^* = (Y_1(z),\dots, Y_N(z))_{z \in \mathcal{Z}} \in \mathcal{R}^{|\mathcal{Z}|}$. Suppose we want to test the restriction that for pre-specified non-random functions $d_i: \mathcal{Z} \to \mathcal{D}$, \eqref{eq:exposure} holds.

    Let $\mathcal{P}_i = \{(z, z') \in \mathcal{Z} \times \mathcal{Z}: d_i(z) = d_i(z')\}$ denote the collection of all pairs of treatment assignments that lead to the same exposure for individual $i$. Then,  
	\[
	\mathcal{R}^* = \{R^* \in \mathcal{Y}^{|\mathcal{Z}|}: Y_i(z) = Y_i(z'), \;\forall (z, z') \in \mathcal{P}_i,  \;\forall i \in \{1, 2\}\}.
	\]
	Further, the map $T$ takes the form $	(y_1, y_2, \dots, y_N) = (Y_1(z), \dots, Y_N(z)),$ i.e., selects a subset of coordinates of $R^*$ corresponding to the treatment assignment $z$. Thus, the assumptions of Proposition \ref{L:helly_simple} are satisfied, and Assumption \ref{A:support_regularity} (i) holds. \qed 
\end{example}

Now we turn to Assumption \ref{A:support_regularity} (ii). There are many known classes of perfect graphs \citep{hougardy06}. For example, all bipartite graphs are perfect because neither a bipartite graph nor its complement contains an induced odd cycle of length at least five. As a result, with a binary input variable, all support restrictions are regular. Another well-known class of perfect graphs, \emph{comparability graphs}, is described as follows.

\begin{proposition}\label{prop:comparability}
	Suppose there exists a strict partial order $\succ$ on $V_G$ such that $(v, u) \in E_G$ if and only if $v \prec u$ or $u \prec v$. Then, $G$ is perfect. 
\end{proposition}

In practice, this connection can be verified by starting with the rule according to which the graph is constructed and verifying that it defines a valid strict partial order. The latter requires checking irreflexivity, i.e., $v \not \prec v$, and transitivity, i.e., that $v \prec u$, $u \prec w$ implies $v \prec w$.

\setcounter{example}{0}
\begin{example}[Instrumental Variables (continued)]
	Let $\mathcal{D}$ and $\mathcal{Z}$ be totally ordered. Suppose $Y(d, z) = Y(d, z')$ for all $d \in \mathcal{D}$ and $z, z' \in \mathcal{Z}$, and $D(z) \succeq D(z')$ whenever $z \succeq z'$. Let $v = (y, d, z)$ and $v' = (y', d', z')$. Define a binary relation
	\[
	v \prec v' \iff (z \prec z') \text{ and } \{d \prec d' \text{ or } \{d = d' \text{ and } y = y'\} \}
	\]
	Then, $v \prec v$ is impossible since $z \prec z$ is. Suppose $v \prec v'$ and $v' \prec v''$. Then, $z \prec z''$ and $d \preceq d''$. If $d = d''$, then it must be that $d = d' = d''$ and $y = y' = y''$, and thus $v \prec v''$. If $d \prec d''$, then $v \prec v''$. Thus, Proposition \ref{prop:comparability} applies, and $G$ is perfect. \qed 
\end{example}

\setcounter{example}{3}
\begin{example}[Intereference (continued)]
Consider the exposure map restriction \eqref{eq:exposure}. For this example, suppose all individuals share a common exposure map, $(d_i(z)=d(z))$ for every $i$, and let $y$ denote the joint outcome vector. Under this restriction, two vertices $v_{y,z}$ and $v_{y',z'}$ are then adjacent if and only if $z\neq z'$ and either $d(z)\neq d(z')$ or $y=y'$. 
Choose an arbitrary bijection $\rho:\mathcal D\to\{1,\ldots,|\mathcal D|\}$ and equip $\mathcal Z$ with a strict total order $<$.\footnote{The particular choices of $\rho$ and $<$ do not affect the conclusion.} Define a binary relation
\begin{equation}
\begin{split}
v_{y,z}\prec v_{y',z'}\quad\Longleftrightarrow\quad
&\rho(d(z))<\rho(d(z'))\\
&\text{or }\bigl[d(z)=d(z'),\ y=y',\ z<z'\bigr].
\end{split}\label{eq:exposure_strictorder}
\end{equation}
This relation is irreflexive. If $v_{y,z}\prec v_{y',z'}\prec v_{y'',z''}$, exposure ranks weakly increase. If $
\rho(d(z))<\rho(d(z'')),
$ the first statement on the right hand side of \eqref{eq:exposure_strictorder} gives $v_{y,z}\prec v_{y'',z''}$. Otherwise, all three exposures and outcomes coincide, and transitivity follows from $<$. Hence, $\prec$ defines a strict partial order, Proposition \ref{prop:comparability} applies, and $G$ is perfect. For the semimonotonicity restriction in \eqref{eq:semimonotone}, the comparability remains to hold if we replace the first statement in \eqref{eq:exposure_strictorder} by $d(z)<d(z')$ and $y\ge y'$, assuming that $\mathcal D$ is totally ordered. \qed
\end{example}

Another useful subclass of perfect graphs is the class of \emph{chordal graphs}, which have no induced cycles of length at least four. As with comparability graphs, chordality can be established through economically interpretable assumptions.

\begin{proposition}[Latent types with a tree structure]\label{prop:tree}
Suppose there is a tree $\mathcal T$ with node set $\mathcal R^*$ such that $S(r,z)$ induces a connected subtree for every $v_{r,z}\in V_G$. Then, $G$ is chordal and hence perfect. 
\end{proposition}

\begin{remark}
Every clique of $G$ given above extends to an admissible latent type. To see this, finite families of subtrees of a tree have the Helly property. Applying this to a clique $C\subseteq V_G$ gives a type in $\bigcap_{v_{r,z}\in C}S(r,z)$.
\end{remark}

\begin{example}[Ordered Choice]\label{ex:ordered-choice}
Let $\mathcal Z=\{z_1,\ldots,z_K\}$ and $\mathcal Y=\{1,\ldots,J\}$. Suppose responses are generated by
\[
Y(z)=j\quad\Leftrightarrow\quad
c_{j-1}(z)<U\le c_j(z),
\qquad j=1,\ldots,J,
\]
where $U\in \mathbb R$ represents unobserved heterogeneity, and $-\infty=c_0(z)<c_1(z)<\cdots<c_J(z)=+\infty$ for every $z$. A special case is the index model, where $Y(z)=j\Leftrightarrow
\kappa_{j-1}<m(z)+U\le\kappa_j$
with fixed thresholds $\kappa_j$.

For a simple illustration, let $J=3$, and $\mathcal Z=\{z_1,z_2\}$, with thresholds
\[c_1(z_1)=0,~c_2(z_1)=2,~c_1(z_2)=1,~c_2(z_2)=3.\]
Pooling the thresholds across $z\in\mathcal Z$ gives the ordered boundaries $0,1,2,3$, which divide the values of $U$ into five cells:
\[
\begin{array}{c|cc|c}
\text{Cell of }u&R(z_1)&R(z_2)&r^*\\ \hline
(-\infty,0]&1&1&(1,1)\\
(0,1]&2&1&(2,1)\\
(1,2]&2&2&(2,2)\\
(2,3]&3&2&(3,2)\\
(3,\infty)&3&3&(3,3)
\end{array}
\]
Using the order of the cells, we may order the latent types along a path:
\[\cR^*:\qquad (1,1)\;-\;
\underbrace{(2,1)\;-\;(2,2)}_{S(2, \,z_1)}
\;-\;(3,2)\;-\;(3,3),\]
where each $S(r,z)$, e.g., $S(2,z_1)=\{(2,1),(2,2)\}$, forms a consecutive segment of this path. 

In general, let the pooled thresholds be $b_1<\cdots<b_L$, with
$b_0=-\infty$ and $b_{L+1}=+\infty$. Each cell
$(b_{\ell-1},b_\ell]$ corresponds to a distinct latent type.
Arrange these types along a path in increasing order of $u$.
For each $(j,z)$, the response $Y(z)=j$ occurs if and only if
$u\in(c_{j-1}(z),c_j(z)]$. Thus, if two types on the path generate
this response, every intervening type does as well. It follows
that $S(j,z)$ is a consecutive segment of the path. Since a path
is a tree, the model satisfies the conditions of
Proposition~\ref{prop:tree}. \qed 
\end{example}

\begin{example}[Targeted Subsidy]\label{ex:subsidy}
Let $|\mathcal Y|=J$, and let $s\in\mathcal Y$ be an alternative for which different levels of subsidy $\mathcal Z=\{z_1<\cdots<z_K\}$ are available. Individuals maximize expected net utility $a\mapsto V_{a}(z),a\in \mathcal Y$, where $V_s(z_k)$ is weakly increasing in $k$, while $V_b(z_k)=V_b$ for every $b\ne s$. The expected utility $V=(V_s(\cdot),s\in\mathcal Y)$ varies across individuals.

An empirical setting that motivates this restriction is \citeposs{cohen2015} study of subsidies for retail artemisinin-based combination therapy (ACT), a malaria treatment, in Kenya.\footnote{We focus on the ACT-subsidies and the no-subsidy control group. The paper also considers arms that combine ACT subsidies with subsidies for rapid diagnostic tests (RDTs), which are outside the scope of this example.} Individuals either purchase ACT at a drug shop ($s$), visit a health center ($h$), or do neither ($n$). Thus, $\mathcal Y=\{s,h,n\}$. They analyze a reduction in the retail ACT price holding other prices fixed, which raises $V_s$ while leaving $V_h$ and $V_n$ unchanged. Their experiment includes subsidy rates of $80\%$, $88\%$, and $92\%$, together with a no-subsidy group.\footnote{The expected value of action $s$ is
\[
V_s(z)=\pi[U_s^P(\pi)-p_s^P(\pi,z)]
+(1-\pi)[U_s^N(\pi)-p_s^N(\pi,z)],
\]
where $\pi$ is the perceived probability of malaria and $P,N$ denote malaria-positive and malaria-negative states. For each $a\in\{s,h,n\}$, $U_a^M$ is the state-specific gross utility and $p_a^M$ the expected treatment cost. The subsidy $z$ only affects $V_s(\pi,z)$ but not $V_a=\pi[U_a^P(\pi)-p_a^P(\pi)]+(1-\pi)[U_a^N(\pi)-p_a^N(\pi)$ for $a\in \{h,n\}$. Holding  $\pi$ and the net utilities from the other alternatives fixed, a subsidy that weakly reduces both costs makes $V_s(\pi,z)$ weakly increasing in $z$. }

Let $B\in\argmax_{a\in\{h,n\}}V_a$ be the preferred outside option
and $M=\max\{V_h,V_n\}$ its utility. The utility-maximizing choice is
$Y(z)=B \bm{1}\{V_s(z)\le M\}+s \bm{1}\{V_s(z)>M\}$.\footnote{We resolve ties
in favor of the outside option, but this convention is inessential.}
More generally, with outside options $\mathcal D\setminus\{s\}$,
an individual can switch only from their preferred outside option
to $s$. Let $r^*_s$ denote the type choosing $s$ for every
$z\in\mathcal Z$. Each remaining type $r^*_{b,k}$ is characterized
by its preferred outside option $b$ and the index $k$ of the
subsidy level above which it switches to $s$:
\[
r^*_{b,k}:\quad Y(z_i)=
\begin{cases}
b,&i\le k,\\
s,&i>k,
\end{cases}
\qquad b\ne s,\quad k=1,\ldots,K.
\]
In the model of \citet{cohen2015}, one can write $r^*_{h,2}=(Y(0),Y(0.8),Y(0.88),Y(0.92))=(h,h,s,s)$ for example.

For each $b$, connect the types by branches
\[
r^*_s\;-\;r^*_{b,1}\;-\;\cdots\;-\;r^*_{b,K},\qquad b\ne s.
\]
This forms a tree characterizing $\mathcal R^*$, with $r^*_s$ as its root (See Figure \ref{fig:act-type-tree}). 
The set representations of the vertices $v_{r,z}$ are
\[
S(b,z_i)=\{r^*_{b,k}:k\ge i\}\quad(b\ne s),
\qquad
S(s,z_i)=\{r^*_s\}\cup\bigcup_{b\ne s}\{r^*_{b,k}:k<i\},
\]
which form subtrees.
This model therefore satisfies the conditions of Proposition~\ref{prop:tree}. \qed
\end{example}

\begin{figure}[t]
\centering
\begin{minipage}{0.9\linewidth}
\centering
\begin{tikzpicture}[
  type/.style={draw,rounded corners,minimum width=1.6cm,minimum height=0.65cm,fill=white},
  htype/.style={type,draw=blue!65!black,fill=blue!5},
  ntype/.style={type,draw=teal!65!black,fill=teal!5},
  every node/.style={font=\small}]
\node[type,thick,fill=gray!12] (root) at (0,0) {$ssss$};
\node[above] at (0,0.45) {Root $r^*_s$: always choose $s$};
\node[htype] (h1) at (-2.5,-1.3) {$hsss$};
\node[htype] (h2) at (-2.5,-2.3) {$hhss$};
\node[htype] (h3) at (-2.5,-3.3) {$hhhs$};
\node[htype] (h4) at (-2.5,-4.3) {$hhhh$};
\node[ntype] (n1) at (2.5,-1.3) {$nsss$};
\node[ntype] (n2) at (2.5,-2.3) {$nnss$};
\node[ntype] (n3) at (2.5,-3.3) {$nnns$};
\node[ntype] (n4) at (2.5,-4.3) {$nnnn$};
\draw[thick,blue!65!black] (root)--(h1)--(h2)--(h3)--(h4);
\draw[thick,teal!65!black] (root)--(n1)--(n2)--(n3)--(n4);
\node[align=center,below] at (-2.5,-4.8) {Outside option $h$};
\node[align=center,below] at (2.5,-4.8) {Outside option $n$};
\end{tikzpicture}

\caption{Tree of latent types under an ACT subsidy.}
\label{fig:act-type-tree}
\smallskip
\small Note: The two branches share the single root $ssss$. Entries in each response vector follow subsidy levels $(0,0.80,0.88,0.92)$. 
\end{minipage}
\end{figure}

\begin{example}[Discrete Choice Models] \label{ex: discrete_choice}
	Let $z \in \mathcal{Z} =\{z_1, \dots, z_K\}$ and $r \in \mathcal{R} = \{0, 1, \dots, J\}$. Consider a discrete choice model with quasi-linear utility,
	\[
	R(v, z) = \argmax_{j \in \mathcal{R}} \left(v_j - c_j(z) \right),
	\]
	where  $c_j: \mathcal{Z} \to \mathbf{R}$, for $j \in \{0, \dots J\}$, and $v = (v_1, \dots, v_J) \in \mathbf{R}^J$. We normalize $v_0 = 0$ and $c_0(z) = 0$, for all $z \in \mathcal{Z}$, and assume that $v$ has full support. Denote
	\[
	S(r, z) = \{v \in \mathbf{R}^J: R(v, z) = r\}.
	\]
	For each $z$, the sets $\{S(r, z)\}_{r \in \mathcal{R}}$ partition $\mathbf{R}^{J}$. Intersections of sets $S(r, z)$ further partition $\mathbf{R}^J$ into the sets $V_1, \dots, V_M$, which satisfy the property
	\[
	v, v' \in V_m \iff R(v, z) = R(v', z), \text{ for all }z\in\mathcal{Z}.
	\]
    The sets $V_1, \dots, V_M$ form the Minimal Relevant Partition in the terminology of \citet{tebaldi2023nonparametric}. 
	
	With each set $V_m$ we can uniquely associate a response type $r^*_m = (R(v, z_1), \dots, R(v, z_K))$. Let $\mathcal{R}^* = \{r_1^*, \dots, r_M^*\}$ denote the collection of such response types. Notice that each $S(r, z)$ can be written as a union of some $V_m$, that is $S(r, z) = \bigcup_{m: \;R(v, z) = r} V_m$. Thus, we can view each $S(r, z)$ as a subset of $\mathcal{R}^*$, so that for each $z \in \mathcal{Z}$, the sets $\{S(r, z)\}_{r \in \mathcal{R}}$ partition $\mathcal{R}^*$. In general, the support restriction $\mathcal{R}^*$ may violate either of the conditions in Assumption \ref{A:support_regularity}, but in several important special cases, it satisfies both.
    
    For example, consider the binary choice model with $J = 1$ and $K \geq 2$. In this case, the sets $S(r, z)$ are half-spaces in $\mathbf{R}$, so the Helly property is satisfied. Moreover, the intersection graph of the family $\{S(r, z)\}_{r \in \mathcal{R}_z, z \in \mathcal{Z}}$ is an interval graph, so Proposition \ref{prop:tree} applies. Another example is when $K = 2$, and $J$ is arbitrary. In this case, the potential response graph is bipartite, and hence perfect, and the Helly property is trivially satisfied.
    
    More generally, suppose we can order the elements of $\mathcal{Z}$ and $\mathcal{R}$ so that $z_1 \prec z_2 \prec \dots \prec z_K$ $\mathcal{R}$ and, for all $r \in \{0, 1, \dots, J-1\}$, the difference $c_{r+1}(z) - c_{r}(z)$ is strictly increasing in $z$. For example, such an ordering is possible if $c_j(z) = a_j t(z) + b_j$, where all $a_j$ and all $t(z)$ are distinct. Then, as we show in Proposition \ref{prop:discrete_choice} in the Appendix, the implied support restriction $\mathcal{R}^*$ satisfies both conditions of Assumption \ref{A:support_regularity}. \qed
\end{example}

While most of the restrictions in the literature fall into one of the above settings, some important exceptions require a different approach to showing Assumption \ref{A:support_regularity} (i). We illustrate the arguments in the encouragement design model from \citet{BaiHuangTabordMeehan2026} in Appendix \ref{app:encouragement}.

\section{Extensions} \label{sec:extensions}
This section provides extensions of the main theoretical results in two directions.  Section \ref{ssec:continuous_outcome} extends Theorems \ref{thm:testable_implication}--\ref{thm:sharpness_of_mis} to accommodate general (e.g., continuous) outcome variables. As an application of this result, we characterize the sharp testable implications of IV validity with multi-valued treatment and instruments. Section \ref{ssec:multiple_models} presents a result that is useful for the practitioner to compare the testable implications of multiple models. Section \ref{sec:nonregular} discusses testable implications in non-regular models.

\subsection{Testable Implications for General Outcome Variables}\label{ssec:continuous_outcome}

Let $\cR$ be a subset of a Polish space, and let $\cX$ be a finite set. Let $\cR^*\subseteq \cR^{\cX}$ be the support of the potential response variable $R^*=(R(x))_{x\in\cX}$.
To characterize testable implications, we construct a sequence of partitions of the outcome space. For each $n\in\mathbb N$ and $k\in\{1,\dots,K\}$, let $\cA_{n,k}\equiv\{A_{1,k},\dots,A_{n,k}\}$ be a partition of $\cR_{z_k}$, consisting of $n$ nonempty sets.  For any $m\leq n$, let $\cA_{n,k}$ be a refinement of $\cA_{m,k}$. That is, every element $A_{j,k}$ of $\cA_{n,k}$ is contained in some element of $\cA_{m,k}$. One can construct such a sequence by recursively splitting the cells in $\cA_{m,k}$ into smaller disjoint sets. 

For each $n$, we define the potential response graph $G_n=(V_{G_n},E_{G_n})$ associated with $\cR^*$ as an undirected graph with vertices $V_{G_n} = \{v_{s,k}: s\in \{1,\dots,n\}, k \in \{1,\dots,K\} \}$ and edges $E_{G_n}=\{(v_{s,k},v_{s',k'})\in V_{G_n}\times V_{G_n}:\exists \, r^* \in \mathcal R^*: T(r^*,z_k) \in A_{s,k}, T(r^*,z_{k'}) \in A_{s',k'}\}.$ This construction naturally generalizes the definition of $G$ for discrete variables to the current setting.
The following theorem characterizes the sharp testable implication of $\cR^*$.

\begin{theorem}\label{thm:testable_implication_general_text}
Suppose Assumptions \ref{A:unconf}-\ref{A:support} hold. Let $\{G_n\}$ be a sequence of potential response graphs associated with $\cR^*$ and partitions $\{\cA_{n,k}\}$.
Let $\mathcal I_{G_n}$ be the collection of all maximal independent sets of $G_n$. Then,
\begin{align} 
\sum_{v_{s, k}\in I}P(R\in A_{s,k}\,|\,Z=z_k)\leq 1,\;\; \forall I\in \mathcal I_{G_n}, \text{and } n\in\mathbb N \label{eq:MIS_inequalities_n_text}
\end{align} 
is necessary for the existence of the joint distribution $Q$ of $R^*$ supported on $\cR^*$ that induces $P$. If Assumption \ref{A:support_regularity} holds for $G=G_n$ for every $n$, \eqref{eq:MIS_inequalities_n_text} is also sufficient. 
\end{theorem}
This theorem shows that sharp testable implications similar to those in the discrete case can be derived by replacing observable outcome values with a sequence of partitions. 
The following section illustrates an application of Theorem \ref{thm:testable_implication_general_text}.

\subsubsection{Exclusion and Monotonicity with Multivalued Treatments and IV}\label{ssec:multivalued_iv}
We now revisit Example \ref{ex:iv}. 
Let $Y$ be a continuous outcome, $D$ be a treatment in a finite, totally ordered set, and let $Z$ take $K$ different values. For example, $Y$ is healthcare utilization, $D$ indicates levels of insurance generosity (e.g., rates of coinsurance), and $Z$ is a random assignment to insurance plans with varying generosity  \citep{AronDine2013RAND}. The setting also includes the judge fixed effects design where $D$ is binary.

Consider testing the assumption that $z$ is excluded from the potential outcome, and $D$ is weakly monotonic in $z$:
\begin{equation} \label{E:exclusion_text}
\begin{array}{c}
	Z \perp (\{Y(d, z)\}_{d \in \mathcal{D}, z \in \mathcal{Z}}, \{D(z)\}_{z \in \mathcal{Z}});\\[3mm]
	Y(d, z) =_{a.s.} Y(d, z'), \;\; \forall d \in \mathcal{D}, \forall z, z' \in \mathcal{Z},
\end{array}
\end{equation}
\begin{equation} \label{E:monotonicity_text}
	D(z_1) \leq_{a.s.} \dots \leq_{a.s.} D(z_K).
\end{equation}
To our knowledge, sharp testable restrictions for this setting have been unknown outside an important special case where both treatment and instrument are binary (see Remark \ref{rem:existing_results}). 

For any $d \in \mathcal{D}$ and $S_d \subseteq \{1, \dots, K\}$, define $\underline{\ell}_d = \min_{\ell \in S_d}\ell$, $\overline{\ell}_d = \max_{\ell \in S_d} \ell$. Consider any finite set of tuples $\{ \{(B_{\ell, d}, d, z_\ell)\}_{\ell \in S_d} \}_{d \in \mathcal{D}}$ such that: (i) $\{B_{\ell, d}\}_{\ell \in S_d}$ form a partition of $\mathcal{Y}$, for each $d \in \mathcal{D}$; and (ii) $d < d' \implies \underline{\ell}_d \geq \overline{\ell}_{d'}$. Figure \ref{fig:exclusion} illustrates with $\mathcal{D} = \{1, 2, 3\}$ and $\mathcal{Z} = \{z_1, z_2, z_3, z_4\}$. Here, different colors represent the values $z_k$ of the instrument, and colored regions correspond to the sets $B_{k,d}, k\in S_d$. On the left side, $S_1 = \{4\}$, $S_2 = \{1, 2, 3, 4\}$, $S_3 = \{1\}$, and on the right side, $S_1 = \{3, 4\}$, $S_2 = \{2, 3\}$, $S_3 = \{1, 2\}$. It is clear that, for each $d$, $\{B_{\ell, d}\}_{\ell \in S_d}$ form a paritition of $\cY$. Furthermore, they satisfy (ii). For example, on the right side of Figure \ref{fig:exclusion}, the index sets $S_d,d\in\cD$ are constructed so that $\underline{\ell}_1=3=\overline{\ell}_2\ge \underline{\ell}_2=2=\overline{\ell}_3.$

Viewing each $(B_{\ell, d}, d, z_\ell)$ as a vertex of a graph, the vertices satisfying (i) and (ii) are independent of each other. This is because, for any pair of such vertices, they are independent in $G$ if either  $z_\ell \ne z_{\ell'}$, $d = d'$ and $B_{\ell,d}\cap B_{\ell',d'} = \varnothing$ (i.e., $Y(d,z_\ell)$ and $Y(d,z_{\ell'})$ take values in disjoint sets despite $d$ being common), violating the exclusion, or $z_\ell \leq z_{\ell'}$ and $d_\ell > d_{\ell'}$ (i.e., $D(z_\ell)>D(z_{\ell'})$ despite $z_\ell\leq z_{\ell'}$), violating the monotonicity. Forming maximal independent sets from such vertices gives sharp testable restrictions.
The following proposition summarizes the argument above. The restrictions' validity and sharpness follow from Theorem \ref{thm:testable_implication_general_text}.
\begin{figure}
\centering 
		\includegraphics[width = 0.35\linewidth]{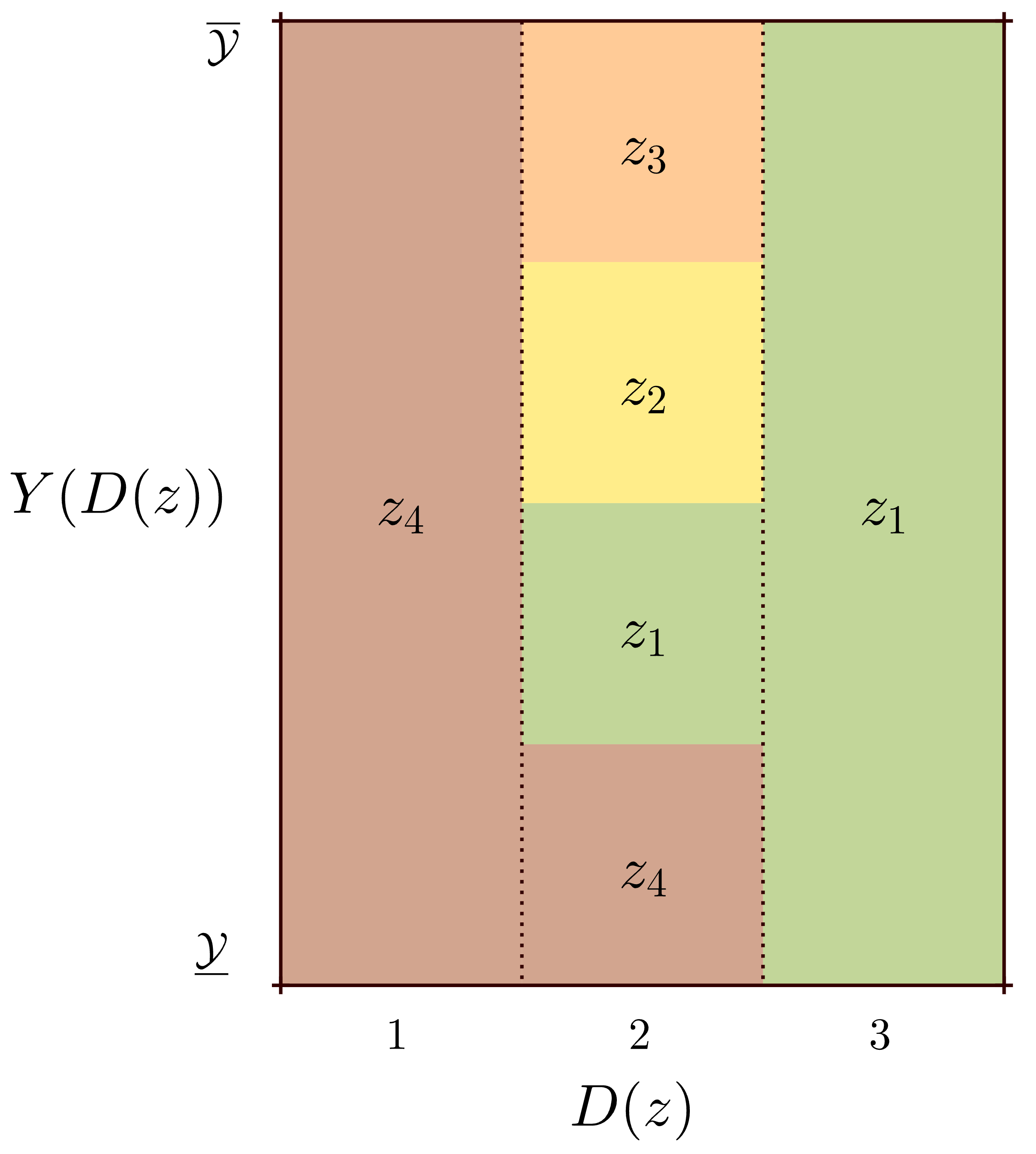} \hspace*{1cm}
		\includegraphics[width = 0.35\linewidth]{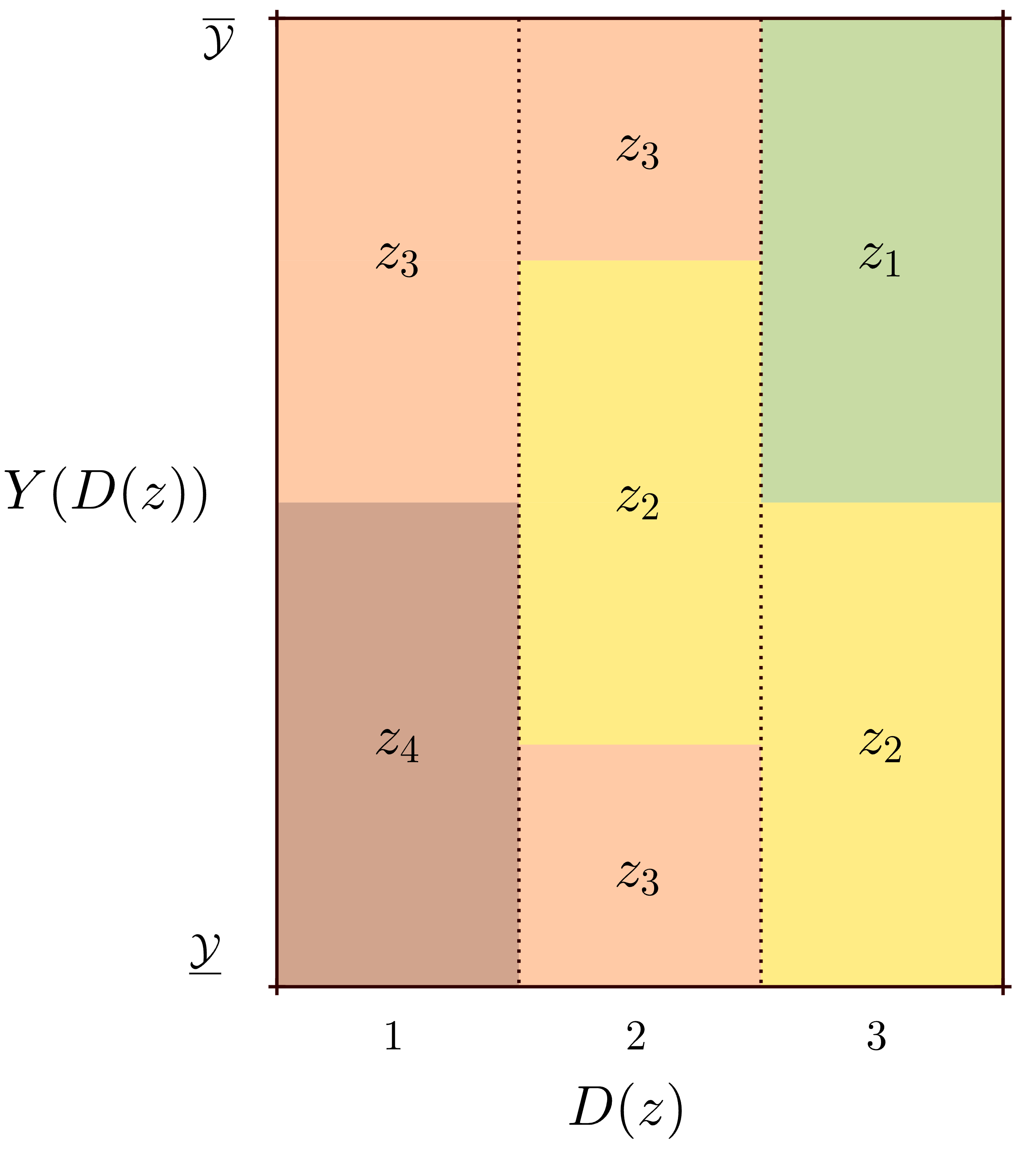}
\caption{Examples of partitions for testing exclusion and monotonicity restrictions.} \label{fig:exclusion}
\end{figure}

\begin{proposition}[Testing Exclusion and Monotonicity]\label{prop:iv_val}
	Conditions \eqref{E:exclusion_text} and \eqref{E:monotonicity_text} jointly hold if and only if
	\begin{align}
	\begin{array}{c}
			\displaystyle \sum_{\ell \in S_d, d \in \mathcal{D}} P(Y \in B_{\ell, d}, D = d \,|\,Z = z_\ell) \leq 1,\\[3mm]
			\forall \{ \{(B_{\ell, d}, d, z_\ell)\}_{\ell \in S_d} \}_{d \in \mathcal{D}}: \bigcup_{\ell \in S_d} B_{\ell, d} = \mathcal{Y}, \; \forall d,d' \in \mathcal{D},\; d < d' \implies \underline{\ell}_d \geq \overline{\ell}_{d'}.\label{eq:ex_mon_implication_K}
	\end{array}
	\end{align}
\end{proposition}

\begin{remark}\label{rem:existing_results}
The testable implication obtained above nests existing restrictions. When the treatment and instrument are binary, \eqref{eq:ex_mon_implication_K} can be expressed as 
\[
	P(Y\in A,D=1\,|\,Z=0) \leq P(Y\in A,D=1\,|\,Z=1)\label{eq:BP1}
\]
or 
\[
P(Y\in A,D=0\,|\,Z=1)\leq P(Y\in A,D=0\,|\,Z=0)	\label{eq:BP2}
\]
for Borel sets $A\subseteq \cY$ (see Corollary \ref{cor:iv_K2}). These conditions were derived in \cite{balke1997bounds} and \cite{Heckman:2005aa}. \cite{Kitagawa15} showed they were sharp. Proposition \ref{prop:iv_val} generalizes the inequalities above to settings where both $D$ and $Z$ are multi-valued.
\qed 
\end{remark}

\subsection{Comparing and Interpreting Testable Implications of Multiple Models}\label{ssec:multiple_models}
Consider a collection of models characterized by different sets of support restrictions. The framework developed thus far can be used to derive testable implications and to apply existing model selection methods \citep[e.g.,][]{shi2015model,hsu2017model}. When the model restrictions are nested, we can further show that adding or relaxing support restrictions has systematic effects on the implied inequalities. We provide a formal result and illustrative examples below.

Consider two models where Model 1 imposes a stronger support restriction than Model 2. The following proposition shows that every inequality implied by Model 1 is either a (weakly) tightened version of an inequality implied by Model 2 or a new inequality that does not appear among the testable implications of Model 2.

\begin{proposition}\label{prop:comparison}
Suppose Assumption \ref{A:unconf} holds. Let $\cR^*_i,i=1,2$ be the support restrictions of two models with $\cR^*_1\subset \cR^*_2$. Let $G_1$ and $G_2$ be their corresponding potential response graphs. For any  $I_1\in \mathcal I_{G_1}$, let
\begin{align} 
\sum_{v_{s, k}\in I_1}P(R\in A_{s,k}\,|\,Z=z_k)\leq 1,~\label{eq:tightened}
\end{align} 
be an inequality that is part of the testable implication for $\cR^*_1$. Then, it satisfies either one of the following conditions:
\begin{enumerate}
	\item[(i)] \eqref{eq:tightened} is a (weakly) tightened version of an inequality implied by $\cR^*_2$:
	\begin{align*}
	\sum_{v_{s, k}\in I_1}P(R\in A_{s,k}\,|\,Z=z_k)=\sum_{v_{s, k}\in I_2}P(R\in A_{s,k}\,|\,Z=z_k)+\sum_{v_{s,k}\in V'}P(R\in A_{s,k}\,|\,Z=z_k)
	\end{align*}
	for some  $I_2\in \cI_{G_2}$ and $V'\subseteq V_{G_1}$ such that  $I_2\cap V'=\emptyset.$
	\item[(ii)] \eqref{eq:tightened} is an inequality that does not appear as a testable implication of $\cR^*_2$. 
\end{enumerate}
\end{proposition}
The tightening of the inequalities occurs because adding support restrictions deletes some of the edges present in the graph of the baseline model.
This result allows practitioners to evaluate how adding or removing assumptions impacts the resulting restrictions. The following examples illustrate this point.
\setcounter{example}{0}
\begin{example}[Instrumental Variables (continued)]
Suppose $\cY=\mathbb R$, $\cD=\{0,1\}$, and $\cZ=\{0,1\}$. Let $\mathcal R^*_2$ be the set of potential responses satisfying exclusion: $Y(d,z)=Y(d), \forall (d,z)\in \cD\times\cZ$. Its sharp testable implication is (see Appendix \ref{ssec:comparison}):

\begin{align}
\begin{array}{c}
 P(Y \in A, D = d \,|\,Z = 0)+P(Y \in A^c, D = d \,|\,Z = 1) \leq 1, \;\;\;  \forall d \in \{0,1\}, ~\forall A\subseteq \mathcal Y.	\\[2mm]
\end{array}\label{eq:nested1}
\end{align}
Now suppose that $\cR^*_1$ collects the potential responses satisfying both exclusion and monotonicity $D(1)\geq_{a.s.} D(0)$. Then, the sharp testable implication becomes
\begin{align}
P(Y\in \cY,D=0|Z=1)+P(Y\in A,D=1|Z=0)+P(Y\in A^c,D=1|Z=1)&\leq 1\label{eq:nested2}\\
P(Y\in A,D=0|Z=0)+P(Y\in A^c,D=0|Z=1)+P(Y\in \cY,D=1|Z=0)&\leq 1,\label{eq:nested3}
\end{align}
for all $A\subseteq \mathcal Y$.
Each of these inequalities results from tightening \eqref{eq:nested1}; e.g., \eqref{eq:nested2} is obtained by adding $P(Y \in \cY,D=0|Z=1)$ to the left side of \eqref{eq:nested1}. This tightening happens because the monotonicity restriction removes support points, thereby expanding the set of maximal independent sets.  \qed 
\end{example}

\begin{example}[Partial Monotonicity (continued)]
 Recall that the sharp testable implication of the partial monotonicity assumption was given by  \eqref{eq:pm_ineq1}-\eqref{eq:pm_ineq5}. If we additionally assume
$D(1,0)\geq_{a.s.}D(0,1)$, we obtain the following ordering of the potential treatment:
\begin{align}
D(1,1)\geq_{a.s.}D(1,0)\geq_{a.s.}	D(0,1)\geq_{a.s.}D(0,0),	\label{eq:IAM2}
\end{align}
which satisfies the IA monotonicity. The additional restriction rules out the $Z_2$-complier ($2c$). This leads to a new maximal independent set, and hence a new inequality:
\begin{align}
&~~P(D=1|Z=(0,1))+P(D=0|Z=(1,0))\leq 1.\label{eq:pm_iam_new}
\end{align}
The new set of inequalities  \eqref{eq:pm_ineq1}-\eqref{eq:pm_ineq5} and \eqref{eq:pm_iam_new} form a sharp testable implication of \eqref{eq:IAM2}.\footnote{Assumption \ref{A:support_regularity} holds in this example with or without $D(1,0)\geq_{a.s.}D(0,1)$.} More details are provided in Appendix \ref{ssec:comparison}. \qed 
\end{example}

\begin{example}[Mediation (continued)]
Let $Y$ be a binary outcome, $D$ be a binary treatment, and $M_1, M_2$ be binary mediators. Suppose the effect of $D$ is fully mediated by either $M_1$ or $M_2$ (or both):
$$Y(m_1,m_2,d)=Y(m_1,m_2,d'),~~\forall (m_1,m_2,d,d')\in \{0,1\}^4.$$ 
The sharp testable implication of this assumption contains eight inequalities, including the following ones (see Appendix \ref{ssec:mediation_comparison}):
\begin{align}
P(Y=0,M_1=0,M_2=0 \mid D=0) + P(Y=1,M_1=0,M_2=0 \mid D=1) &\leq 1 \label{eq:d_ex1}\\
P(Y=0,M_1=0,M_2=1 \mid D=0) + P(Y=1,M_1=0,M_2=1 \mid D=1) &\leq 1. \label{eq:d_ex3} 
\end{align}
Suppose we strengthen the assumption to $Y(m_1,m_2,d)=Y(m_1,m_2',d'),\forall (m_1,m_2,m_2',d,d')\in \{0,1\}^5$, meaning the second mediation channel does not exist. Then, the sharp testable implication becomes four inequalities, including the following:
\begin{align}
&P(Y=0,M_1=0,M_2=0 \mid D=0) + P(Y=0,M_1=0,M_2=1 \mid D=0) \notag\\
&\quad + P(Y=1,M_1=0,M_2=0 \mid D=1) + P(Y=1,M_1=0,M_2=1 \mid D=1) \leq 1. \label{eq:m2d_ex1}
\end{align}
Each inequality is obtained by tightening one or more of the inequalities that define the sharp testable implication of the full mediation assumption; e.g., \eqref{eq:m2d_ex1} is obtained by adding the left-hand side of \eqref{eq:d_ex1} and \eqref{eq:d_ex3}, yielding a tighter restriction. 
By representing different modeling assumptions on mediators through their corresponding graphs, our framework facilitates direct comparisons among them, complementing the approach of \cite{Kwon:2024aa}, who analyze each model individually. \qed 
\end{example}

\subsection{Beyond Regular Support Restrictions}\label{sec:nonregular}
Our results can also be extended to some non-regular support restrictions that do not satisfy Assumption \ref{A:support_regularity} (ii). To elaborate, for an undirected graph $H$ with $N$ vertices, let $\mathcal{S}_H$ denote the set of all (not necessarily maximal) independent subsets of $H$. A graph $H$ is called \textit{rank-perfect} if 
\[
\text{Conv}(\{\bm{1}_S: S \in \mathcal{S}_H\}) = \{x \geq 0: x'\bm{1}_V \leq \max_{S \in \mathcal{S}_H}|S \cap V|, \forall\, V \subseteq V_H\}.
\]
Let $\overline{G}$ denote the graph complement of $G$, i.e., a graph with $V_{\overline{G}} = V_{G}$ and $E_{\overline{G}} = \{(v, v') \in V_G \times V_G: (V, v') \notin E_G\}$. Let $\mathcal{C}_G$ denote the set of all maximal cliques in $G$. 

The following result provides a direct extension of Theorem \ref{thm:sharpness_of_mis}.
\begin{theorem} \label{thm:rank_perfect}
    Suppose Assumptions \ref{A:unconf}, \ref{A:support}, and \ref{A:support_regularity} (i) hold, and $\overline{G}$ is rank-perfect. Let $\ell_G(V) = \max_{C \in\mathcal{C}_G}|C \cap V|$. Then, 
   \begin{align} 
\sum_{v_{r, z}\in V}P(R=r\,|\,Z=z)\leq \ell_G(V),\;\; \forall \,V\; \subseteq V_G, \label{eq:level_inequalities}
\end{align} 
is necessary and sufficient for the existence of the joint distribution $Q$ of $R^*$ supported on $\mathcal{R}^*$ that induces $P$. 
\end{theorem}
The MIS inequalities in Theorem \ref{thm:sharpness_of_mis} are a special case of the above, since $\ell_G(I) = 1$, for all $I \in \mathcal{I}_G$. In addition, many of the inequalities in \eqref{eq:level_inequalities} are implied by the MIS inequalities. Indeed, consider a set $V \subseteq V_G$ such that the induced subgraph $G[V]$ is perfect. Since $\ell_G(V)$ is precisely the clique number (i.e., the size of the largest clique) of $G[V]$, the set $V$ can be partitioned in $\ell_G(V)$ independent sets. Then, the sum of the corresponding MIS inequalities implies \eqref{eq:level_inequalities} for $V$. In general, however, testable implications of support restrictions that correspond to rank-perfect, but not perfect, graphs, are much richer. For example, the $G$ may contain a large number of odd holes, or more complicated imperfect subgraphs, each leading to a non-trivial inequality of the form \eqref{eq:level_inequalities}. The following example illustrates.  

\begin{example}[IV Exclusion]
    Let $\mathcal{Y} = \mathcal{D} = \{0, 1\},$ and $\mathcal{Z} = \{z_1, \dots, z_K\}$ with $K \geq 2$. Let $\mathcal{R}^*$ be defined only by the exclusion restriction, $Y(d, z) = Y(d, z')$, for any $d \in \{0, 1\}$, and $z, z' \in \mathcal{Z}$. For $K = 2$, the graph $G$ is bipartite, hence perfect, so Assumption \ref{A:support_regularity} (ii) holds and Theorem \ref{thm:sharpness_of_mis} applies. For $K \geq 3$, the graph complement $\overline{G}$ contains odd cycles of size five, so it is not perfect, but can be shown to be rank perfect. Sharp testable implications for this model have been obtained by \citet{KedagniMourifie2020}. In Section \ref{ssec:existing_inequalities} in the Appendix, we show that their instrumental variable inequalities take the form \eqref{eq:level_inequalities} for certain sets $V$ with $\ell_G(V) \in \{1, 2, 3\}$. \qed 
\end{example}

\section{Computation and Inference} \label{sec:computation_inference}
This section discusses computational aspects of the proposed framework and inference procedures.
Section \ref{ssec:computation} outlines how to construct potential response graphs, compute the testable implication, and check their sharpness. A pseudo-code for the main steps is given in Algorithm \ref{alg:master}.
To facilitate implementation, we provide a \texttt{Python} library for these tasks.\footnote{The library is available at \url{https://github.com/hkaido0718/SupportRestriction}. It contains functions for defining graphs, obtaining maximal independent sets, and checking the regularity condition, along with their applications to the examples in this paper.}  Section \ref{ssec:inference} shows how to conduct inference, and Section \ref{ssec:alternative_approaches} discusses alternative approaches.

\subsection{Computing Testable Implications and Checking Regularity}\label{ssec:computation}
A potential response graph $G$ can be constructed by one of two methods:
1. By checking pairwise compatibility of vertices with the support restriction; or 
2. By enumerating all $r^* \in \mathcal{R}^*$.

To elaborate on Method 1, consider Example~\ref{ex:iv}, in which the model imposes the exclusion restriction:
\[
Y(d,z)=Y(d,z'), \qquad \forall (d,z,z')\in \mathcal{D}\times\mathcal{Z}^2.
\]
Let $v_{y,d,z}$ and $v_{y',d',z'}$ denote a pair of vertices such that the treatment status is identical, $d=d'$, but $z\neq z'$.
If $y\neq y'$, no potential response function can be compatible with this pair, as doing so would violate the exclusion restriction. Hence, no edge exists between $v_{y,d,z}$ and $v_{y',d',z'}$. By contrast, if $y=y'$, the pair is compatible and an edge exists. This observation motivates a first method for graph construction: checking, for each pair of vertices, whether they are jointly compatible with the imposed support restrictions. This method is applicable to any model for which Assumption \ref{A:support_regularity} (i) holds.

Next, we discuss how to compute the inequalities derived in Theorem \ref{thm:testable_implication}. Since most support restrictions $\mathcal{R}^*$ of interest are regular, we focus on the characterization in \eqref{eq:MIS_inequalities}.  The potential response graph $G$ has $|\mathcal{Z}| = K$ parts with $|\mathcal{R}_k| = J$ vertices in each, so $|V_G| = KJ$ vertices in total. The number of edges is at most $|E_G| \leq J^2K(K - 1)/2$.

First, to compute all maximal independent sets (MIS), we  can use the algorithm of \citet{tsukiyama1977new}. The algorithm is output-sensitive: its time complexity is proportional to the total number of MIS and is of order $O(|V_G| |E_G| |\mathcal{I}_G|)$. Depending on the structure of the graph $G$, the number of MIS may be drastically different, ranging from $O(1)$ to $O(2^{|V(G)|})$, so the computing time will scale accordingly. The algorithm is implemented in the functions \texttt{maximal\_ivs} in an \texttt{R} package \texttt{igraph} and \texttt{maximal\_independent\_vertex\_sets} in the \texttt{Python} interface of the same package.

 Second, to guarantee that the obtained testable implications are sharp, we need to verify that a given support restriction $\mathcal{R}^*$ is regular, that is, (i) Every maximal clique of \(G\) corresponds to an admissible response type and (ii) $G$ is a perfect graph. Both conditions can be verified either analytically, as discussed in Section \ref{ssec:support_regularity}, or numerically, as discussed below. Condition (i) can be verified by finding all maximal cliques of size $K$ in $G$ and comparing them with the support points. This step can be performed efficiently using a version of \citet{bron1973algorithm} algorithm as in \citet{eppstein2010listing}, implemented, for example, in the \texttt{max\_cliques} function in an \texttt{R}  package \texttt{igraph}. 
 Condition (ii) can be verified in polynomial time using the fact that a graph $G$ is perfect if and only if none of its induced subgraphs is an odd cycle of length five or more, or a complement of one \citep{chudnovsky2006strong, chudnovsky2020detecting}. An implementation is available in \texttt{SageMath} as a built-in function \texttt{is\_perfect} or using our Algorithm \ref{alg:det_odd_holes} in Appendix \ref{app:computation}.

\subsection{Inference Based on Moment Inequalities}\label{ssec:inference}
Recall that $N = \sum_{k \in \mathcal{Z}}|\mathcal{R}_k|$ and $M = |\mathcal{R}^*|$. Denote $\beta(r\,|\,z) = P(R = r\,|\,Z = z)$ and $\beta_0(P) = ( (\beta(r \,|\,z))_{r \in \mathcal{R}_z})_{z \in \mathcal{Z}}$. If $R$ contains continuous components, let $\beta(s\,|\,z)=P(R \in  A_{s,z}\,|\,Z = z)$ and define $\beta_0(P)= ( (\beta(s \,|\,z))_{A_{s,z} \in \mathcal{A}_{n,z}})_{z \in \mathcal{Z}}$ as in Section \ref{ssec:continuous_outcome}.  All objects introduced below are understood to be redefined accordingly.

Denoting $\mu(P) = (\bm{1}_I'\beta_0(P) - 1)_{I\in \mathcal{I}_G}$ and given an i.i.d. sample $\{(R_i, Z_i)\}_{i = 1}^n$ from a distribution $P \in \mathbf{P}$, the goal is to test $H_0: P \in \mathbf{P}_0$ against $H_1: P \in \mathbf{P} \backslash \mathbf{P}_0$, where
\begin{equation} \label{E:mim_formulation}
\mathbf{P}_0 = \{P \in \mathbf{P}: \mu(P) \leq 0\}.
\end{equation}
The restrictions of this form are known as \emph{moment inequalities} without any nuisance parameters.
 Testing hypotheses of this form is an extensively studied problem; see \citet{canay2017practical} for a technical review and \citet{canay2023user} for a user's guide. 
 In the following two sections, we provide a step-by-step implementation of the tests based on \eqref{E:mim_formulation}, which we employ in our Monte Carlo experiments and empirical application.

\subsubsection{Inference Based on Instruments and Discrete Covariates}\label{ssec:inference_disc}
Denote 
\[
\begin{array}{l}
	p(r, z) = P(R = r, Z = z);\\[2mm]
	\pi(z) = P(Z = z); \\[2mm]
	\beta(r\,|\,z) = \frac{P(R = r, Z = z)}{P(Z = z)}; \\[2mm]
	\beta_0(P) = (\beta(r \,|\,z)_{r \in \mathcal{R}})_{z \in \mathcal{Z}};
\end{array}
\hspace{2cm}
\begin{array}{l}
	\hat{p}_n(r, z) =  \frac{1}{n} \sum_{i = 1}^n \bm{1}(R_i = r, Z_i = z);\\[2mm]
	\hat{\pi}_n(z) = \frac{1}{n} \sum_{i = 1}^n \bm{1}(Z_i = z);\\[2mm]
	\hat{\beta}_n(r\,|\,z) = \frac{\hat{p}_n(r, z)}{\hat{\pi}_n(z)}; \\[2mm]
	\hat{\beta}_{0, n} = \left((\hat{\beta}_n(r \,|\,z))_{r \in \mathcal{R}}\right)_{z \in \mathcal{Z}}.
\end{array}
\]

A straightforward calculation shows that, for each $r \in \mathcal{R}_z$ and $z \in \mathcal{Z}$,
\[
    \sqrt{n}(\hat{\beta}_{n}(r\,|\,z) - \beta(r\,|\,z)) = \frac{1}{\sqrt{n}} \sum_{i = 1}^n \frac{\pi(z) \bm{1}(D_i = r, Z_i = z) - p(d, z) \bm{1}(Z_i = z)}{\pi(z)^2} + o_P(1).
\]
Thus, the asymptotic covariance matrix of $\hat{\beta}_{0, n}$ can be consistently estimated by
\begin{align}
\hat{V}_n = \frac{1}{n}\sum_{i = 1}^n \hat{b}(R_i, Z_i) \hat{b}(R_i, Z_i)',\label{eq:Vn}
\end{align}
where
\[
\hat{b}(R_i, Z_i) = \left(\left(\frac{\hat{\pi}_n(z) \bm{1}(R_i = r, Z_i = z) - \hat{p}_n(r, z) \bm{1}(Z_i = z)}{\hat{\pi}_n(z)^2} \right)_{r \in \mathcal{R}_z} \right)_{z \in \mathcal{Z}}.
\]
Next, let $\bm{A} \in \{0, 1\}^{|\mathcal{I}_G| \times N}$ denote the matrix with rows $a_I' = \bm{1}_I'$. Let $\mu_{I}(P) = a_I'\beta_0(P) - 1$, and
\[
\begin{array}{l}
	\hat{\mu}_n = (\hat{\mu}_{I, n})_{I \in \mathcal{I}_G} = (a_I'\hat{\beta}_{0, n} - 1)_{I\in \mathcal{I}_G};\\[2mm]
	\hat{\Sigma}_n = ( a_{I_1}'\hat{V}_n a_{I_2} )_{I_1, I_2 \in \mathcal{I}_G} = \bm{A} \hat{V}_n \bm{A}';\\[2mm]
	\hat{D}_n = \text{diag}((\hat{\sigma}_{I, n})_{I \in \mathcal{I}_G} ) =\text{diag}(\hat{\Sigma}_{n})^{1/2}; \\[2mm]
	\hat{\Omega}_n = \hat{D}_n^{-1} \hat{\Sigma}_n \hat{D}_n^{-1}.
\end{array}
\]
Consider the test statistic:\footnote{Other choices are possible, see \citet{andrews2010inference}. We focus on the maximum test statistic since theoretical guarantees for the resulting tests are also available in high-dimensional settings; see \citet{chernozhukov2019inference} and \citet{bai2022two}.}
\[
T_n = \max \left\{ \max \limits_{I \in \mathcal{I}_G} \frac{\sqrt{n}\hat{\mu}_{I, n} }{\hat{\sigma}_{I, n}},\; 0\right\}.
\]
The existing tests differ in how they construct the critical value with which $T_n$ is compared. 
The difficulty lies in bounding the slackness parameter $\sqrt{n}\mu_I(P)$, for $P \in \mathbf{P}_0$, which cannot be consistently estimated. Letting  $\hat{u}_{I, n}$ denote a suitable upper bound on $\sqrt{n}\mu_I(P)$, we set $\hat{c}_{n, \alpha}$ to be a consistent estimator of the $(1-\alpha)$-th quantile of the distribution of 
\[
\max \left\{ \max \limits_{I \in \mathcal{I}}  \frac{\sqrt{n}(\hat{\mu}_{I, n} - \mu_{I}(P)) + \hat{u}_{I,n}}{\hat{\sigma}_{I, n}} + , \;0 \right\}.
\]
Such an estimator may be obtained using bootstrap or Normal approximation, conditional on the data. Depending on the number of inequalities in \eqref{E:mim_formulation} relative to the sample size, different testing procedures may be appropriate. One may use, for example, \citet{andrews2010inference} or \citet{romano2014practical} when $|\mathcal{I}_G|$ is small, and \citet{chernozhukov2019inference} or \citet{bai2022two} when $|\mathcal{I}_G|$ is large. Then, under regularity conditions, the test
\[
\phi_n = \bm{1}(T_n > \hat{c}_{n, \alpha})
\]
is uniformly asymptotically valid over a large class of distributions.

If discrete covariates $W$ are present, one can use the same inference procedure, while replacing each conditioning statement $Z=z$ with $Z=z,\X=\x$, e.g., $\beta(r\,|\,z,w)=P(R =r\,|\,Z = z,W=w)$.

\subsubsection{Inference with Continuous Covariates} \label{ssec:inference_cont}
Suppose $W$ contains finitely supported components $W_d$ and continuous components $W_c$, distributed over $[0,1]^{|\x_c|}$. For such settings, one can use methods developed by \citet{andrews2013inference}, \citet{CLR13}, \citet{armstrong2015asymptotically}, \citet{cox2023simple}, and \citet{andrews2023inference}. 

We outline below the test of  \cite{CLR13} (CLR, henceforth).
Let $\x=(\x_d',\x_c')'$. For each $\x\in \mathcal{\X}$, let $\beta_0(\x) = (\beta(r \,|\,z,\x)_{r \in \mathcal{R}})_{z \in \mathcal{Z}}$. 
Let $v=(\x,I)\in \mathcal V=\{(\x,I):\x\in\mathcal \X,I\in\cI\}$ and define $\mu(v)=a_I'\beta_0(\x) - 1$. 
 Then, a testable implication can be formulated as an intersection bound:
\begin{align}
\sup_{v\in \cV}\mu(v)\leq 0.
\end{align}
Let $\hat\beta_{0,n}(\x)=((\hat \beta_n(r\,|\,z,\x))_{r\in\cR})_{z\in \cZ}$ be a series estimator of $\beta_0(\x)$ such that
\begin{align}
\hat \beta_n(r\,|\,z,\x)=b_n(\x)'\hat\chi_n(r,z),~r\in\cR,~z\in\cZ,	
\end{align}
where $b_n(\x)=(b_{n1}(\x),\dots,b_{nm_n}(\x))'$ is a vector of  basis functions. For example, $b_{nk}(\x)=\bm{1}(\x_d=\tilde \x)\times \psi_{nh}(\x_c)$, where $\tilde \x$ is a support point of the discrete component, and $\psi_{nh}$ is a tensor-product $B$-spline of a fixed order.
Let $\hat\mu_n=(\hat\mu_n(v),v\in \cV)$ be an estimator of $\mu$ defined pointwise by $\hat \mu_n(v)=	a_I'\hat\beta_{0,n}(\x) - 1.$ Define the precision-corrected bound:
$$S_{n,\alpha}=\sup_{v\in \mathcal{V}}(\hat\mu_n(v) - \hat c_{n,\alpha}\hat \sigma_n(v)),
$$
 where $\hat c_{n,\alpha}$ is a critical value (computed by CLR's Algorithm 1), and $\hat \sigma_n(v)$ is an estimator of the standard error of $\hat\mu_n(v)$. Then, under regularity conditions, the test
 \[
\phi_n = \bm{1}(S_{n,\alpha} > 0)
 \]
is uniformly asymptotically valid (CLR, Theorem 5). The details of how to compute $\hat\sigma_n$ are summarized in Appendix \ref{ssec:clr_details}.

\subsection{Alternative Approaches}\label{ssec:alternative_approaches}
An alternative approach directly tests compatibility with the support restriction using the vertex representation of the model’s observable distributions. The corresponding null set is
\begin{equation} \label{E:null}
	\mathbf{P}_0 = \{P \in \mathbf{P}: \beta_0(P) \in \overline{\mathbb{P}} \},
\end{equation}
where $\overline{\mathbb{P}}$ is defined in \eqref{E:P_bar}. This formulation is valid for any support restriction. Under Assumption \ref{A:support_regularity}, it coincides with the null set in \eqref{E:mim_formulation}. Otherwise, the MIS inequalities remain necessary but may not be sufficient for model compatibility. Denoting
\[
A = \begin{bmatrix}
	A^* \\
	\bm{1}_{M}'
\end{bmatrix};
\;\;\;\;
\beta(P) = \begin{bmatrix}
	\beta_0(P)\\
	1
\end{bmatrix},\;\;
\]
the null set of distributions can also be written as
\begin{equation} \label{E:null_FSST}
\mathbf{P}_0 = \{P \in \mathbf{P}: Ax = \beta(P) \text{ for some } x \in \mathbf{R}^{M}_{+}\}.
\end{equation}
In this form, the null hypothesis can be tested  using any of the available methods for testing existence of solutions in linear systems with known coefficients, such as \cite{KitamuraStoye2018}, \cite{fang2023inference}, and \cite{goff2025inferencevaluelinearprogram}. For example, \cite{fang2023inference} computes a test statistic by solving linear programs (LP) and compares it to a critical value obtained by repeatedly solving LPs across bootstrap replications.\footnote{If $\hat\beta_{n,0}$ is not in the range of $A^*$, one also needs to additionally solve quadratic programs.} See Appendix \ref{sec:fsst} for a detailed discussion.  

Another approach relies on random set theory \citep[see, e.g.][]{molchanov2018random}.  Recall that $S(r, z) = \{r^* \in \mathcal{R}^*: r = T(r^*, z)\}$ denotes the set of response types compatible with a given values of response and input. The random set $S(R, Z)$ gives a set-valued prediction for the latent random vector $R^*$, and the model can be summarized as
\[
\begin{array}{l}
	R^* \in_{a.s.} S(R, Z); \\
	R^* \perp Z.
\end{array}
\]
By Theorem 2.34 in \citet{molchanov2018random}, the above is equivalent to
\begin{equation} \label{E:artstein_all}
Q(R^* \in A) \geq P(S(R, z) \subseteq A \,|\, Z = z), \text{ for all } A \subseteq \mathcal{R}^*, z \in \mathcal{Z}.
\end{equation}
However, for $r \ne r'$, we have $S(r, z) \cap S(r', z) = \varnothing$, by construction. Then, Theorem 1 in \citet{luo_ponomarev_wan} implies that \eqref{E:artstein_all} is equivalent to
\[
Q(R^* \in S(r, z)) = P(R = r\,|\,Z = z), \text{ for all } r \in \mathcal{R}_z, z \in \mathcal{Z},
\]
which are precisely the equations defining $\overline{\mathbb{P}}$. Thus, under our assumptions, the linear system approach and random set approach are equivalent. The random set approach is especially useful when $Z$ is partially independent of $R^*$ and has rich support. See \citet{luo_ponomarev_wan} for a discussion.

The choice between testing procedures depends on the structure of the support restriction $\cR^*$ and the associated polytope $\overline{\mathbb{P}}$. When $\overline{\mathbb{P}}$ has many vertices (i.e., latent types) but relatively few faces (i.e., inequalities), tests based on \eqref{E:mim_formulation} are typically preferable. As discussed earlier, a broad class of established testing and model selection procedures is available. 
Conversely, when $\overline{\mathbb{P}}$ has few vertices and many faces, tests based on \eqref{E:null_FSST} may be more suitable.

To assess the relative complexity, we suggest applying the output-sensitive algorithm of \citet{tsukiyama1977new}.  If $|I_G|$ is moderate, it will quickly enumerate all MISs, making implementation of tests based on \eqref{E:mim_formulation} straightforward. The hypotheses considered in our empirical application have $|I_G|$ ranging from 1 to 726. For each hypothesis, our library found all MISs within 3 seconds.\footnote{The library was run on Google Colab with Python 3 Google Compute Engine with 12.67 GB RAM.}
 If $|I_G|$ is prohibitively large, we suggest checking the number of vertices of $\overline{\mathbb{P}}$. If this number is moderate, tests based on \eqref{E:null_FSST} are feasible.

Another possibility is to treat the testable implications as moment inequalities involving a nuisance parameter, and to apply inference methods designed for such settings \citep[see][]{andrews2023inference}. For instance, the equality in \eqref{E:null_FSST} can be reinterpreted as a pair of opposing moment inequalities with nuisance parameter $x \in \mathbf{R}^M_+$. 
In contrast to the approach based on \eqref{E:mim_formulation}, however, this method introduces additional parameters, which can make inference computationally intensive when $M$ is large.\footnote{\citeposs{andrews2023inference} method uses the vertices of a feasibility set defined by a linear program. As they note (p. 2781), “Enumerating all of the vertices is, however, computationally prohibitive when there are many moments or nuisance parameters,” and they offer two computational shortcuts. One applies only to special cases, while the other, a more general method, requires solving a linear program at each step of a bisection algorithm, which becomes increasingly costly as $M$ (or $N$) grows.}

Finally, the nature of the available covariates is also an important consideration. An advantage of the half-space representation of $\mathbf{P}_0$ is its flexibility in accommodating covariates with rich, potentially continuous, support. For discrete covariates with $L$ support points, the computational cost of tests based on \eqref{E:mim_formulation} scales linearly in $L$. Moreover, our framework naturally accommodates continuous covariates through conditional moment inequality tests, as discussed in Section \ref{ssec:inference_cont}.
By contrast, tests based on \eqref{E:null_FSST} require explicit discretization of covariates. While discretization is not inherently problematic, it introduces additional computational complexity. In particular, discretizing a covariate $W$ into $L$ bins expands the linear system in \eqref{E:null_FSST} to dimension $(N+1)L \times ML$. This increase substantially raises the computational burden, especially when the resulting optimization problems must be solved repeatedly across bootstrap replications.\footnote{Although computational shortcuts may be available in specific implementations, the time complexity of the commonly used interior-point method is $O!\left(\sqrt{d},\log!\left(\tfrac{1}{\epsilon}\right)(p^{2}d + d^{2}p)\right)$, where $p=(N+1)L$ and $d=ML$, to achieve accuracy $\epsilon$. This cost is incurred in each bootstrap replication.}

\section{Monte Carlo Experiments}

In this section, we investigate the power of testing procedures based on the testable implications derived in Theorem \ref{thm:testable_implication}.

\subsection{Partial Monotonicity}

Let $D \in \{0, 1\}$ and $Z_1, Z_2 \in \{0, 1\}$, and consider the null hypothesis of partial monotonicity:
\[
D(z_1, z_2) \leq_{a.s.} D(z_1', z_2') \iff z_1 \leq z_1' \text{ and } z_2 \leq z_2'.  
\]
This support restriction is regular (by Propositions \ref{L:helly_simple} and \ref{prop:comparability}), so the testable implications of Theorem \ref{thm:testable_implication} are sharp. 

Suppose potential treatments are generated from a binary threshold model $D(z_1, z_2) = \bm{1}(B_0 + B_1z_1 + B_2z_2  \geq 0)$, 
where $(B_0, B_1, B_2)$ represent individual heterogeneity. Further, suppose $B_j = \beta_jW + U_j,  j \in \{1, 2\}$, where $W \geq 0$, $\beta_j \geq 0$, $U_j \geq 0$, and $(W, U_1, U_2) \perp (Z_1, Z_2)$.  This DGP satisfies full independence, $((D(z_1, z_2))_{z_1, z_2 \in \{0, 1\}}, W) \perp Z$, so we can compare the performance of conditional (on $W$) and unconditional tests. The null hypothesis corresponds to $B_1, B_2 \geq_{a.s} 0$, and we construct a family of DGP's under the alternative hypothesis as follows. Fix $\delta, \gamma \geq 0$ with $\delta + \gamma \leq 1$, and consider a family of DGPs under which $B_1, B_2 \geq 0$ are positive for share $h$ of the population, $B_1 \geq 0$ and $B_2 < 0$ for share $\gamma(1-h)$, $B_2 < 0$ and $B_2 \geq 0$ for share $\delta(1-h)$, and both $B_1, B_2 < 0$ for the remaining share $(1-\gamma-\delta)(1-h)$. We take $Z_1, Z_2, W \sim \text{i.i.d.}\; \text{Bernoulli}(1/2)$, $B_0 \sim N(1, 1)$, $\beta_1 = \beta_2 = 1/2$, and $U_1, U_2 \sim \text{i.i.d.} \; U[0, 1]$. We conduct $5000$ simulations with sample size $200$, four different values of $(\delta, \gamma)$, and $h \in [0, 1]$.

Figure \ref{fig:pm_sim_results} depicts the power functions of two tests based on the inequalities derived in Theorem \ref{thm:testable_implication}. Since $W$ is discrete, we stack the corresponding conditional moment inequalities and use the GMS test of \citet{andrews2010inference}. We find that conditioning on the exogenous covariates $W$ substantially improves power, and that power changes dramatically across different alternatives. For example, in Panel (a), the power of the conditional test reaches 0.5 already when a third of the population violates monotonicity with respect to both $z_1$ and $z_2$. However, in Panel (d), even though 100\% of the population violates partial monotonicity with respect to $z_1$ or $z_2$, the power of the conditional test is still below 0.5. Since the test exhausts all information available in the data, these results imply that certain violations of the support restriction are fundamentally hard to detect.

\begin{figure}[!htbp]
	\centering 
    \begin{subfigure}{0.45\textwidth} \centering
		\includegraphics[width=\textwidth]{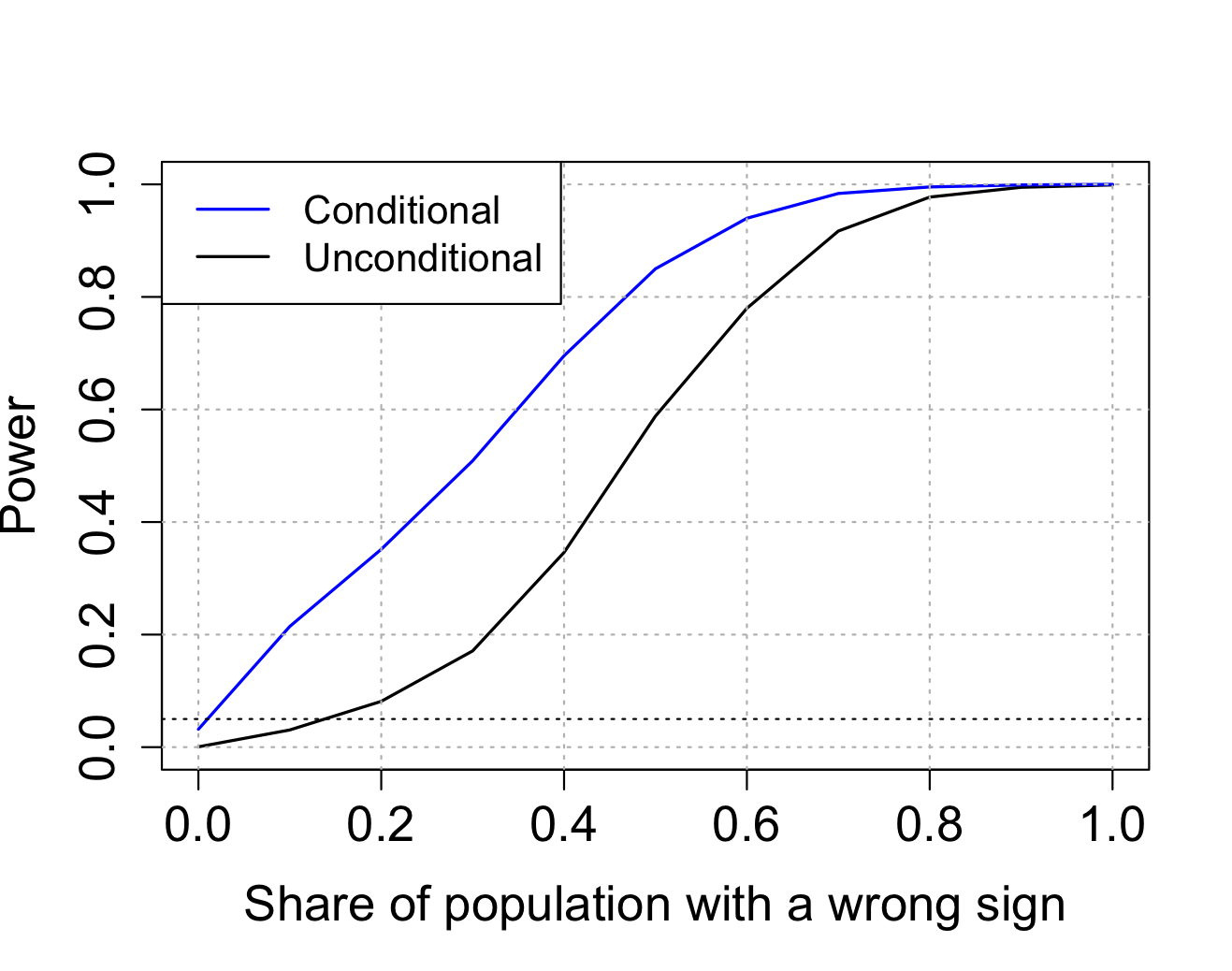}
		\subcaption{$\delta = 0, \gamma = 0$}
	\end{subfigure} \hfill
	\begin{subfigure}{0.45\textwidth} \centering
		\includegraphics[width=\textwidth]{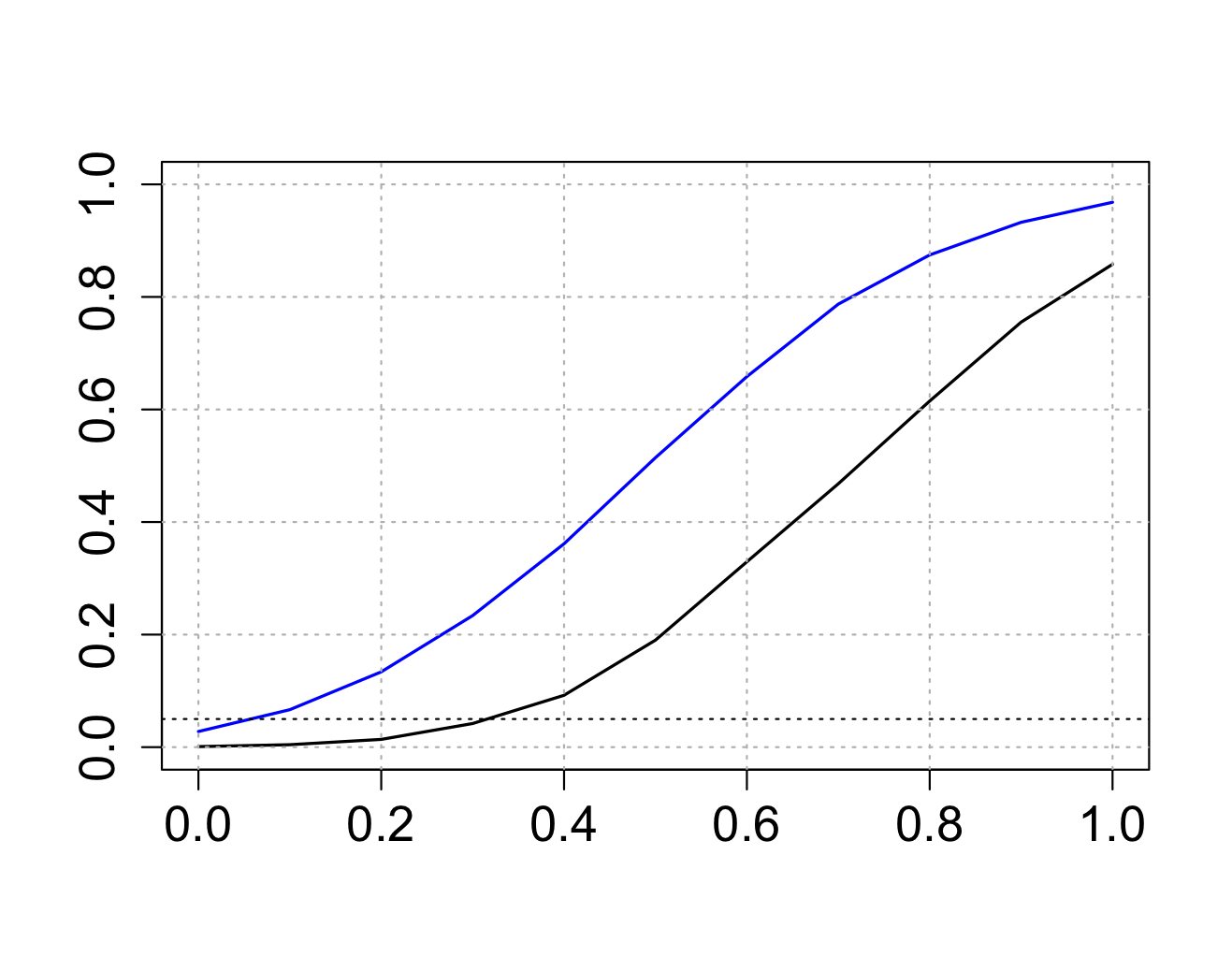}
		\subcaption{$\delta = 1, \gamma = 0$}
	\end{subfigure} 
	\begin{subfigure}{0.45\textwidth} \centering
		\includegraphics[width=\textwidth]{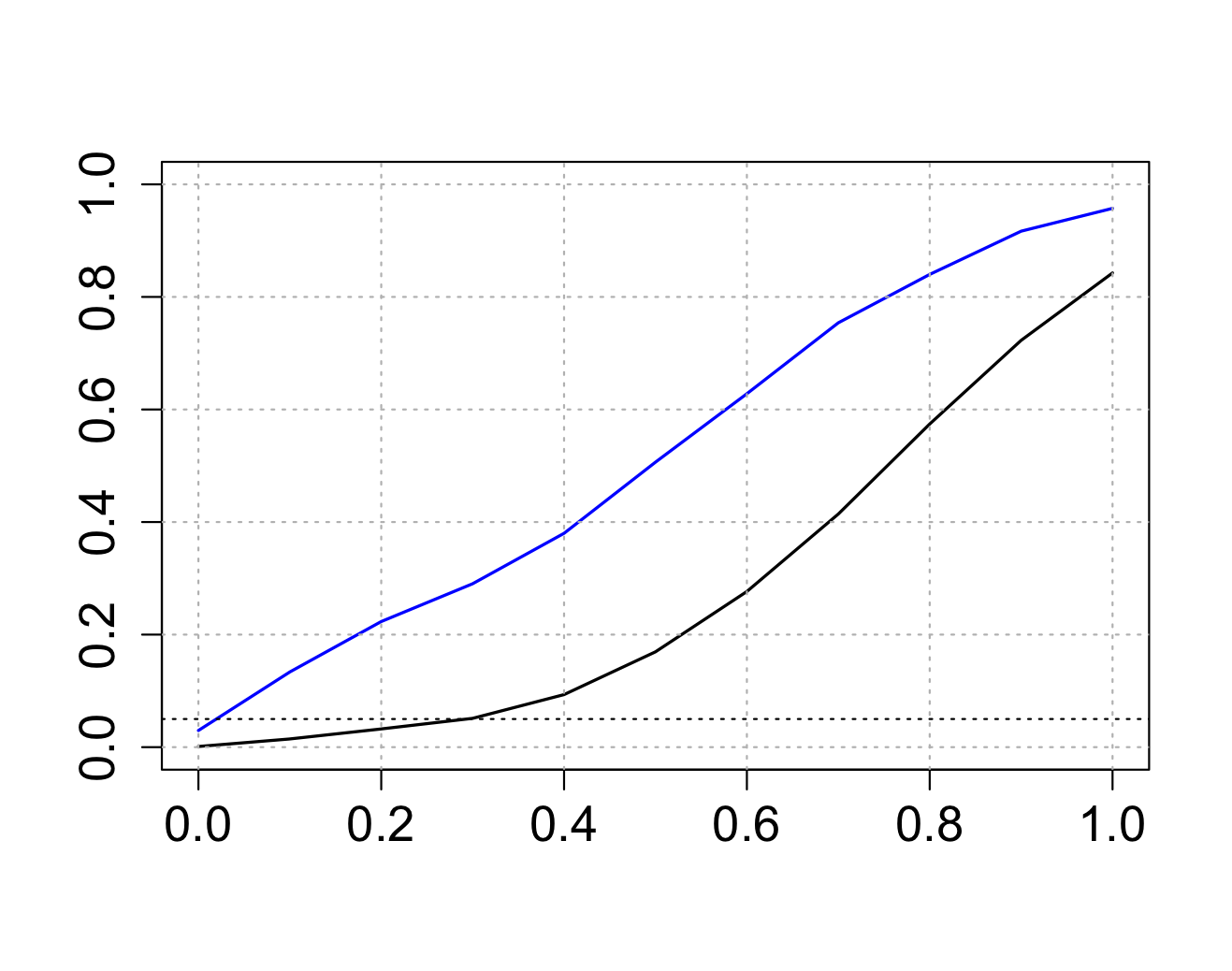}
		\subcaption{$\delta = 0.25, \gamma = 0.25$}
	\end{subfigure} \hfill
	\begin{subfigure}{0.45\textwidth} \centering
		\includegraphics[width=\textwidth]{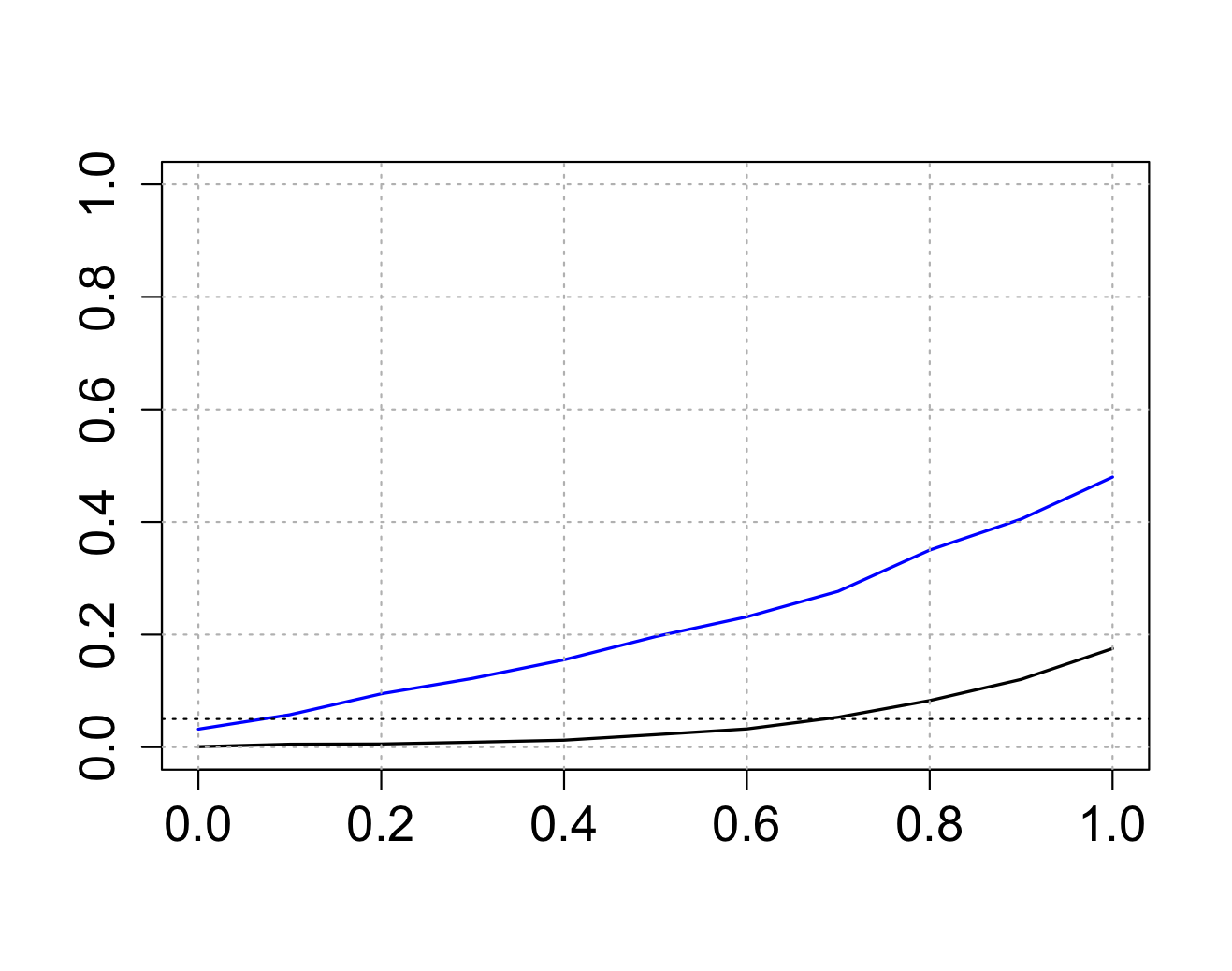}
		\subcaption{$\delta = 0.5, \gamma = 0.5$}
	\end{subfigure} 
	\caption{Power Functions of Conditional and Unconditional Tests in Example \ref{ex:partial_monotonicity}.} \label{fig:pm_sim_results}
\end{figure}
\begin{figure}[!htbp]
	\centering 
    \begin{subfigure}{0.45\textwidth} \centering
		\includegraphics[width=\textwidth]{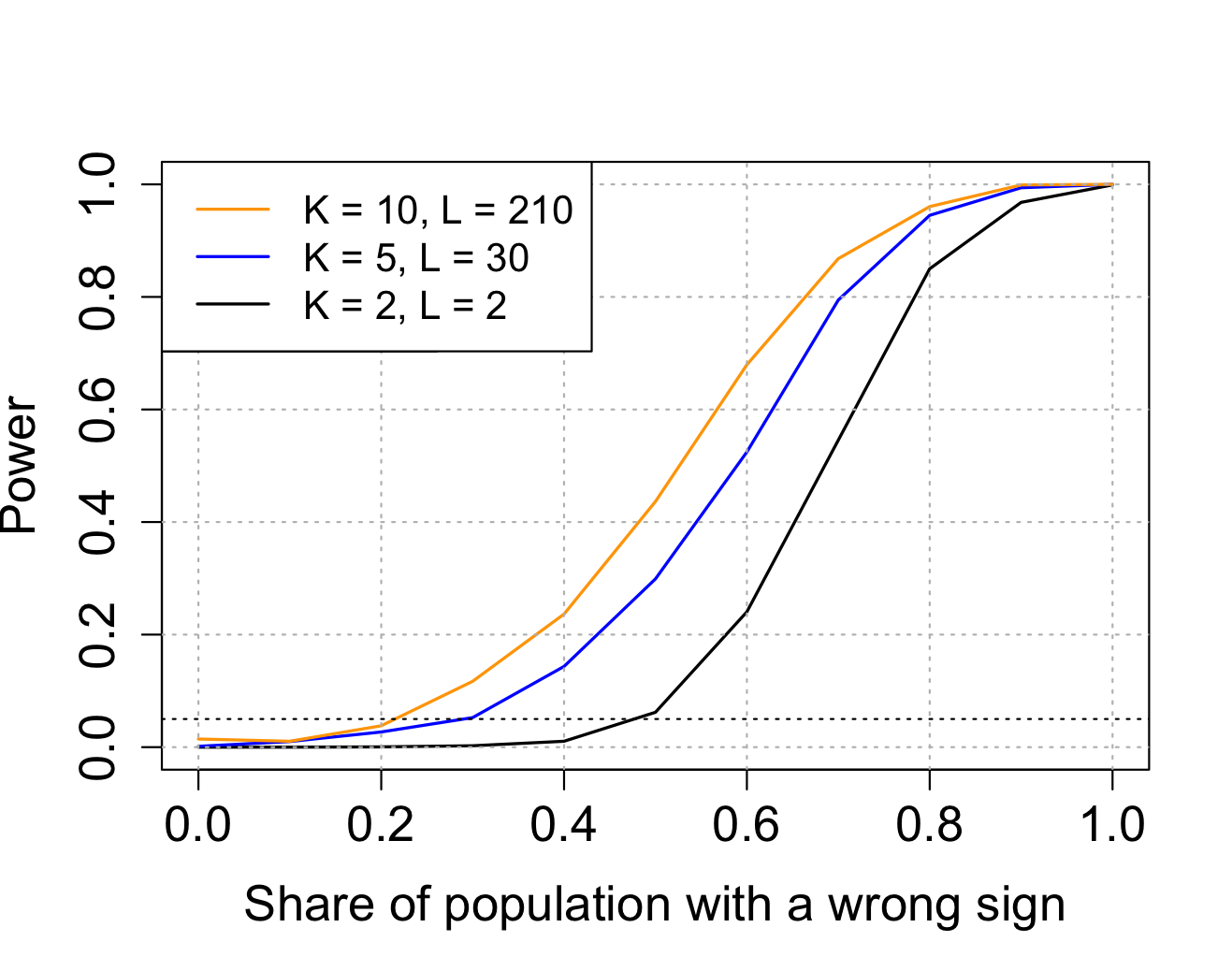}
		\subcaption{$J = 3$}
	\end{subfigure} \hfill
	\begin{subfigure}{0.45\textwidth} \centering
		\includegraphics[width=\textwidth]{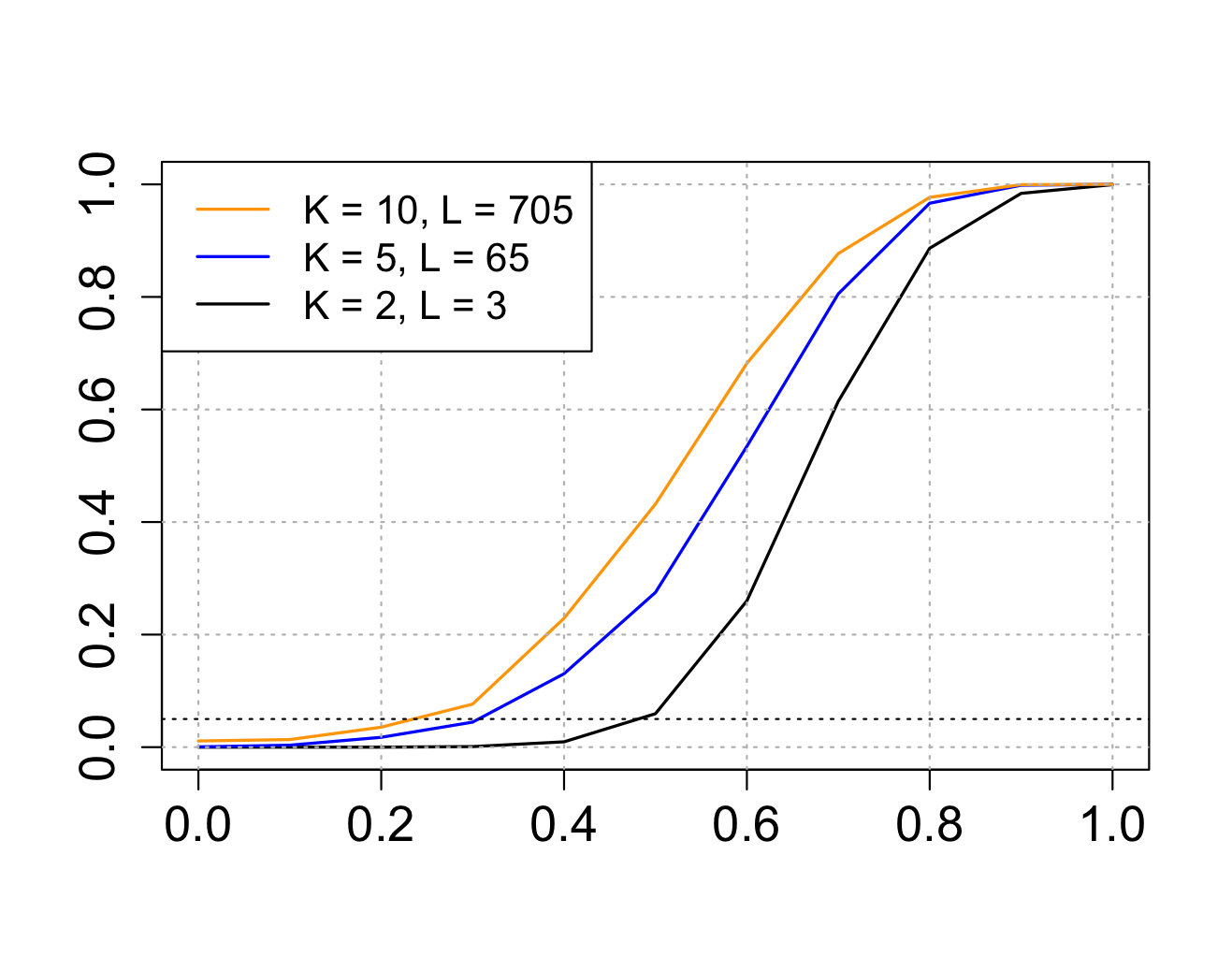}
		\subcaption{$J = 4$}
	\end{subfigure}\\ 
	\caption{Power Functions for Tests of Monotonicity.} \label{fig:m_sim_results}
\end{figure}

\subsection{Monotonicity}
Let $D = \{1, \dots, J\}$  and  $\mathcal{Z} = \{1, \dots, K\}$, and consider the null hypothesis
\[
D(z) \leq D(z'), \;\; \forall z \leq z'.
\]
This support restriction is regular (by verifying the conditions of Propositions \ref{L:helly_simple} and \ref{prop:comparability} or numerical checks), so the testable implications of Theorem \ref{thm:testable_implication} are sharp. 

Suppose the potential treatments are generated from the multiple thresholds model
\[
D(z) = \textstyle \sum_{j = 1}^{J} j\cdot \bm{1}\left( t_{j-1} \leq B_0 + \sum_{k = 1}^K B_k \bm{1}(z = k) \leq t_j \right)
\]
where $(B_0, B_1, \dots, B_K)$ represent individual heterogeneity, and $t_0, \dots, t_J$ is a set of thresholds. The null hypothesis corresponds to $B_k \leq_{a.s.} B_{k+1}$, for all $k \in\{ 1, \dots, K-1\}$. We focus on alternative DGPs such that, for each individual, monotonicity either holds for all $k$ or is violated for all $k$. Specifically, let $U_1, \dots, U_{K} \sim \text{i.i.d. } F_U$  and $U_{1:K}, \dots, U_{K:K}$ denote the order statistics ($U_{1:K}$ is the smallest). Let $U_0 \sim U[0, 1]$ be drawn independently of $U_1, \dots, U_K$, and  
\[
(B_1, \dots, B_K) = (U_{1:K}, \dots, U_{K:K})\cdot \bm{1}(U_0 \leq  h) + (U_{K:K}, \dots, U_{1:K})\cdot \bm{1}(U_0 > h).
\]

We set the thresholds $t_0, \dots, t_J$ uniformly on a grid within $[-1, 1]$ and $U_k \sim \exp(1)$. We conduct $S = 2000$ simulations with sample size $n = 200$, for $J \in \{3, 4\}, K \in \{2, 5, 10\}$ and $h \in [0, 1]$. Since the number of inequalities in some specifications exceeds the sample size, we use the test of \citet{bai2022two} with the pre-test level $\beta = \alpha/10$. Figure \ref{fig:m_sim_results} presents the results. We find that instruments with richer support lead to substantially more powerful tests, at the expense of testing many moment inequalities.

\section{Empirical Illustration} \label{sec:empirical}
As an empirical illustration, we study the testable implications of various support restrictions in a potential outcome model of smoking cessation interventions. The purpose of this exercise is threefold: to demonstrate that empirically relevant hypotheses can be formulated as support restrictions, to show that the proposed tests are computationally feasible in practice, and to illustrate how they can yield economically meaningful insights.

We use data from the US Lung Health Study (LHS), a randomized clinical trial that enrolled subjects aged 35 to 59 with early chronic obstructive pulmonary disease. The study assigned subjects to three treatment groups: control ($C$), intensive smoking cessation therapy with an inhaled bronchodilator ($SIA$), and the same therapy with a placebo bronchodilator ($SIP$). The therapy intervention consisted of physician-led counseling on the health consequences of smoking, along with orientation sessions involving family members and friends. The sample size is $n=5,887$.

\subsection{Interference}
The LHS followed subjects and recorded both their smoking status and that of their spouses for up to five years after the intervention. Using a linear regression framework, \citet{Fletcher:2017aa} document evidence of spillover effects of the SIA and SIP interventions on spouses’ smoking behavior. Motivated by these findings, we adopt the following formulation.

Let $i \in \{1,2\}$ index individuals, where $i=1$ denotes the subject and $i=2$ denotes the subject’s spouse. Let $Y_i$ be a binary outcome indicating whether individual $i$ smoked after the intervention. Let $z=(z_1,z_2)$ denote the vector of treatment assignments for individuals~1 and~2, and let $Y_i(z)$ denote individual $i$’s potential outcome under treatment assignment $z$.

For a given treatment assignment $z$, let $D_i(z)$ denote the level of exposure experienced by individual $i$. We define
\begin{align}
D_i(z)=
\begin{cases}
3 & \text{if } z_i = SIA,\\
2 & \text{if } z_i = SIP,\\
1 & \text{if } z_i = C,\; z_{-i} \in \{SIA, SIP\},\\
0 & \text{if } z_i = C,\; z_{-i} = C,
\end{cases}
\label{eq:exposure_map}
\end{align}
where $z_i$ denotes individual $i$’s treatment status and $z_{-i}$ denotes the treatment status of individual $i$’s spouse. Here, $D_i=3$ corresponds to receiving the SIA treatment, $D_i=2$ corresponds to receiving the SIP treatment, $D_i=1$ indicates exposure through the spouse receiving either SIA or SIP, and $D_i=0$ indicates no treatment exposure. For now, we maintain the exposure mapping assumption in \eqref{eq:exposure_map}, and we test this restriction in the next section.

Consider testing for the absence of spillover effects from a subject to their spouse.
 Using the exposure map, the restriction states\footnote{Under \eqref{eq:exposure}, we can define a potential outcome $\tilde Y_i(d)$ such that $\tilde Y_i(d)=Y_i(z)$ for any $z$ with $D_i(z)=d$. The no spillover effect restriction can also be expressed more succinctly as $\tilde Y_i(1)= \tilde Y_i(0), ~a.s.$ Similarly, \eqref{eq:negspill} is equivalent to $\tilde Y_i(1)\le \tilde Y_i(0), ~a.s.,$
which can be viewed as a monotone treatment response assumption on the exposure-based potential outcome. }
\begin{align}
D_i(z)=1,~ D_i(z')=0~\Rightarrow~Y_i(z)= Y_i(z').	\label{eq:nospill}
\end{align}
We also test a weaker restriction that allows non-positive spillover effects:
\begin{align}
D_i(z)=1,~ D_i(z')=0~\Rightarrow~Y_i(z)\leq Y_i(z').	\label{eq:negspill}
\end{align}
Applying Theorem \ref{thm:testable_implication} and conditioning on covariates $\X_i$, the testable implications of the no spillover restriction are
\begin{align}
P(Y_i=1\,|\,D_i=0,\X_i)+P(Y_i=0\,|\,D_i=1,\X_i)&\leq 1\\
P(Y_i=0\,|\,D_i=0,\X_i)+P(Y_i=1\,|\,D_i=1,\X_i)&\leq 1.	
\end{align}
Similarly, the testable implication of the non-positive spillover restriction is as follows.
\begin{align}
P(Y_i=0\,|\,D_i=0,\X_i)+P(Y_i=1\,|\,D_i=1,\X_i)\leq 1.	\label{eq:spillover_restriction}
\end{align}
We numerically verify that Assumption~\ref{A:support_regularity} holds for both restrictions. Hence, by Theorem~\ref{thm:sharpness_of_mis}, the model yields no additional testable implications.

We restrict the sample to married units with no missing observations, yielding 3,874 observations. Along with the exposure level $D_i$, we condition on survey wave $W_i \in \{1,2,3,4,5\}$. Since $W_i$ is discrete, we stack the conditional inequalities together and apply the tests of \citet{andrews2010inference} (AS).\footnote{As a comparison, we have also implemented the test of \cite{romano2014practical} and obtained the same conclusion.}

Panel (A) of Table~\ref{tab:empirical} reports the results. The test rejects the null hypothesis of no spillover effects at the 5\% significance level, but it does not reject the null of non-positive spillover effects. The test statistic for the latter restriction is zero, indicating that the sample satisfies the empirical analog of \eqref{eq:spillover_restriction}. This finding suggests that a subject’s participation in the smoking cessation program can reduce a spouse’s likelihood of smoking.  Hence, the smoking cessation programs may generate broader benefits while remaining cost-effective. One possible explanation is the spouse’s attendance at orientation sessions, as well as mutual support within the couple for quitting smoking and participating in the program.

Following \citet{Fletcher:2017aa}, we next include an extended set of covariates: age, sex, education, BMI, survey wave, and an indicator for the spouse’s baseline smoking status, treating age and BMI as continuous. To handle continuous covariates, we use the test of \citet{CLR13} (CLR).
The second row of Panel (A) of Table \ref{tab:empirical} shows that the CLR test rejects the no-spillover null because the precision-corrected intersection bound is strictly positive, but does not reject the non-positive spillover null. These findings confirm that the earlier conclusions are robust to controlling for additional covariates.

\newcolumntype{C}[1]{>{\centering\arraybackslash}p{#1}}

\begin{table}[t] 
\renewcommand{\arraystretch}{1.}
\centering
\begin{threeparttable}
\caption{Tests of spillovers, exposure mapping, and persistence of treatment effects.}
\label{tab:empirical}

\begin{tabular}{c*{5}{C{2.4cm}}}
\toprule
Panel & Test & Test Stat. & Crit. Value & Test Stat. & Crit. Value \\ \midrule

\multirow{4}{*}{(A)} 
 &  & \multicolumn{2}{c}{No Spillover Effect} & \multicolumn{2}{c}{Non-positive Spillover Effect} \\ 
 \cmidrule(lr){3-4} \cmidrule(lr){5-6}
 & AS   & 4.117 & 2.381 & 0.000 & 2.079 \\
 & CLR  & 0.090 & 0     & -0.003 & 0     \\ 
\cmidrule(lr){2-6}
 & \# Inequalities 
     & \multicolumn{2}{c}{2} 
     & \multicolumn{2}{c}{1} \\ \midrule

\multirow{4}{*}{ (B) } 
 &  & \multicolumn{2}{c}{Exposure Mapping (EM)} & \multicolumn{2}{c}{Semimonotonicity \& EM} \\ 
 \cmidrule(lr){3-4} \cmidrule(lr){5-6}
 & AS   & 0.728 & 2.392 & 1.236 & 2.998 \\
 & CLR  & -0.002 & 0    & -7e-5  & 0     \\ 
\cmidrule(lr){2-6}
 & \# Inequalities 
     & \multicolumn{2}{c}{2} 
     & \multicolumn{2}{c}{22} \\ \midrule

\multirow{4}{*}{ (C) } 
 &  & \multicolumn{2}{c}{No Effect on $L$} & \multicolumn{2}{c}{Monotone Effects on $L$} \\ 
 \cmidrule(lr){3-4} \cmidrule(lr){5-6}
 & AS   & 22.853 & 3.296 & 0.000 & 1.864 \\
 & CLR  & 0.331 & 0    & -0.009  & 0     \\ 
 & \# Inequalities 
     & \multicolumn{2}{c}{726} 
     & \multicolumn{2}{c}{25} \\
\bottomrule
\end{tabular}

\begin{tablenotes}
\small
\item \textit{Notes:} The table reports test statistics and critical values for AS and CLR tests. All tests are implemented at the 5\% significance level. The \# Inequalities corresponds to the number of conditional moment inequalities implied by each hypothesis. \end{tablenotes}
\end{threeparttable}
\end{table}

\subsection{Exposure Mapping and Semi-monotonicity}
So far, we have assumed that the exposure level determined the outcome. We now test this assumption.
 Let $Y=(Y_1,Y_2)$. 
By Theorems \ref{thm:testable_implication}-\ref{thm:sharpness_of_mis} (and ensuring Assumption \ref{A:support_regularity} numerically), the following inequalities are the sharp testable implications of the specified exposure mapping (EM) and imposing  \eqref{eq:exposure}:
\begin{multline}
P(Y=(0,1)\,|\,Z=(SIP,C),\X)+P(Y=(1,1)\,|\,Z=(SIP,C),\X)\\
+P(Y=(0,0)\,|\,Z=(SIA,C),\X)+P(Y=(1,0)\,|\,Z=(SIA,C),\X)\leq 1,\label{eq:exposure_ineq1}
\end{multline}
and
\begin{multline}
P(Y=(0,0)\,|\,Z=(SIP,C),\X)+P(Y=(1,0)\,|\,Z=(SIP,C),\X)\\
+P(Y=(0,1)\,|\,Z=(SIA,C),\X)+P(Y=(1,1)\,|\,Z=(SIA,C),\X)\leq 1.\label{eq:exposure_ineq2}	
\end{multline}
In addition to the exposure mapping assumption, we consider imposing the following semimonotonicity restriction:
\begin{align}
D_i(z)\geq D_i(z')~\Rightarrow~Y_i(z)\leq Y_i(z').
\end{align} 
This restriction simultaneously imposes (i) administering an inhaled bronchodilator (as opposed to a placebo) weakly reduces $Y_i$, (ii) receiving the therapy as a subject weakly reduces $Y_i$ compared to receiving it as a subject's spouse, and (iii) being a subject's spouse weakly reduces $Y_i$ relative to the baseline. Imposing them adds 20 conditional moment inequalities (for each $\x$) to \eqref{eq:exposure_ineq1}-\eqref{eq:exposure_ineq2}. %

Panel (B) of Table \ref{tab:empirical} summarizes the results of the AS and CLR tests. When implementing the AS test, we used a sample of married units and included only the survey wave as a conditioning variable. The AS test does not reject the exposure mapping assumption. Similarly, conditioning on the extended set of covariates, the CLR test does not reject the specified exposure mapping.
The joint hypothesis of semimonotonicity and the exposure mapping is not rejected by either test.

\begin{table}[t]
\centering
\begin{tabular}{lcc}
\toprule
Graph Properties        & No Effect on $L$ & Monotone Effects on $L$ \\
\midrule
Number of nodes          & 96 & 96    \\
Number of edges          & 1,026 & 2,049  \\
Number of support points & 4,682 & 12,494 \\
Number of MISs           & 726 & 25    \\
Support restriction regularity & \checkmark & \checkmark \\
\bottomrule
\end{tabular}
\caption{Summary of the Potential Response Graph $G$: Effects on Cessation Length}
\label{tab:graph_summary_cessation_length}
\end{table}

\subsection{Persistence of the Cessation Program}
Finally, we examine the persistence of the cessation program on the subject's smoking status. 
Let $Y(z)=(Y_1(z),\dots,Y_5(z))$ be the subject's potential outcome across the five survey waves at the treatment status $z\in \{C, SIP, SIA\}$. Let $L(z)$ denote the potential length of cessation, defined as the duration of the initial spell where $Y_t(z)=0$. Specifically, $L(z)=\max\{t\in\{0,1,\dots,5\}:Y_1(z)=Y_2(z)=\cdots=Y_t(z)=0\}$. 
We first consider the null hypothesis that treatment has no effect on cessation length:
\begin{align}
L(C)= L(SIP)= L(SIA).	\label{eq:no_L_effect}
\end{align} Next, we examine a monotonicity restriction that assumes more intensive treatments weakly extend cessation length:
\begin{align}
L(C)\leq L(SIP)\leq L(SIA).	\label{eq:monotone_L_effect}
\end{align}
The restrictions allow many possible support points for $Y(\cdot)$: 4,682 types satisfy \eqref{eq:no_L_effect} and 12,494 satisfy \eqref{eq:monotone_L_effect}, posing a computational challenge for the alternative approaches. Nevertheless, both satisfy Assumption \ref{A:support_regularity}, and verifying this took under 10 seconds using the algorithm of \cite{chudnovsky2020detecting},
despite the graph size.\footnote{The graph perfectness is checked by applying \texttt{is\_perfect()} in \texttt{SageMath}.} Thus, it suffices to check inequalities implied by the maximal independent sets --- 726 for \eqref{eq:no_L_effect} and 25 for \eqref{eq:monotone_L_effect} --- making the AS test computationally feasible. 

Panel (C) of Table \ref{tab:empirical} reports results based on 5,738 subjects. The AS test rejects the null hypothesis of no effect, while failing to reject the monotonicity restriction. These findings are consistent with the interpretation that the intervention prolongs smoking cessation and that its effect is (weakly) larger when combined with a medical component.

\section{Conclusion}
This paper develops a systematic approach for deriving sharp testable implications of exclusion and shape restrictions in potential outcome models. The proposed framework covers a broad class of restrictions studied in the literature and accommodates continuous outcomes as well as control variables. A result of independent interest is a simple necessary and sufficient condition for the existence of a joint distribution with a given finite collection of marginals and support. The class of distributions to which our sharp characterization applies can be further intersected with sets imposing additional distributional assumptions on potential outcomes. Investigating the resulting testable implications in such settings is a promising direction for future research.

\newpage

\appendix

\section{Proofs of the Main Results}

This appendix is organized as follows. Section \ref{AS:preliminaries} introduces relevant concepts from graph theory. Sections \ref{AS:proof_of_thm_1}--\ref{AS:proof_of_thm_2} provide proofs of our main results, Theorems \ref{thm:testable_implication}--\ref{thm:sharpness_of_mis}. Section \ref{AS:existence_of_joint} states a result on the existence of the joint distribution with given marginals and restricted support, formally equivalent to Theorem \ref{thm:sharpness_of_mis}. Section \ref{AS:verifying_regularity} proves Propositions \ref{L:helly_simple}--\ref{prop:tree} that can be used to verify the regularity condition, Assumption \ref{A:support_regularity} in practice. Finally, Section \ref{AS:proof_of_thm_3} gives a proof of Theorem \ref{thm:testable_implication_general_text}, which extends Theorem \ref{thm:sharpness_of_mis} to settings in which the outcome space is a Polish space.

\subsection{Preliminaries} \label{AS:preliminaries}

Let $G = (V_G, E_G)$ is a finite undirected graph with vertices $V_G$ and edges $E_G$. We adopt the following standard terminology.  The \textit{graph complement} of $G$ is a graph $\overline{G} = (V_G, \overline{E}_G)$ with $\overline{E}_G = \{(v, v'): (v, v') \notin E_G\}$. A graph $G$ is \textit{complete} if $E_G = \{(v, v'): v, v' \in V_G,  v \ne v'\}$. The \textit{subgraph of $G$ induced by a set of vertices} $C \subseteq V_G$, denoted $G[S]$, is an undirected graph with vertices $V_{G[C]} = C$ and edges $E_{G[C]} = \{(v, v') \in E_G: v, v' \in C\}$. A subset $C \subseteq V_G$ is a \textit{clique} if $G[C]$ is complete. A \textit{clique} $C$ is \textit{maximal} if for any $v \notin C$, the graph $G[C\cup \{v\}]$ is not complete.  A subset $I \subseteq V_G$ is an \textit{independent set} if $E_{G[I]} = \varnothing$. An \textit{independent set} $I$ is \textit{maximal} if for any $v \notin I$, $E_{G[I \cup \{v\}]} \ne \varnothing$. A graph $G$ is \textit{$K$-partite} if its' vertices can be partitioned as $V_G = \bigcup_{k = 1}^K V_k $ such that $V_k \cap V_l= \varnothing$ for all $k \ne l$ and each $V_k$ is an independent set. Let $\mathcal{C}_G$ denote the set of all maximal cliques and $\mathcal{I}_G$ the set of all maximal independent sets in $G$.

\subsection{Proof of Theorem \ref{thm:testable_implication}} \label{AS:proof_of_thm_1}
Denote $K = |\mathcal{Z}|$, $N = \sum_{z \in \mathcal{Z}} |\mathcal{R}_z|$ and $M = |\mathcal{R}^*|$. Let $A^* \in \{0, 1\}^{N \times M}$ be a binary matrix with rows indexed by response-input pairs $(r, z) \in \mathcal{R}_z \times \mathcal{Z}$, columns indexed by response types $r^* \in \mathcal{R}^*$, and elements $A^*_{(r, z), r^*} = \bm{1}(T(r^*, z) = r)$.  Define a polytope
\[
\overline{\mathbb{P}} = \{b \in \mathbf{R}^N: b = A^* Q, \text{ for some } Q \in \Delta(\mathcal{R}^*)\}.
\]
Assumptions \ref{A:unconf} and \ref{A:support} are equivalent to $\beta(P) \in \overline{\mathbb{P}}$. Further, define a bounded polyhedron
\[
\mathbb{P} = \{b \in \mathbf{R}^N_{+}: b'\bm{1}_{V_k} = 1, \text{ for all }k = 1, \dots, K, \; b'\bm{1}_I \leq 1, \text{ for all } I \in \mathcal{I}_G\}.
\]
Then, the statement of Theorem \ref{thm:testable_implication} is equivalent to $\overline{\mathbb{P}} \subseteq \mathbb{{P}}$, so we show this inclusion below.

Indeed, let $b = A^*Q$ be an arbitrary element of $\overline{\mathbb{P}}$. Then, for any $k \in \{1, \dots, K\}$,
\[
b'\bm{1}_{V_k} = Q'(A^*)'\bm{1}_{V_k} = \sum_{r^* \in \mathcal{R}^*} Q(r^*) \sum_{r \in \mathcal{R}_{z_k} }\bm{1}(T(r^*, z_k) = r) = \sum_{r^* \in \mathcal{R}^*} Q(r^*)  = 1.
\]
Further, by construction of the potential response graph $G$, each column of $A^*$ (equivalently, each $r^* \in \mathcal{R}^*$) corresponds to an indicator of a maximal clique $C$, denoted $\bm{1}_C$. Let $\mathcal{C}^*$ denote the set of all such maximal cliques. Thus, any $b \in \overline{\mathbb{P}}$ can be written as $b = \sum_{C \in \mathcal{C}^*} Q(C) \bm{1}_C$. Therefore, for any $I \in \mathcal{I}_G$,
\[
b'\bm{1}_I = \sum_{C \in \mathcal{C}^*} Q(C) |C \cap I| \leq \sum_{C \in \mathcal{C}^*} Q(C) = 1.
\]
Thus, $\overline{\mathbb{P}} \subseteq \mathbb{P}$ and the proof is complete. \qed 

\subsection{Proof of Theorem \ref{thm:sharpness_of_mis}} \label{AS:proof_of_thm_2}

Let $\overline{\mathbb{P}}$ and $\mathbb{P}$ be defined as in the proof of Theorem \ref{thm:testable_implication}. The MIS inequalities verify $\beta(P) \in \mathbb{P}$, so the goal is to prove that $\overline{\mathbb{P}} = \mathbb{P}$.

By construction, each response type $r^* \in \mathcal{R}^*$ corresponds to a maximal clique in $G$, which contains exactly one vertex from each part $V_k$, for $k \in \{1, \dots, K\}$. Letting $\mathcal{C}^*$ denote the set of all such maximal cliques, we have
\[
\overline{\mathbb{P}} = \text{ConvHull}(\{\mathbf{1}_C: C \in \mathcal{C}^*\}).
\]

First, we show that $\overline{\mathbb{P}}$ can be equivalently represented as
\begin{equation} \label{E:support_general}
	\overline{\mathbb{P}} = \{b \in \mathbf{R}^N_{+}: b'\bm{1}_{V_k} = 1, \text{ for all }k, \;\;h'b \leq \max_{C \in \mathcal{C}^*} h'\bm{1}_C,  \text{ for all } h \in \mathbf{Z}^N_{+}\}.
\end{equation}
Each element $b \in \overline{\mathbb{P}}$ satisfies $b \in \mathbf{R}^N_{+}$ and $b'\bm{1}_{V_k} = 1$ by construction. Moreover, since $\overline{\mathbb{P}}$ is convex, it can be characterized by the support-function inequalities
\begin{equation} \label{E:generic_support}
h'b \leq \max_{C \in \mathcal{C}^*} h'\bm{1}_C, \text{ for all } h \in \mathbf{R}^N.
\end{equation}
We show that it suffices to verify these inequalities for $h \in \mathbf{Z}^N_{+}$. Since $\overline{\mathbb{P}}$ is integral, it suffices to check $h \in \mathbf{Z}^N$.  Consider any such inequality with $h \in \mathbf{Z}^N$. For each $k$, let $\underline{h}_k = \min_{v \in V_k} h_v$ and consider $\tilde{h} = h - \sum_{k}\underline{h}_k \bm{1}_{V_k} \geq 0$. Then,
\[
\tilde{h}'b = h'b - \sum_{k = 1}^K\underline{h}_k
\]
and 
\[
\max_{C \in \mathcal{C}^*} \tilde{h}'\bm{1}_C =  \max_{C \in \mathcal{C}^*} (h'\bm{1}_C - \sum_{k = 1}^K \underline{h}_k \bm{1}_{V_k}'\bm{1}_C) = \max_{C \in \mathcal{C}^*} h'\bm{1}_C - \sum_{k = 1}^K \underline{h}_k,
\]
where the last equality holds because $\bm{1}_C'\bm{1}_{V_k} = 1$, for all $C \in \mathcal{C}^*$, holds by construction. As a result, the inequality $h'b \leq \max_{C \in \mathcal{C}^*} h'\bm{1}_C$ is equivalent to $\tilde{h}'\bm{1}_C \leq \max_{C \in \mathcal{C}^*} \tilde{h}'\bm{1}_C$. Thus, checking the support inequalities in \eqref{E:generic_support} for $h \in \mathbf{Z}^N_{+}$ is sufficient, and \eqref{E:support_general} is valid. 

Next, we show that under Assumption \ref{A:support_regularity}, $\overline{\mathbb{P}} = \mathbb{P}$. Let $b \in \mathbb{P}$. Let $h \in \mathbf{Z}^N_{+}$ and consider a graph $G^h$ obtained by making $h_v$ copies of each vertex $v$ of $G$, connecting them into a clique, and, for each edge $(u, v) \in E(G)$, connecting each copy of $u$ to all copies of $v$ in $G^h$. Since $G$ is perfect, by Theorem 1 in \citet{lovasz1972characterization}, $G^h$ is perfect as well. By construction, the largest clique of $G^h$ has size $\ell(h) =\max_{C \in \mathcal{C}_G} h'\bm{1}_C$. Thus, the vertices of $G^h$ can be partitioned into $\ell( h)$ independent sets, denoted $I_1(G^h), \dots, I_{\ell(h)}(G^h)$. Let $I_1, \dots, I_{\ell(h)}$ denote the projections of these sets back onto $G$. Each vertex $v$ in $G$ will be covered by the sets $I_l$ exactly $h_v$ times, one for each of the copies of $v$ that were created. As a result,
\[
h'b = \sum_{v \in V_G} b_v h_v = \sum_{v \in V_G} b_v \cdot |\{I_l: I_l \text{ contains a copy of }v\}| = \sum_{l = 1}^{\ell(h)} b'\bm{1}_{I_l} \leq \ell(h) = \max_{C \in \mathcal{C}_G} h'\bm{1}_C,
\]
where the inequality holds as a sum of MIS inequalities $b'\bm{1}_{I_l} \leq 1$ satisfied by $b$. Finally, since $\mathcal{C}^* = \mathcal{C}_G$, the right-hand side of the preceding display equals $\max_{C \in \mathcal{C}^*} h'\bm{1}_C$ and the general support inequality follows. Thus, under Assumption \ref{A:support_regularity}, $\overline{\mathbb{P}} = \mathbb{P}$. \qed

\subsection{Existence of a Joint Distribution with Given Marginals and Support} \label{AS:existence_of_joint}

Theorem \ref{thm:sharpness_of_mis} can also be viewed as a statement about the existence of the joint distribution with given marginals and support. 

\begin{theorem} \label{T:main}

  Let $\{P_z: z \in \mathcal{Z}\}$ be a finite collection of probability distributions with finite supports $\{\mathcal{R}_z: z \in \mathcal{Z}\}$, and $\mathcal{R}^* \subseteq \prod_{z \in \mathcal{Z}} \mathcal{R}_z$ a non-empty set with a typical element $r^* = (r^*_{z})_{z \in \mathcal{Z}}$.  Define an undirected graph $G = (V_G, E_G)$  with vertices $V_G = \bigcup_{z \in \mathcal{Z}}\{v_{r, z}: r \in \mathcal{R}_z\}$ and edges $E_G = \{(v_{r, z}, v_{r', z'}): \exists \, r^* \in \mathcal{R}^*: r^*_z = r, r^*_{z'} = r' \}$. Each $r^* \in \mathcal{R}^*$ induces a maximal clique of size $|\mathcal{Z}|$ in $G$, which contains exactly one point from each part $V_k$. Let $\mathcal{I}_G$ denote the set of all maximal independent sets in $G$. If $G$ is perfect, and $\mathcal{C}_G^*$ contains all maximal cliques of $G$, 
	\begin{equation} \label{E:main2}
	\sum_{v_{r, z} \in I} P_z(r) \leq 1, \;\;\; \forall\, I \in \mathcal{I}_G.
	\end{equation}
 is necessary and sufficient for the existence of the joint distribution $Q$ supported on $\mathcal{R}^*$ with marginals $\{P_z: z \in \mathcal{Z}\}$.
\end{theorem}

\begin{proof}
    The proof is identical to that of Theorem \ref{thm:sharpness_of_mis}. 
\end{proof}

\subsection{Sufficient Conditions for Regularity} \label{AS:verifying_regularity}

\subsubsection{Proof of Proposition \ref{L:helly_simple}}
Pick any tuple $\{(r_k, z_k)\}_{k = 1}^K$ such that $\bigcap_{k = 1}^K S(r_k, z_k) = \varnothing$. The goal is to show that there exists a collection $\mathcal{K}_0 \subseteq \{1, \dots, K\}$ with $|\mathcal{K}_0| \leq 2$ such that $\bigcap_{k \in \mathcal{K}_0} S(r_k, z_k) = \varnothing$. Suppose toward a contradiction that the minimal size of such collection is $|\mathcal{K}_0| \geq 3$. Then, for some distinct indices $k_1, k_2, k_3 \in \mathcal{K}_0$ there exist response types $u, v, w \in \mathcal{R}^*$ such that
	\[
	\begin{array}{lll}
			u[T(z_{k_1})] \ne r_{k_1}; \hspace*{1cm} &  u[T(z_{k_2})] = r_{k_2}; \hspace*{1cm}& u[T(z_{k_3})] = r_{k_3};\\
			v[T(z_{k_1})] = r_{k_1}; \hspace*{1cm}& v[T(z_{k_2})] \ne r_{k_2};\hspace*{1cm} & v[T(z_{k_3})] = r_{k_3};\\
			w[T(z_{k_1})] = r_{k_1}; \hspace*{1cm}&  w[T(z_{k_2})] = r_{k_2}; \hspace*{1cm} & w[T(z_{k_3})] \ne r_{k_3}.
	\end{array}
	\]
	Consider a vector $r^*$ defined componentwise as $	r^*_i = \text{median}(u_i, v_i, w_i).$ Since $\mathcal{R}^*$ is defined by constraints featuring only two coordinates at a time, we have $r^* \in \mathcal{R}^*$. Indeed, by the properties of the median,
	\[
	\begin{array}{l}
			ar_i^* = \text{median}(a u_i, a v_i, a w_i)\\
			br_j^* = \text{median}(b u_j, b v_j, b w_j)
	\end{array}
	\]
	Thus, among the three vectors $x \in \{u, v, w\}$, at least two satisfy $a_ix_i \geq a r_i^*$ and at least two satisfy $b_j x_j \geq b r_j^*$. Thus, there must be at least one  $x \in \{u, v, w\}$ that satisfies both, and therefore we have
	\[
	ar_i^* + br_j^* \leq ax_i + bx_j \leq c.
	\]	
	Since, for each $k_j$, two out of three vectors $\{u, v, w\}$ agree on the value of $r_{k_j}$, we have $	r^*[T(z_{k_j})] = r_{k_j}, \text{ for } j \in \{1, 2, 3\},$ which means $r ^* \in \bigcap_{k \in \mathcal{K}_0} S(r_k, z_k) = \varnothing$, a contradiction. Thus, the minimal size of $\mathcal{K}_0$ is at most two, and the Helly property is satisfied. \qed 

\subsubsection{Proof of Proposition \ref{prop:comparability}}
Comparability graphs of strict partial orders are perfect by \citep[Theorem 5.34 in][]{golumbic2004algorithmic}. \qed 

\subsubsection{Proof of Proposition \ref{prop:tree}}
The graph $G$ is the intersection graph of the subtrees $S(r,z)$. By the subtree-intersection characterization of \cite{gavril1974}, it is chordal. Chordal graphs are perfect. \qed

\subsection{Proof of Theorem \ref{thm:testable_implication_general_text}} \label{AS:proof_of_thm_3}

First, we introduce a few objects. Let $N_n=nK$. Fix $r^*\in \cR^*$ and $n$. Let $r^*_n\in\{0,1\}^{N_n}$ denote a support point vector that stacks $(\bm{1}(T(r^*,k)\in A_{s,k})_{s\in \{1,\dots,n\}})_{k\in \cZ}$. Let $\cR^*_n\subset\{0,1\}^{N_n}$ be the collection of all such support points. Let $\pi_n:\cR^*\to\cR^*_n$ denote the function that maps $r^*$ to $r^*_n$.

\textbf{Step 1: Necessity.}  Let $Q$ be a distribution of $R^*=R(\cdot)$ over $\cR^*$. For each $n$, let $Q_n=Q\circ \pi_n^{-1}$ be the induced distribution over $\cR^*_n$. $G_n$ is a potential response graph associated with $\cR^*_n$. By Theorem \ref{thm:testable_implication}, the existence of $Q_n$ implies \eqref{eq:MIS_inequalities_n_text}. These implications hold for any $n\in\mathbb N$.

\textbf{Step 2: Sufficiency.} Suppose Assumption \ref{A:support_regularity} holds for $G_n$ for all $n$ and \eqref{eq:MIS_inequalities_n_text} holds. Then, Theorem \ref{thm:sharpness_of_mis} ensures that there exists a sequence $\{Q_n\}$ of distributions such that each $Q_n$ is supported on $\cR^*_n$. Furthermore, one can ensure that such a sequence has a consistency property.
For any $m<n$, let $\pi^n_m:\cR^*_n\to\cR^*_m$ map $r^*_n=(\bm{1}(T(r^*,k)\in A_{s,k})_{s\in \{1,\dots,n\}})_{k\in \cZ}$ to $r^*_m=(\bm{1}(T(r^*,k)\in A_{s,k})_{s\in \{1,\dots,m\}})_{k\in \cZ}$. Then,
\begin{align}
\pi^n_\ell=\pi^m_\ell\circ \pi^n_m,~\ell\leq m\leq n.	\label{eq:consistency}
\end{align}
Note that the system is consistent in the sense that, for any distribution $Q_{n}$ on $\cR^*_n$, one can construct $Q_\ell$ so that
\begin{align}
Q_n\circ(\pi^n_\ell)^{-1}=Q_\ell,	~\ell\leq n\in \mathbb N.
\end{align}
This is because any support point in $\mathcal R^*_\ell$ can be derived from the support points in $\mathcal R^*_n$ since $\cA_{n,k}$ is a refinement of $\cA_{\ell,k}$, for all $k$. For such a system, we may apply an extension theorem, Corollary 8.22 in \cite{kallenberg2021foundations}. To apply this result, we observe that $\cR^*$ is a Borel space, which is ensured by Theorem 1.8 in \cite{kallenberg2021foundations} and $\cR^*$ being a Polish space. Similarly, $\cR^*_1,\cR^*_2,\dots$ are also Borel spaces when each $\cR^*_n$ is equipped with the discrete topology. The maps $\pi^n:\cR^*\to\cR^*_n$ and $\pi^n_\ell:\cR^*_n\to\cR^*_\ell,\ell\le n$ are measurable and satisfy \eqref{eq:consistency}. Then,
Corollary 8.22 in \cite{kallenberg2021foundations} ensures that there exists a distribution $Q$ on $\cR^*$ such that $Q\circ \pi_n^{-1}=Q_n$ for all $n$. \qed

\section{Extensions and Applications}

This appendix is organized as follows. Section \ref{ssec:testing_exclusion_monotonicity} applies Theorem \ref{thm:testable_implication_general_text} to the IV model with exclusion and first-stage monotonicity, and obtains a concrete characterization of the resulting testable implications. Section \ref{ssec:comparison} compares testable implications of multiple models, provides a proof for Proposition \ref{prop:comparison}, and gives further examples. Section \ref{ssec:pm_random_coeff} describes a random coefficient model for the partial monotonicity restriction and illustrates the resulting partition of the latent variable space. Section \ref{ssec:existing_inequalities} discusses the relationship between the proposed approach and existing approaches based on Artstein's inequalities for binary instruments and generalized IV inequalities of \citet{KedagniMourifie2020}. Section \ref{ssec:rank_perfect_proof} provides a proof of Theorem \ref{thm:rank_perfect}. Section \ref{app:encouragement} shows that the encouragement design model from \citet{BaiHuangTabordMeehan2026} and a discrete choice model with quasi-linear utility and increasing differences in costs satisfy the regularity conditions imposed in Assumption \ref{A:support_regularity}. 

\subsection{Exclusion and Monotonicity Restrictions with Multi-valued IVs}\label{ssec:testing_exclusion_monotonicity}
Let $Y \in \mathcal{Y}$, $D \in \mathcal{D}$ and $Z \in \mathcal{Z}$, and consider testing the following assumptions:
\begin{equation} \label{E:exclusion}
\begin{array}{l}
	Z \perp (\{Y(d, z)\}_{d \in \mathcal{D}, z \in \mathcal{Z}}, \{D(z)\}_{z \in \mathcal{Z}});\\[3mm]
	Y(d, z) =_{a.s.} Y(d, z'), \;\; \forall d \in \mathcal{D}, \forall z, z' \in \mathcal{Z}.
\end{array}
\end{equation}

Consider any non-empty disjoint sets $B_1, \dots, B_L \subseteq \mathcal{Y}$, any $d \in \mathcal{D}$, and any $z_1,\dots, z_L \in \mathcal{Z}$. Under the exclusion restriction, it cannot happen that $D(z_\ell) = d$ yet $Y(d) \in B_\ell$, for all $\ell$. Combined with the independence assumption, this implies that the inequality
\begin{equation} 
\sum_{\ell = 1}^L P(Y \in B_\ell, D = d \,|\,Z = z_\ell) \leq 1,
\end{equation}
must hold. If some values of $z_\ell$ are repeated, the corresponding sets $B_\ell$ can be combined without changing the inequality. Thus, it suffices to consider distinct values $z_1, \dots, z_L \in \mathcal{Z}$, with $L \leq K$.  Also, the inequality is the tightest when $B_1, \dots, B_L$ form a partition of $\mathcal{Y}$. 

The following proposition characterizes a testable implication of \eqref{E:exclusion}.
\begin{proposition}[Testing Exclusion]\label{prop:exclusion_necessity} Suppose Condition \eqref{E:exclusion} holds. Then,
\[
\begin{array}{c}
\sum_{\ell = 1}^L P(Y \in B_\ell, D = d \,|\,Z = z_\ell) \leq 1, \;\;\;  \forall d \in \mathcal{D}, 	\\[2mm]
\forall \{B_1, \dots, B_L\}: B_k \cap B_\ell = \varnothing,\, \bigcup_{\ell = 1}^L B_\ell = \mathcal{Y}, \;\; \forall \{z_1, \dots, z_L\} \subseteq \mathcal{Z}, z_k \ne z_\ell, \;\;  \forall L \leq K.
\end{array}
\]
\end{proposition}

\begin{proof}
Define the sequence of potential response graphs $G_n$ associated with the exclusion restriction as in Theorem \ref{thm:testable_implication_general_text}. Index the vertices of $G_n$ by tuples $(A_{i, z}, d, z)$. Let $\mathcal{I}_{G_n}$ be the collection of all maximal independent sets of $G_n$. We start by observing that each  $I\in \mathcal I(G_n)$ (that is not one of the parts of $G_n$) takes the form
\begin{align}
	\{(A_{s_1,z_1}, d, z_1), (A_{s_2,z_2}, d, z_2), \dots, (A_{s_{m_n}, z_{m_n}}, d, z_{m_n})\}, \label{eq:MIS2}
\end{align}
for some $d\in \mathcal{D}$, $m_n\in\mathbb N$, $(z_1, \dots, z_{m_n})\in \cZ^{m_n}$, and mutually disjoint sets $A_{s_i,z_i}, i\in\{1,  \dots, m_n\}$ such that for any $(A_{s, z}, d, z)$ excluded from  \eqref{eq:MIS2}, $A_{s, z} \cap \bigcup_{j=1}^{m_n} A_{s_j, z_j} \ne \varnothing$.

Indeed, consider vertices $(A, d, z)$ and $(A', d', z')$ such that $z \ne z'$. By the exclusion restriction, the two vertices are independent if and only if $A\cap A'=\varnothing$ and $d=d'$. Hence, if an independent set contains vertices with different $z$ values, their $d$ values must coincide. On the other hand, if $(A, d, z)$ and $(A', d', z')$ are independent and $d\ne d'$, then it must be the case that $z=z'$. Thus, all independent sets must take the form \eqref{eq:MIS2}.   To see that the set $I$ in \eqref{eq:MIS2} is maximal, consider adding another vertex $(A_{s, z}, d, z)$ to $I$. Since $A_{s, z} \cap \bigcup_{j=1}^{m_n} A_{s_j, z_j} \ne \varnothing$, this vertex would necessarily have an edge with some vertices already included in $I$. 

By Theorem \ref{thm:testable_implication_general_text}, the inequalities
\begin{equation} \label{eq:exclusion_implication}
\sum_{i = 1}^{m_n} P(Y \in A_{s_i, z_i}, D = d \,|\,Z = z_i) \leq 1
\end{equation}
corresponding to the independent sets in \eqref{eq:MIS2} are necessary for the existence of a joint distribution satisfying condition \eqref{E:exclusion}.

Finally, note that the inequalities in \eqref{eq:exclusion_implication} are implied by those in the statement of the proposition.  Indeed, for each $I\in \mathcal I(G_n)$, let $B_k= \bigcup \{A_{s_i,z_i}:(A_{s_i,z_i},d,z_i)\in I, z_i=z_k\}$, for $k \in \{1, \dots, K\}$, allowing $B_k= \varnothing$. Let $\ell \in \{1, \dots, L\}$, for $L \leq K$, index $B_\ell \ne \varnothing$. Then, we can equivalently write \eqref{eq:exclusion_implication} as
\begin{align}
	\sum_{\ell = 1}^L P(Y\in B_\ell, D=d \mid Z = z_\ell)\leq 1.
\end{align}
These inequalities are the tightest when $B_1, \dots, B_L$ form a partition of $\mathcal{Y}$, which concludes the proof. 
\end{proof}

\subsubsection{Proof of Proposition \ref{prop:iv_val}}

As discussed in the text, we may characterize the sharp testable implications of the exclusion and monotonicity restrictions. The proof of Proposition \ref{prop:iv_val} is as follows.

\begin{proof}
For each $n\in\mathbb N$ and $k\in\{1,\dots,K\}$, let $\cA_{n,k}=\{A_1,\dots,A_n\}$ be a measurable partition of $\cY$. 
Define the sequence of potential response graphs $G_n$ associated with the exclusion and monotonicity restriction as in Theorem \ref{thm:testable_implication_general_text}. Index the vertices of $G_n$ by tuples $(A, d, z)$. Let $\mathcal{I}_{G_n}$ be the collection of all maximal independent sets of $G_n$. We start by observing that each  $I\in \mathcal I(G_n)$ (that is not one of the parts of $G_n$) takes the form
\begin{align}
\{\{(A_{\ell,d},d,z_\ell)\}_{\ell\in S_d}\}_{d\in\cD},~\label{eq:LATE_MIS}
\end{align}
for some mutually disjoint sets $\{A_{\ell,d},~\ell\in S_d\}$ such that $A_{\ell,d}\in \{A_1,\dots,A_n\}$ for all $\ell$.  For any $(A, d, z)$ excluded from  \eqref{eq:LATE_MIS}, $A \cap  \bigcup_{\ell\in S_d}A_{\ell,d}\ne \varnothing$, and $\forall d \in \mathcal{D},\; d > d' \implies \overline{\ell}_d \leq \underline{\ell}_{d'}.$ 

Indeed, consider vertices $(A,d,z)$ and $(A',d',z')$. The two vertices are independent if $A\cap A'=\emptyset$ and $d=d'$ due to the exclusion restriction. In addition, the vertices are independent if $z\leq z'$ with $d>d'$ due to the violation of the monotonicity. Finally, for $z=z'$, the two distinct vertices are always independent. Therefore, we must have $\overline{\ell}_d \leq \underline{\ell}_{d'}$ for any pair such that $d>d'$.  To see that the set in \eqref{eq:LATE_MIS} is maximal, consider adding another vertex $(A,d,z)$ to this set. Such a vertex would generate an edge with some vertex in the set because of $A \cap  \bigcup_{\ell\in S_d}A_{\ell,d}\ne \varnothing$.

By Theorem \ref{thm:testable_implication_general_text}, the inequalities
\begin{equation} 
\sum_{\ell\in S_d,d\in \cD} P(Y \in A_{\ell,d}, D = d \,|\,Z = z_\ell) \leq 1,\label{eq:iv_implication}
\end{equation}
corresponding to the independent sets in \eqref{eq:LATE_MIS} are necessary for the existence of a joint distribution satisfying the exclusion and monotonicity conditions. Below we show that, for each $n$, the exclusion and monotonicity restrictions are regular.

First, we verify that, for each $n$, the potential response graph $G_n$ is perfect. This result follows from showing that the exclusion and monotonicity restrictions define a strict partial order between the vertices of $G_n$ (Lemma \ref{lem:spo}). Lemma \ref{lemma:IV_perfectness} then shows the perfectness of $G_n$ by leveraging the fact that $G_n$ is a comparability graph based on this partial order, which is a subclass of the perfect graphs. Then, $G_n$ satisfies Assumption 3 (ii).

Second, we verify that every maximal clique in $G_n$ corresponds to a support point $r^*$. 
Consider any such clique $\{(A_\ell, d_\ell, z_\ell)\}_{\ell = 1}^L$, where all $z_\ell$ are distinct, listed in the ascending order, and $L \leq K$. Vertices $(A_\ell, d_\ell, z_\ell)$ and $(A_{\ell'}, d_{\ell'}, z_{\ell'})$ with $z_{\ell} \ne z_{\ell'}$ are connected if and only if (i)  $d_{\ell} < d_{\ell'}$ or (ii) $d_{\ell} = d_{\ell'}$ and $A_\ell = A_{\ell'}$.  Since $G_n$ is a comparability graph, $(A_1,d_1,z_1)\prec (A_2,d_2,z_2)\prec\dots\prec (A_L,d_L,z_L)$  where $\prec$ is the partial order in \eqref{eq:def_spo}.  If $L < K$,
one can insert a vertex $v=(A, d, z)$ with $z \notin \{z_1, \dots, z_L\}$ to the clique and achieve either (a) $v\prec (A_1,d_1,z_1) \prec\dots\prec (A_L,d_L,z_L)$, (b) $(A_1,d_1,z_1)\prec \dots \prec v \prec\dots\prec (A_L,d_L,z_L)$, or (c) $ (A_1,d_1,z_1) \prec\dots\prec (A_L,d_L,z_L)\prec v$. Hence, one can form a larger clique by this operation.
Thus, every maximal clique must be of size $K$. Moreover, the intersection of events $ \bigcap_{k = 1}^K\{Y(d_k)\in A_k, D(z_k) = d_k\}$ does not contradict the exclusion and monotonicity restrictions, meaning that there exists a support point corresponding to every such clique. 

Finally, note that the inequalities in \eqref{eq:iv_implication} are implied by those in the statement of the proposition.  Indeed, for each $I\in \mathcal I(G_n)$, let $B_{\ell,d}= \bigcup \{A:(A,d,z)\in I, z=z_\ell\}$ for $\ell \in \{1, \dots, K\}$, allowing $B_{\ell,d}= \varnothing$. Then, we can equivalently write \eqref{eq:iv_implication} as
\begin{align}
	\sum_{\ell \in S_d, d \in \mathcal{D}} P(Y \in B_{\ell, d}, D = d \,|\,Z = z_\ell) \leq 1.
\end{align}
These inequalities are the tightest when $B_{\ell,d}, \ell\in S_d$ form a partition of $\mathcal{Y}$, which concludes the proof. \end{proof}

Next we show the two lemmas used in the above proof.

\begin{lemma}[Strict Partial Order on $V_G$]\label{lem:spo}
Let $\mathcal{A} = \{A_1, \dots, A_n\}$ be a measurable partition of $\mathcal{Y}$.
Let $\cD$ be a finite totally ordered set, and let $\cZ=\{z_1,\dots,z_K\}$.
Define $V_G = \mathcal{A} \times \mathcal{D} \times\cZ.$
Define a binary relation $\prec$ on $V_G$ by:
\begin{align}
(A,d,z) \prec (A',d',z')~ \Leftrightarrow ~
\begin{cases}
\text{(i)}\ z < z' \text{ and } d < d', \\
\text{or } \text{(ii)}\ z < z',\ d = d',\ A = A'.
\end{cases}\label{eq:def_spo}
\end{align}
Then $\prec$ is a strict partial order on $V_G$.
\end{lemma}

\begin{proof}
We verify the three properties:

\begin{enumerate}
    \item[(i)] Irreflexivity: For any $(A,d,z)$, neither condition (i) nor (ii) can hold with $z = z'$, so $(A,d,z) \not\prec (A,d,z)$.

    \item[(ii)] Asymmetry: If $(A,d,z) \prec (A',d',z')$, then $z < z'$. It cannot be that $z' < z$. Thus, $(A',d',z') \prec (A,d,z)$ cannot hold.

    \item[(iii)] Transitivity:
Let $u=(A_1,d_1,z_1)$, $v=(A_2,d_2,z_2)$, and $w=(A_3,d_3,z_3)$. Suppose $u\prec v$ and $v\prec w$. We analyze the four possible combinations:
    \begin{itemize}
        \item[(a)] Both comparisons use (i): $z_1 < z_2 < z_3$ and $d_1 < d_2 < d_3\Rightarrow z_1 < z_3$, $d_1 < d_3\Rightarrow u \prec w$ by (i).

        \item[(b)] Both comparisons use (ii): $z_1 < z_2 < z_3$, $d_1 = d_2 = d_3$, and $A_1 = A_2 = A_3\Rightarrow u \prec w$ by (ii).

        \item[(c)] The first comparison uses (ii), and the second one uses (i): $z_1 < z_2$, $d_1 = d_2$, $A_1 = A_2$, and $z_2 < z_3$, $d_2 < d_3$. This implies $d_1 < d_3$ and $z_1 < z_3$, Hence, $u \prec w$ by (i).

        \item[(d)] The first comparison uses (i), and the second one uses (ii): This case is essentially the same as case (c).
    \end{itemize}
\end{enumerate}
Thus, $\prec$ is a strict partial order.
\end{proof}

\begin{lemma}[Perfectness of $G$ associated with Exclusion and Monotonicity restrictions]\label{lemma:IV_perfectness}
Let $G=(V_G,E_G)$ be a potential response graph such that $V_G = \mathcal{A} \times \mathcal{D} \times \{z_1,\dots,z_K\}$ and 
$E_G$ consists of edges between $u=(A,d,z)$ and $v=(A',d',z')$ such that either of the following conditions hold:
\begin{itemize}
	\item[(a)] $z < z'$, $d = d'$, $A = A'$, or $z > z'$, $d = d'$, $A = A'$;
	\item[(b)] $z < z'$, $d < d'$ or $z > z'$ and $d > d'$.
\end{itemize}
Then, $G$ is a perfect graph.
\end{lemma}

\begin{proof}
 By the definition of $G$, an edge exists between $u,v\in V_G$ if and only if:
\begin{align}
u \prec v \quad \text{or} \quad v \prec u,
\end{align}
where $\prec$ is the binary relation in \eqref{eq:def_spo}. By  Lemma \ref{lem:spo}, $\prec$ is a strict partial order. Hence, $G$ is the comparability graph of the strict partial order $(V_G, \prec)$. Comparability graphs of strict partial orders are perfect \citep[][Theorem 5.34]{golumbic2004algorithmic}, which ensures the claim.
\end{proof}

\begin{corollary}\label{cor:iv_K2}
Suppose $\cD=\{0,1\}$ and $\cZ=\{0,1\}$. Let $A\subset\cY$. Then, the inequalities
\begin{align}
	P(Y\in A,D=1|Z=0)&\leq P(Y\in A,D=1|Z=1)\label{eq:bphvk1}\\
P(Y\in A,D=0|Z=1)&\leq P(Y\in A,D=0|Z=0)\label{eq:bphvk2}	
\end{align}
follow from \eqref{eq:ex_mon_implication_K}. 
\end{corollary}

\begin{proof}
 Let $\cZ=\{z_1,z_2\}=\{0,1\}.$ We specialize Proposition \ref{prop:iv_val} as follows. Let $S_0=\{2\},S_1=\{1,2\}$.  Let $A\subseteq \cY$ and define the tuple $\{(\cY,0,z_2),(A,1,z_1),(A^c,1,z_2)\}$. This construction satisfies $\cup_{\ell\in S_d}B_{\ell,d}=\cY$ for $d\in \{0,1\}$, and $\underline{\ell}_0=2= \overline{\ell}_1$. Now,  \eqref{eq:ex_mon_implication_K} can be expressed as
\begin{align}
P(Y\in \cY,D=0|Z=z_2)+P(Y\in A,D=1|Z=z_1)+P(Y\in A^c,D=1|Z=z_2)\leq 1,	
\end{align}
which is equivalent to \eqref{eq:bphvk1}. One can obtain \eqref{eq:bphvk2} from \eqref{eq:ex_mon_implication_K} by a similar argument by taking $S_0=\{1,2\},S_1=\{1\}$ and letting $\{(A,0,z_1),(A^c,0,z_2),(\cY,1,z_1)\}$ be a tuple in Proposition \ref{prop:iv_val}.
\end{proof}

\subsection{Comparing Testable Implications of Multiple Models}\label{ssec:comparison}
We start with the following lemma. 
\begin{lemma}\label{lem:MIS_expansion}
Let $G=(V_G,E_G)$ be an undirected graph. Let $H\subseteq V_G.$ Let $G'$ be a graph that is obtained by removing all edges between the vertices belonging to $H$, i.e., $E_{G'} = E_G \setminus \{(u,v) \in E_G:u,v\in H\}.$
Then,  for any $I’ \in \mathcal{I}_{G’}$, exactly one of the following holds: 
\begin{enumerate}
	\item[(i)] (Expansion) There exists $I\in \cI_G$ and $V'\subseteq H\setminus I$ such that $I'=I\cup V'$;
	\item[(ii)] (New MIS) $I'\subseteq H$.
\end{enumerate} 
\end{lemma}
\begin{proof}
Let $I' \in \mathcal{I}_{G'}$ be an arbitrary maximal independent set in $G'$. Define
\[
I_1 := I' \cap (V_G \setminus H), \quad I_2 := I' \cap H,
\]
so that $I' = I_1 \cup I_2$. Here, $I_1$ is the part of $I'$ outside $H$ and $I_2$ is the part of $I'$ inside $H$.

\noindent
\textbf{Case (i):} Suppose $I_1 \neq \emptyset$.
Then, $I_1$ is a subset of $I'$ consisting of vertices outside $H$, and since $G'$ and $G$ agree on edges outside $H$, $I_1$ is an independent set in $G$.
If $I_1$ is not a maximal independent set of $G$, extend $I_1$ to a maximal independent set $I \in \mathcal{I}_G$. 

Now define
\[
V' := I' \setminus I = (I_1 \cup I_2) \setminus I = I_2 \setminus I.
\]
Then $V' \subseteq H \setminus I$, and $I' = I \cup V'$. Since $I' \in \mathcal{I}_{G'}$ by assumption, this shows that $I'$ is a maximal independent set in $G'$ obtained by adding vertices from $H$ to an MIS $I$ of $G$.

\noindent
\textbf{Case (ii):} Suppose $I_1 = \emptyset$. Then $I' \subseteq H$, and $I'$ is a maximal independent set of $G'$.
\end{proof}

\begin{proof}[Proof of Proposition \ref{prop:comparison}]
 The result follows immediately from Lemma \ref{lem:MIS_expansion}.
\end{proof}

\subsubsection{Adding IA Monotonicity to the Exclusion Restriction}
As an illustration, we consider the canonical setting in which $\cY=[\underline\cY,\overline\cY]$, and the treatment and instrument are binary, a special case of Example \ref{ex:iv}. First, consider testing the exogeneity and exclusion restrictions:
\begin{equation} \label{E:exclusion1}
\begin{array}{l}
	Z \perp (\{Y(d, z)\}_{d \in \mathcal{D}, z \in \mathcal{Z}}, \{D(z)\}_{z \in \mathcal{Z}});\\[3mm]
	Y(d, z) =_{a.s.} Y(d, z'), \;\; \forall d \in \mathcal{D}, \forall z, z' \in \mathcal{Z}.
\end{array}
\end{equation}

The following is the sharp testable implication of this assumption, which follows from Proposition \ref{prop:exclusion_necessity} with $K=2$, and noting that $\cR^*$ is regular since $Z$ is binary.\footnote{For discrete outcomes, this inequality is equivalent to \citeposs{Pearl1995Testability} instrumental inequality.}
\begin{align}
 P(Y \in A, D = d \,|\,Z = 0)+P(Y \in A^c, D = d \,|\,Z = 1) \leq 1, \;\;\;  \forall d \in \{0,1\}, ~\forall A\subseteq \mathcal Y.	\label{eq:exclusion_only_K2}
\end{align}
Now we add the IA monotonicity 
\begin{equation} \label{E:exclusion_monotonicity1}
\begin{array}{l}
	D(1)\geq_{a.s.}D(0).
\end{array}
\end{equation}
Then, the testable implication becomes
\[
\begin{array}{c}
P(Y\in \cY,D=0|Z=1)+P(Y\in A,D=1|Z=0)+P(Y\in A^c,D=1|Z=1)\leq 1\\
P(Y\in A,D=0|Z=0)+P(Y\in A^c,D=0|Z=1)+P(Y\in \cY,D=1|Z=0)\leq 1.
\end{array}
\]
The inequalities above are equivalent to \eqref{eq:bphvk1}-\eqref{eq:bphvk2} (see Remark \ref{rem:existing_results} and Corollary \ref{cor:iv_K2}). Hence, they are sharp. Each of these inequalities results from tightening the original inequality under the exclusion restriction. For example, the first inequality is obtained by adding $P(Y \in \cY,D=0|Z=1)$ to $P(Y \in A, D=1|Z=0)+P(Y \in A^c, D=1 | Z=1)$, which is the left side of \eqref{eq:exclusion_only_K2} when $d=1$. The tightening occurs because the monotonicity restriction removes the edges $(v_{(A,1),0}, v_{(A,0),1})$ and $(v_{(A,1),0}, v_{(A^c,0),1})$ from the graph (see Figure \ref{fig:ex_mon}).

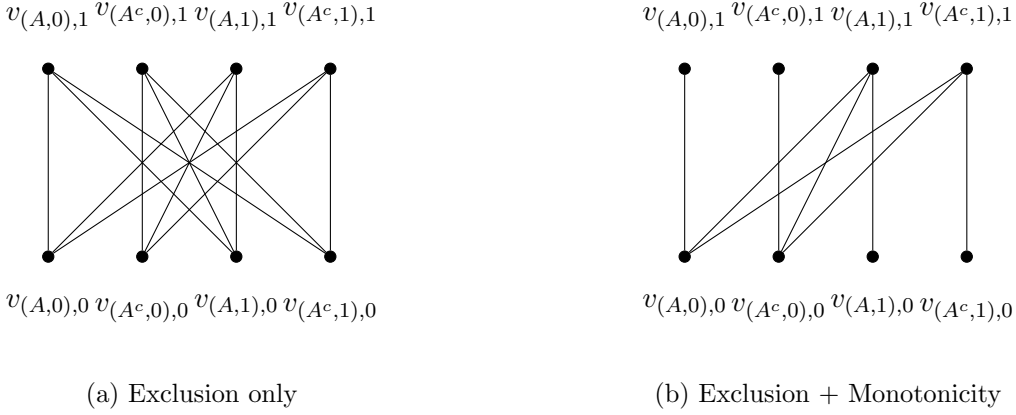
\begin{figure}[t]
\centering

\begin{subfigure}[t]{0.49\textwidth}
\centering
\begin{tikzpicture}[scale=1.25, every node/.style={circle, draw, fill=black, inner sep=1.5pt}]

\node (za00) at (0, -1) {};
\node (zc00) at (1, -1) {};
\node (za10) at (2, -1) {};
\node (zc10) at (3, -1) {};

\node (za01) at (0, 1) {};
\node (zc01) at (1, 1) {};
\node (za11) at (2, 1) {};
\node (zc11) at (3, 1) {};

\node[draw=none, fill=none, below] at (za00) {$v_{(A,0),0  }$};
\node[draw=none, fill=none, below] at (zc00) {$v_{(A^c,0),0}$};
\node[draw=none, fill=none, below] at (za10) {$v_{(A,1),0  }$};
\node[draw=none, fill=none, below] at (zc10) {$v_{(A^c,1),0}$};

\node[draw=none, fill=none, above] at (za01) {$v_{(A,0),1  }$};
\node[draw=none, fill=none, above] at (zc01) {$v_{(A^c,0),1}$};
\node[draw=none, fill=none, above] at (za11) {$v_{(A,1),1  }$};
\node[draw=none, fill=none, above] at (zc11) {$v_{(A^c,1),1}$};


\draw (za00) -- (za01);
\draw (za00) -- (za11);
\draw (za00) -- (zc11);
\draw (zc00) -- (zc01);
\draw (zc00) -- (za11);
\draw (zc00) -- (zc11);
\draw (za10) -- (za11);
\draw (za10) -- (za01);
\draw (za10) -- (zc01);
\draw (zc10) -- (zc11);
\draw (zc10) -- (za01);
\draw (zc10) -- (zc01);

\end{tikzpicture}
\caption{Exclusion only}
\end{subfigure}
\hfill
\begin{subfigure}[t]{0.49\textwidth}
\centering
\begin{tikzpicture}[scale=1.25, every node/.style={circle, draw, fill=black, inner sep=1.5pt}]

\node (za00) at (0, -1) {};
\node (zc00) at (1, -1) {};
\node (za10) at (2, -1) {};
\node (zc10) at (3, -1) {};

\node (za01) at (0, 1) {};
\node (zc01) at (1, 1) {};
\node (za11) at (2, 1) {};
\node (zc11) at (3, 1) {};

\node[draw=none, fill=none, below] at (za00) {$v_{(A,0),0  }$};
\node[draw=none, fill=none, below] at (zc00) {$v_{(A^c,0),0}$};
\node[draw=none, fill=none, below] at (za10) {$v_{(A,1),0  }$};
\node[draw=none, fill=none, below] at (zc10) {$v_{(A^c,1),0}$};

\node[draw=none, fill=none, above] at (za01) {$v_{(A,0),1  }$};
\node[draw=none, fill=none, above] at (zc01) {$v_{(A^c,0),1}$};
\node[draw=none, fill=none, above] at (za11) {$v_{(A,1),1  }$};
\node[draw=none, fill=none, above] at (zc11) {$v_{(A^c,1),1}$};


\draw (za00) -- (za01);
\draw (za00) -- (za11);
\draw (za00) -- (zc11);
\draw (zc00) -- (zc01);
\draw (zc00) -- (za11);
\draw (zc00) -- (zc11);
\draw (za10) -- (za11);
\draw (zc10) -- (zc11);

\end{tikzpicture}
\caption{Exclusion + Monotonicity}
\end{subfigure}

\caption{Potential Response Graph for Exclusion and Monotonicity }
\label{fig:ex_mon}
\end{figure}

\subsubsection{Partial Monotonicity and IA Motnotonicity}
As discussed in the text, removing the $Z_2$-complier results in the removal of an edge in the potential response graph. This is illustrated in Figure \ref{fig:pm_iam}.  Under the PM assumption, the latent types compatible with $v_{1,(0,1)}$ (i.e., $D=1$ when $Z=(0,1)$) include $2c,ec,$ and $at$. Similarly, those compatible with $v_{0,(1,0)}$ (i.e., $D=0$ when $Z=(1,0)$) are $nt,rc,$ and $2c$ (Panel (a)). Because they share the $Z_2$-complier, the sum of the conditional probabilities is not necessarily bounded by 1. Once we exclude the $Z_2$-complier, the latent types compatible with the two vertices become mutually exclusive (Panel (b)), resulting in \eqref{eq:pm_iam_new}.

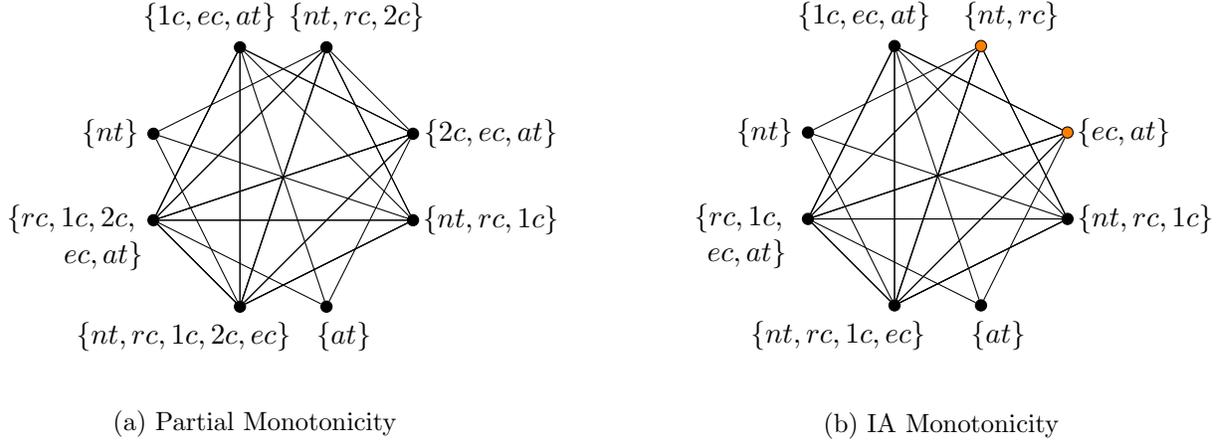
\begin{figure}[t]
\centering
\begin{subfigure}[t]{0.45\textwidth}
\centering
\begin{tikzpicture}[scale=1.15, every node/.style={draw, circle, fill=black, inner sep=1.5pt}, baseline=(current bounding box.north)]
\node (a0) at (0, 0) {};
\node (a1) at (1, 0) {};
\node (b0) at (2, 1) {};
\node (b1) at (2, 2) {};
\node (c0) at (1, 3) {};
\node (c1) at (0, 3) {};
\node (d0) at (-1, 2) {};
\node (d1) at (-1, 1) {};

\node[draw=none, fill=none] at (-0.65,-0.35) {$\{nt,rc,1c, 2c,ec\}$};
\node[draw=none, fill=none] at (1.2,-0.35) {$\{at\}$};

\node[draw=none, fill=none] at (2.9,1) {$\{nt,rc,1c\}$};
\node[draw=none, fill=none] at (2.9,2) {$\{2c,ec,at\}$};

\node[draw=none, fill=none] at (1.35,3.35) {$\{nt,rc,2c\}$};
\node[draw=none, fill=none] at (-0.35,3.35) {$\{1c,ec,at\}$};

\node[draw=none, fill=none] at (-1.5,2) {$\{nt\}$};
\node[draw=none, fill=none] at (-1.95,0.8) {$\begin{array}{c}
    \{rc,1c,2c, \\
     \;\;\;\;\;\;\;\;ec,at\} 
\end{array}$};


\draw (a0) -- (b0) -- (c0) -- (d0) -- (a0);
\draw (a0) -- (c0);
\draw (d0) -- (b0);
\draw (a0) -- (b0) -- (c0) -- (d1) -- (a0);
\draw (a0) -- (c0);
\draw (d1) -- (b0);
\draw (a0) -- (b0) -- (c1) -- (d1) -- (a0);
\draw (a0) -- (c1);
\draw (d1) -- (b0);
\draw (a0) -- (b1) -- (c0) -- (d1) -- (a0);
\draw (a0) -- (c0);
\draw (d1) -- (b1);
\draw (a0) -- (b1) -- (c1) -- (d1) -- (a0);
\draw (a0) -- (c1);
\draw (d1) -- (b1);
\draw (a1) -- (b1) -- (c1) -- (d1) -- (a1);
\draw (a1) -- (c1);
\draw (d1) -- (b1);
\end{tikzpicture} \vspace*{-8mm}
\caption{Partial Monotonicity}
\end{subfigure}
\hfill
\begin{subfigure}[t]{0.45\textwidth}
\centering
\begin{tikzpicture}[scale=1.15, every node/.style={draw, circle, fill=black, inner sep=1.5pt}, baseline=(current bounding box.north)]
\node (a0) at (0, 0) {};
\node (a1) at (1, 0) {};
\node (b0) at (2, 1) {};
\node [fill=orange](b1) at (2, 2) {};
\node [fill=orange](c0) at (1, 3) {};
\node (c1) at (0, 3) {};
\node (d0) at (-1, 2) {};
\node (d1) at (-1, 1) {};

\node[draw=none, fill=none] at (-0.65,-0.35) {$\{nt,rc,1c,ec\}$};
\node[draw=none, fill=none] at (1.2,-0.35) {$\{at\}$};

\node[draw=none, fill=none] at (2.9,1) {$\{nt,rc,1c\}$};
\node[draw=none, fill=none] at (2.65,2) {$\{ec,at\}$};

\node[draw=none, fill=none] at (1.35,3.35) {$\{nt,rc\}$};
\node[draw=none, fill=none] at (-0.35,3.35) {$\{1c,ec,at\}$};

\node[draw=none, fill=none] at (-1.5,2) {$\{nt\}$};
\node[draw=none, fill=none] at (-1.8,0.8) {$\begin{array}{c}
     \{rc,1c,  \\
     \;\;ec,at\}
\end{array}$};


\draw (a0) -- (b0) -- (c0) -- (d0) -- (a0);
\draw (a0) -- (c0);
\draw (d0) -- (b0);
\draw (a0) -- (b0) -- (c0) -- (d1) -- (a0);
\draw (a0) -- (c0);
\draw (d1) -- (b0);
\draw (a0) -- (b0) -- (c1) -- (d1) -- (a0);
\draw (a0) -- (c1);
\draw (d1) -- (b0);
\draw (a0) -- (b1); 
\draw (c0) -- (d1) -- (a0);
\draw (a0) -- (c0);
\draw (d1) -- (b1);
\draw (a0) -- (b1) -- (c1) -- (d1) -- (a0);
\draw (a0) -- (c1);
\draw (d1) -- (b1);
\draw (a1) -- (b1) -- (c1) -- (d1) -- (a1);
\draw (a1) -- (c1);
\draw (d1) -- (b1);
\end{tikzpicture} \vspace*{-5mm}
\caption{IA Monotonicity}
\end{subfigure}
\caption{Potential Response Graphs for Partial Monotonicity and IA Monotonicity.}
\label{fig:pm_iam}
\floatfoot{Note: The set representations of the two graphs. Each node is represented by the set of compatible latent types.
The orange vertices in Panel (b) indicate the new MIS ($\{v_{1,(0,1)},v_{0,(1,0)}\}$) under the IAM assumption.}
\end{figure}

\subsubsection{Restricting Mediation Channels}\label{ssec:mediation_comparison}
The sharp testable implication for the full mediation assumption (via $M_1$ and $M_2$) is as follows:
\begin{align}
P(Y=0,M_1=0,M_2=0 \mid D=0) + P(Y=1,M_1=0,M_2=0 \mid D=1) &\leq 1 \\
P(Y=0,M_1=0,M_2=0 \mid D=1) + P(Y=1,M_1=0,M_2=0 \mid D=0) &\leq 1 \\
P(Y=0,M_1=0,M_2=1 \mid D=0) + P(Y=1,M_1=0,M_2=1 \mid D=1) &\leq 1 \\
P(Y=0,M_1=0,M_2=1 \mid D=1) + P(Y=1,M_1=0,M_2=1 \mid D=0) &\leq 1 \\
P(Y=0,M_1=1,M_2=0 \mid D=0) + P(Y=1,M_1=1,M_2=0 \mid D=1) &\leq 1 \\
P(Y=0,M_1=1,M_2=0 \mid D=1) + P(Y=1,M_1=1,M_2=0 \mid D=0) &\leq 1 \\
P(Y=0,M_1=1,M_2=1 \mid D=0) + P(Y=1,M_1=1,M_2=1 \mid D=1) &\leq 1 \\
P(Y=0,M_1=1,M_2=1 \mid D=1) + P(Y=1,M_1=1,M_2=1 \mid D=0) &\leq 1.
\end{align}
If we strengthen the assumption to the mediation only through $M_1$, the sharp testable implication becomes
\begin{align}
&P(Y=0,M_1=0,M_2=0 \mid D=0) + P(Y=0,M_1=0,M_2=1 \mid D=0) \notag\\
&\quad + P(Y=1,M_1=0,M_2=0 \mid D=1) + P(Y=1,M_1=0,M_2=1 \mid D=1) \leq 1\\
&P(Y=0,M_1=0,M_2=0 \mid D=1) + P(Y=0,M_1=0,M_2=1 \mid D=1) \notag\\
&\quad + P(Y=1,M_1=0,M_2=0 \mid D=0) + P(Y=1,M_1=0,M_2=1 \mid D=0) \leq 1 \\
&P(Y=0,M_1=1,M_2=0 \mid D=0) + P(Y=0,M_1=1,M_2=1 \mid D=0) \notag\\
&\quad + P(Y=1,M_1=1,M_2=0 \mid D=1) + P(Y=1,M_1=1,M_2=1 \mid D=1) \leq 1 \\
&P(Y=0,M_1=1,M_2=0 \mid D=1) + P(Y=0,M_1=1,M_2=1 \mid D=1) \notag\\
&\quad + P(Y=1,M_1=1,M_2=0 \mid D=0) + P(Y=1,M_1=1,M_2=1 \mid D=0) \leq 1.
\end{align}

\subsection{Partial Monotonicity, MIS Inequalities and a Random Coefficient Model}\label{ssec:pm_random_coeff}
Consider the following selection model satisfying partial monotonicity:
\[
D(z_1, z_2) = \bm{1}(B_0 + B_1 z_1 + z_2 \geq 0),
\]
where $(B_0, B_1)$ are individual-specific random variables with $B_1 \geq_{a.s.} 0$.
Panel~(a) of Figure~\ref{fig:partial_monotonicity} partitions the $(B_0, B_1)$-space into regions defined by 
\[
\{(b_0,b_1)\in \mathbb R\times\mathbb R_+:\bm{1}(b_0 + b_1 z_1 + z_2 \geq 0) = d\}, \quad d \in \{0,1\}, \; z \in \{0,1\}^2.
\]
The conditional distribution of $D$ reflects the probability mass assigned to each such region. 
For instance, $P(D=1 \mid Z=(0,0)) = Q(B_0 \geq 0) = Q(at)$, while 
$P(D=0 \mid Z=(1,0)) = Q(B_0 + B_1 \leq 0) = Q(nt) + Q(rc) + Q(2c)$.
Events lying on opposite sides of a dashed line in the figure are mutually exclusive. 
Accordingly, the two observable events represented by $v_{1,(0,0)}$ and $v_{0,(1,0)}$ correspond to the regions 
$\{B_0 \geq 0\}$ and $\{B_0 + B_1 \leq 0\}$, respectively, as shown in Panel~(d) of Figure~\ref{fig:partial_monotonicity}. 
Since these regions are disjoint, the sum of their probabilities must not exceed one—a testable implication of the model. Hence, the MIS inequalities are equivalent to the testable implication derived from the random coefficient model above. We note that our approach does not require the existence of a structural model that is consistent with the support restriction.

\subsection{Relationships to Existing Inequalities} \label{ssec:existing_inequalities}
This section illustrates how the inequalities in this paper are related to the existing ones. 
\subsubsection{Binary Instruments and Artstein's Ienquality}\label{ssec:IV_artstein}
As discussed in Remark \ref{rem:artstein}, the inequality in Theorem \ref{thm:testable_implication} recovers Artstein's theorem when $Z$ is binary. 
Below, we show that Artestein's theorem can be used to derive existing sharp testable implications for settings with binary instruments.

Suppose $Z \in \{0, 1\}, D \in \mathcal{D}$ and $Y \in \mathcal{Y}$. Suppose that 
\begin{align}\label{E:star} 
\begin{array}{l}
	Y(d, z) = Y(d), \text{ for } z \in \{0, 1\}, \text{ for all }d \in \mathcal{D}\\[2mm]
	D(0) \leq_{a.s.} D(1).
\end{array}
\end{align}
Let $\cR^*$ be the support of $R^*=(Y(\cdot,\cdot),D(\cdot))$ satisfying the conditions above.
Let $\mathcal{R}_z = \mathcal{Y} \times \mathcal{D}$ for $z \in \{0, 1\}$, and consider random vectors $	(Y(D(0)), D(0)) \in \mathcal{R}_0$ and $	(Y(D(1)), D(1)) \in \mathcal{S}_1$. Define a correspondence $F: \mathcal{R}_0 \rightrightarrows \mathcal{R}_1$ via
\begin{align*}
F(y, d) &=\{r_1=(y_1,d_1) \in \mathcal{R}_1:(y,d,y_1,d_1)\in\cR^*\} \\
&= (\{y\} \times \{d\}) \cup \bigcup_{d' > d} (\mathcal{Y} \times \{d'\}).
\end{align*}
The set $F(y, d)$ contains all values of $(Y(D(1)), D(1))$ compatible with $(Y(D(0)), D(0)) = (y, d)$ under the modeling assumptions \eqref{E:star}.

For every closed set $A \subseteq \mathcal{S}_0$, denote
\[
 	 F^{-}(A) = \{(y, d) \in \mathcal{R}_0: F(y, d) \subseteq A\}.
\]
By Corollary 2.14 in \citet{molchanov2018random}, Assumption \eqref{E:star} holds if and only if
\begin{equation} \label{E:containment}
P((Y, D) \in A \,|\,Z = 1) \geq P((Y, D) \in F^-(A) \,|\,Z = 0), \text{ for all closed } A \subseteq \mathcal{R}_1.
\end{equation}
We remark that although checking the inequalities for all closed sets $A \subseteq \mathcal{R}_1$ is sufficient, they remain well-defined and should hold for all measurable subsets $A \subseteq \mathcal{R}_1$. 

Generic closed subsets of $\mathcal{R}_1$ take the form $A = \bigcup_{d \in C} B_d \times \{d\}$ for some closed sets $B_d \subseteq \mathcal{Y}$ and $C \subseteq \mathcal{D}$. It is easy to check that $F^-(A) = \varnothing$ for all $A$ except sets of the form
\[
A_{B,d} = (B \times \{d\}) \cup \bigcup_{d' > d} (\mathcal{Y} \times \{d'\}).
\]
Hence, it suffices to check \eqref{E:containment} for sets belonging to $\cA=\{A_{B,d},B\subseteq \cY,d\in \cD\}$. For this class, it holds that $F^-(A_{B,d})=A_{B,d},~\forall (B,d)$ which we show below. Hence, \eqref{E:containment} with $A=A_{B,d}$ can be written as
\begin{equation} \label{E:containment1}
P((Y, D) \in A_{B,d} \,|\,Z = 1) \geq P((Y, D) \in A_{B,d} \,|\,Z = 0).
\end{equation}

We now show $A_{B,d}\subseteq F^-(A_{B,d})$ and $F^-(A_{B,d})\subseteq A_{B,d}$.
First, take any $(y_0,d_0)\in A_{B,d}$, if $d_0>d$, then
$$F(y_0,d_0)=\{(y_0,d_0)\}\cup\bigcup_{d'>d_0}(\mathcal Y\times\{d'\})
\subseteq \bigcup_{d'>d}(\mathcal Y\times\{d'\})\subseteq A_{B,d}.$$
Hence, $A_{B,d}\subseteq F^-(A_{B,d})$ in this case. If $d_0=d$ and $y_0\in B$, then
$$F(y_0,d)=\{(y_0,d)\}\cup\bigcup_{d'>d}(\mathcal Y\times\{d'\})
\subseteq (B\times\{d\})\cup\bigcup_{d'>d}(\mathcal Y\times\{d'\})=A_{B,d},$$
so $(y_0,d_0)\in F^{-}(A_{B,d})$. Therefore, $A_{B,d}\subseteq F^-(A_{B,d})$. 

Next, take any $(y_0,d_0)\in F^{-}(A_{B,d})$. Then, $F(y_0,d_0)\subseteq A_{B,d}$. But $F(y_0,d_0)$ always contains $(y_0,d_0)$. Hence, $(y_0,d_0)\in A_{B,d}$.

We further rewrite \eqref{E:containment1}  as
\begin{multline*}
	P(D > d \,|\,Z = 1) + P(Y \in B, D = d \,|\,Z = 1) \geq\\ P(D > d \,|\,Z = 0) + P(Y \in B, D = d \,|\,Z = 0). 
\end{multline*}
Thus, denoting
$\Delta_d(B) = P(Y \in B, D = d \,|\,Z = 1) - P(Y \in B, D = d \,|\,Z = 0)$,
we obtain the following sharp testable implications
\begin{equation} \label{E:KIT1}
\Delta_d(B) \geq P(D > d \,|\,Z = 0) - P(D > d \,|\,Z = 1), \;\; \text{ for all closed $B \subseteq \mathcal{Y}$}.
\end{equation}
To relate these to the existing literature, notice that
\[
\Delta_d(B) + \Delta_d(B^c) = P(D = d \,|\,Z = 1) - P(D = d \,|\,Z = 0).
\]
Thus, evaluating \eqref{E:KIT1} at $B^c$ and using the above equality, we obtain 
\begin{equation} \label{E:KIT2}
\Delta_d(B) \leq P(D < d \,|\,Z = 0) - P(D < d \,|\, Z = 1),  \;\; \text{ for all closed $B \subseteq \mathcal{Y}$}.
\end{equation}
Now, consider a special case, $\mathcal{D} =  \{0, 1\}$, as in \cite{Kitagawa15}. Setting $d = 1$ in \eqref{E:KIT1} and $d = 0$ in \eqref{E:KIT2}, the right-hand sides are equal to zero, so we obtain
\[
\begin{array}{c}
	P(Y \in B, D = 1 \,|\,Z = 1) \geq P(Y \in B, D = 1 \,|\,Z = 0);\\[2mm]
	P(Y \in B, D = 0 \,|\,Z = 0) \geq P(Y \in B, D = 0 \,|\,Z = 1),
\end{array}
\]
which are precisely the inequalities in \cite{Kitagawa15}. 

\cite{Kwon:2024aa} consider the same setup with $|\mathcal{D}| > 2$. With some algebra, the inequalities obtained in their Remark 1 (Proposition B.1) are 
\[ 
\sup_{A}\Delta_{d}(A) \leq \min \left\{P(D < d \,|\, Z = 0) - P(D < d \,|\, Z = 1), \;0 \right\}. 
\]
Equation \eqref{E:KIT1} with $B = \mathcal{Y}$ implies $P(D < d \,|\,Z = 0) \leq P(D < d \,|\,Z = 1)$, meaning that the minimum in the above display can actually be removed.

\subsubsection{Generalized Instrumental Inequalities}\label{ssec:instrumental_inequalities}
Suppose $\cY=\{0,1\}$, $D=\{0,1\}$, and $\cZ$ is unrestricted. 
Consider the following inequalities:
\begin{align}
\max_{d\in\{0,1\}}
\left\{
\sup_{z\in\cZ} P(Y=1, D=d \mid Z=z)+
\sup_{z\in\cZ} P(Y=0, D=d \mid Z=z)
\right\}\leq 1,
\label{eq:ineq9}	
\end{align}
\begin{align}
&\begin{aligned}
\max_{d \in \{0,1\}} \max_{y \in \{0,1\}}&
\Big[
\sup_{z\in\cZ} P(Y=y, D=d \mid Z=z)\\
&-
\inf_{z\in\cZ} \!\left\{ P(Y=y, D=0 \mid Z=z) + P(Y=y, D=1 \mid Z=z) \right\}
\\[3pt]
&-
\inf_{z\in\cZ} \!\left\{ P(Y=y, D=d \mid Z=z) + P(Y=1-y, D=1-d \mid Z=z) \right\}
\Big]
\leq 0,
\end{aligned}
\label{eq:ineq10}
\end{align}
and
\begin{align}
\begin{aligned}
&
\inf_{z\in\cZ} \!\left\{ P(Y=1,D=0 \mid Z=z) + P(Y=1,D=1 \mid Z=z) \right\}
\\[3pt]
&+ \inf_{z\in\cZ} \!\left\{ P(Y=1,D=0 \mid Z=z) + P(Y=0,D=1 \mid Z=z) \right\}
\\[3pt]
&+ \inf_{z\in\cZ} \!\left\{ P(Y=0,D=0 \mid Z=z) + P(Y=0,D=1 \mid Z=z) \right\}
\\[3pt]
&+ \inf_{z\in\cZ} \!\left\{ P(Y=0,D=0 \mid Z=z) + P(Y=1,D=1 \mid Z=z) \right\}\geq 1.
\end{aligned}
\label{eq:ineq11}	
\end{align}
\cite{KedagniMourifie2020} show \eqref{eq:ineq9}-\eqref{eq:ineq11} are sharp for the potential outcome model satisfying the exclusion restriction and joint statistical independence $(Y(0),Y(1))\perp Z$.
The inequality in \eqref{eq:ineq9} is due to \cite{Pearl1995Testability}. The remaining inequalities are due to \cite{KedagniMourifie2020}.

To connect the inequalities above to the additional inequalities in \eqref{E:level_k}, 
let us define a graph $G=(V_G,E_G)$ with $V_G=\{(y,d,z)\}_{y\in\cY,d\in\cD,z\in\cZ}$, where $\cZ=\{z_1,\dots,z_K\}.$ An edge between $v=(y,d,z)$ and  $v'=(y',d',z')$ exists if and only if $(y,d)=(y',d')$ for $z\ne z'$ or $d\ne d'$ and $z\ne z'$, due to the exclusion restriction.

Let $V\subseteq V_G$. Consider the inequality
\begin{align}
\sum_{v_{r, z}\in V}P(R=r\,|\,Z=z)\leq \ell_G(V),~\label{eq:levelineq}
\end{align} 
as in \eqref{E:level_k}. Assumption \ref{A:support_regularity} (i) holds for this model.
Hence, $\ell_G(V)=\omega(G[V])$ for any $V\subseteq V_G$. However, Assumption \ref{A:support_regularity} (ii) is violated when $K\geq 3$ because $G$ admits induced odd cycles of length five  (see \eqref{eq:oddhole1} and \eqref{eq:oddhole2} below). By the Strong Perfect Graph Theorem \citep{chudnovsky2006strong}, $G$ is imperfect.

Let $V=\{(1,d,z),(0,d,z')\}$ for $d\in\{0,1\}$ and $z\in\cZ$. Any $C\in\cC^*_G$ can share at most one vertex with $V$ because it violates the exclusion restriction otherwise. Hence, $\ell_G(V)=1$. Substituting this set into \eqref{eq:levelineq} yields
\begin{align}
	P(Y=1,D=d|Z=z)+P(Y=0,D=d|Z=z')\leq 1.
\end{align}
Taking the supremum over $d\in \{0,1\}$ and $z,z'\in \cZ$ recovers Pearl's inequality \eqref{eq:ineq9}.

Similarly, let \begin{align}
V=\{(y,0,z),(1-y,1,z'),(1-y,0,z''),(1-y,0,z'),(y,1,z'')\}\label{eq:oddhole1}
\end{align}
or
\begin{align}
V=\{(y,1,z),(1-y,0,z'),(1-y,1,z''),(1-y,1,z'),(y,0,z'')\},\label{eq:oddhole2}
\end{align}
for $y\in \{0,1\}, (z,z',z'')\in\cZ^3$.
Observe that $G[V]$ forms a cycle of length 5, which has $\omega(G[V])=2$.   With some rearranging, both sets can be written as
\begin{align}
V=\{(y,d,z),(1-y,1,z'),(1-y,d,z''),(1-y,0,z'),(y,1-d,z'')\},	\label{eq:oddhole3}
\end{align}
for either $d=0$ or $d=1$.
Substituting \eqref{eq:oddhole3} into \eqref{eq:levelineq} and recalling that $\ell_G(V)=\omega(G[V])=2$,
\begin{multline*}
P(Y=y,D=d|Z=z)+P(Y=1-y,D=1|Z=z')+P(Y=1-y,D=d|Z=z'')\\
+P(Y=1-y,D=0|Z=z')+P(Y=y,D=1-d|Z=z'')\leq 2.
\end{multline*}
By taking complements and re-arranging them,
\begin{multline*}
P(Y=y,D=d|Z=z)\leq P(Y=y,D=0|Z=z')+P(Y=y,D=1|Z=z')\\
+P(Y=y,D=d|Z=z'')+P(Y=1-y,D=1-d|Z=z'')	
\end{multline*}
Taking sup over $z$ on the left-hand side and inf over $z'$ and $z''$ on the right-hand side of the expression above yields
\begin{multline*}
\sup_{z\in\cZ}P(Y=y,D=d|Z=z)\leq \inf_{z'\in\cZ}\{P(Y=y,D=0|Z=z')+P(Y=y,D=1|Z=z')\}\\
+\inf_{z''\in \cZ}\{P(Y=y,D=d|Z=z'')+P(Y=1-y,D=1-d|Z=z'')\}.
\end{multline*}
Taking the difference between both sides of the inequality and maximizing over $(y,d)\in \{0,1\}^2$ recovers \eqref{eq:ineq10}.

 Finally, consider
\begin{multline}
	V=\{(y,d,z),(y,1-d,z),(y,d,z'),(1-y,1-d,z'),\\
	(1-y,d,z''),(1-y,1-d,z''),(y,1-d,z'''),(1-y,d,z''')\}
\end{multline}
for $(y,d)\in\{0,1\}^2, (z,z',z'',z''')\in\cZ^4$. 

One can show numerically $\ell_G(V)=\omega(G[V])=3$. Substituting this into \eqref{eq:levelineq}
yields
\begin{equation}
\begin{aligned}
&P(Y=0,D=0 | Z=z) + P(Y=0,D=1 | Z=z)
\\[3pt]
&+ P(Y=0,D=0 | Z=z') + P(Y=1,D=1 | Z=z') 
\\[3pt]
&+  P(Y=1,D=0 | Z=z'') + P(Y=1,D=1 | Z=z'')
\\[3pt]
&+ P(Y=0,D=1 | Z=z''') + P(Y=1,D=0 | Z=z''') \leq 3,~(z,z',z'',z''')\in\mathcal Z^4.
\end{aligned}
\end{equation}
Following an argument similar to the previous case recovers \eqref{eq:ineq11}.
Hence, the sharp instrumental inequalities can be 
interpreted as those in \eqref{E:level_k} for specific choices of $V$.

\subsection{Proof of Theorem \ref{thm:rank_perfect}} \label{ssec:rank_perfect_proof}
Note that $\mathcal{S}_{\overline{G}} = \mathcal{C}^{\text{all}}_G$, where $\mathcal{C}^{\text{all}}_G$ is the collection of all (not necessarily maximal) cliques in $G$. Moreover, 
\[
\max_{I \in \mathcal{S}_{\overline{G}}}|I \cap V| = \max_{C \in \mathcal{C}_{G}^{\text{all}} } |C \cap V| = \max_{C \in \mathcal{C}_{G}}|C \cap V| = \ell_G(V),
\]
since the maximum over $C$ can always be attained by a maximal clique. Thus, rank-perfectness of $\overline{G}$ can be equivalently expressed as
\[
\text{ConvHull}(\bm{1}_C: C \in \mathcal{C}^{\text{all}}_G) = \{b \in \mathbf{R}^N_{+}: b'\bm{1}_{V} \leq \ell_G(V), \forall\, V \subseteq V_G\}.
\]
Thus, we can write
	\[
	\mathbb{P} = \text{ConvHull}(\bm{1}_C: C \in \mathcal{C}^{\text{all}}_G) \cap \{b \in \mathbf{R}^{N}_{+}: b'\bm{1}_{V_k} = 1, \text{ for all } k \in \{1, \dots, K\}\}.
	\]
On the other hand, under Assumption \ref{A:support_regularity} (i),
\[
\overline{\mathbb{P}} = \text{ConvHull}(\bm{1}_C: C \in \mathcal{C}_G),
\]
where $\mathcal{C}_G$ is the set of all maximal cliques in $G$. Note that we have $\overline{\mathbb{P}} \subseteq \text{ConvHull}(\bm{1}_C: C \in \mathcal{C}^{\text{all}}_G)$ trivially, and $\overline{\mathbb{P}} \subseteq \{b \in \mathbf{R}^{N}_{+}: b'\bm{1}_{V_k} = 1, \text{ for all } k \in \{1, \dots, K\}\}$ due to the fact that $\bm{1}_C'\bm{1}_{V_k} = 1$, for all $C \in \mathcal{C}_G$ and all $k \in \{1, \dots, K\}$. Thus,  $\overline{\mathbb{P}} \subseteq \mathbb{P}$. 

  Now, consider the reverse inclusion, $\mathbb{P} \subseteq \overline{\mathbb{P}}$. Any vector $b \in \text{ConvHull}(\bm{1}_C: C \in \mathcal{C}^{\text{all}}_G)$ takes the form $b = \sum_{C \in \mathcal{C}^{\text{all}}_G} Q(C) \bm{1}_C$, for some $Q(C) \geq 0$ with $\sum_{C \in \mathcal{C}^{\text{all}}_G} Q(C) = 1$. For any such $b$, since $|C \cap V_k| \in \{0, 1\}$, for any $C \in \mathcal{C}_G^{\text{all}}$ and $V_k$, the equality
   \[
   b'\bm{1}_{V_k} = \sum_{C \in \mathcal{C}_G}Q(C)|C \cap V_k| = 1
   \]
   holds if and only if $|C \cap V_k| = 1$, for all $C$ such that $Q(C) > 0$. Thus, $b \in \mathbb{P}$ if and only if $b = \sum_{C \in \mathcal{C}^{\text{all}}_G} Q(C) \bm{1}_C$ for a distribution $Q$ supported on the set of cliques $C$ taking one vertex from each $V_k$. This set is precisely $\mathcal{C}_G$, so $b \in \text{ConvHull}(\bm{1}_C: C \in \mathcal{C}_G) = \overline{\mathbb{P}}$, and the proof is complete. \qed

\subsection{Encouragement Designs and Discrete Choice Models} \label{app:encouragement}
Propositions \ref{L:helly_simple}--\ref{prop:tree} do not apply to the encouragement design model studied in \citet{BaiHuangTabordMeehan2026}. Nevertheless, as we show below, this model satisfies Assumption \ref{A:support_regularity}, so our Theorem \ref{thm:sharpness_of_mis} applies. 

\begin{example}[Encouragement Designs]
	Suppose that $\mathcal{Z} = \mathcal{D} = \{1, \dots, K\}$. Denote $D^* = (D(1), \dots, D(K), D_0)$ and $\mathcal{D}^* = \{(d_1,\dots,d_K,d_0) \in \mathcal{D}^{K+1}: d_z \in \{z, d_0\}, \forall z \in \mathcal{Z}\}$. That is, an instrument $Z=z$ encourages an individual to either take treatment $D(Z)=z$ or to choose a default option $D(Z)=D_0$, which can vary across individuals.

First, consider Assumption \ref{A:support_regularity}-(i). Consider the sets $	S(d, z) = \{D^* \in \mathcal{D}^*: D(z) = d\}$. If $d = z$,  we have $S(z, z) = \{D^* \in \mathcal{D}^*: D(z) = z\}$, so there are no restriction on $D_0$. In turn, if $d \ne z$, it must be $d_0 = d$, so $S(d, z) = \{D^* \in \mathcal{D}^{*}: D(z) = d,  D_0 = d\}.$

	Consider any collection  $\mathcal{T}_0 = \{(d_k, z_k)\}_{k \in \mathcal{K}_0}$ such that $S(d_k, z_k)$ are pairwise intersecting. The goal is to show that it has a joint intersection. If $\mathcal{T}_0$ contains two tuples $(d, z)$ and $(d', z')$ with $d \ne z$ and $d' \ne z'$, their pairwise intersection implies $D_0 = d = d'$. Thus, all pairs $(d, z) \in \mathcal{T}_0$ with $d \ne z$ must have the same $d$, denoted $d_0$. Moreover, if $\mathcal{T}_0$ contains a tuple $(d_0, z)$ with $d_0 \ne z$, then it cannot contain $(z, z)$ because the sets $S(d_0, z)$ and $S(z, z)$ cannot intersect. Thus, any pairwise intersecting $\mathcal{T}_0$ is of the form
	\[
	\{S(d_0, z), z \ne d_0, z \in A\} \cup \{S(z, z): z \in B\}
	\]
	with $A \cap B = \varnothing$. Any such collection has a joint intersection that contains, e.g., an element
	\[
	D^*(z) = \begin{cases}
		d_0 & z \in A\\
		z & z \notin A
	\end{cases}
	\]
	and therefore has a non-empty joint intersection. Therefore, Assumption \ref{A:support_regularity}-(i) holds.

Next, consider Assumption \ref{A:support_regularity}-(ii).  The vertices of the potential response graph are indexed by pairs $(d, z) \in \mathcal{Z}^2$. Moreover:

\begin{enumerate}
		\item The vertices $\{(z, z): z \in \mathcal{Z}\}$ form a clique.
		\item For each fixed $d$, vertices of the form $\{(d, z): z \ne d\}$ form a clique of size $|\mathcal{Z}| - 1$. The corresponding cliques for $d \ne d'$ are non-adjacent.  
		\item A vertex $(d, z)$ is adjacent to all vertices $(z', z')$ for $z' \ne z$.
	\end{enumerate}

	Thus, the vertices of $G$ can be partitioned into a clique and a disjoint union of cliques. Graphs with this property are called \textit{unipolar}, and are known to be perfect. We can also provide a direct proof. Let $V'\subseteq V_G$ and $G' = G[V']$ denote the corresponding induced subgraph. Let $	A = \{z: (z, z) \in V'\}$, and  $B_d = \{z \ne d: (d, z) \in V' \}$, for each $d \in \mathcal{Z}$. For each $d$, the set
	\[
	C = \{(z, z): z \in A\} \cup \{(d, z): z \in B_d\setminus A\} \subseteq V'
	\]
	forms a clique. Letting $d^* = \argmax_{d}|B_d \backslash A|$ and $m^* = |B_{d^*} \backslash A|$, we have $\omega(G') \geq |A| + m^*$. Moreover, there are no edges between $\{(d, z): z \in B_d\}$ and $\{(d', z): z \in B_{d'}\}$ for $d \ne d'$. Thus, we can consider the following coloring: (i) Color each vertex $(z, z)$ a different color, for a total of $|A|$ colors; (ii) Color each vertex of $B_{d^*}$ in $m^*$ additional different colors; (iii) Color each vertex of $B_{d}$ for $d \ne d^*$ using a subset of those $m^*$ colors. Thus, $\chi(G') \leq |A| + m^*$. Since $\chi(G') \geq \omega(G')$, we have $\chi(G') = \omega(G')$, and the proof is complete. \qed  
\end{example}

Finally, the following proposition verifies the claim in Example \ref{ex: discrete_choice}.

	\begin{proposition} \label{prop:discrete_choice}
		Suppose the instruments and responses can be ordered so that $z_1 \prec z_2 \prec \dots \prec z_K,$ and, for all $r \in \{0, \dots, J-1\}$,
		the difference $c_{r+1}(z) - c_r(z)$ is strictly increasing in $z$. Then, $\mathcal{R}^*$ is regular.
	\end{proposition}
	
	\begin{proof}
		We have $S(r, z) \cap S(r', z') \ne \varnothing$ for some $z \prec z'$ if and only if 	
		\[
	\exists v \in \mathbf{R}^J: \begin{cases}
		v_r - c_r(z) > v_{r'} - c_{r'}(z)\\
			v_{r'} - c_{r'}(z') > v_{r} - c_{r}(z') 
	\end{cases}
		 \iff c_r(z) - c_{r'}(z) < c_{r}(z') - c_{r'}(z').
	\]
	since all other components of $v$ can be made arbitrarily small. Under the stated assumption, this is only possible if $r \geq r'$. Similarly, a collection $\{S(r_1, z_1), \dots, S(r_K, z_K)\}$ has non-empty pairwise intersections if and only if $r_1 \geq r_2 \geq ... \geq r_K$. In this case, however, the joint intersection is also non-empty, so the Helly property is satisfied. 
	
	Since, for $z \prec z'$, $S(r, z) \cap S(r', z') \ne \varnothing$ holds if and only if $r \geq r'$, the potential response graph is a comparability graph for the strict partial order $(r, z) \prec (r', z') \iff z \prec z', r \geq r'.$ Thus, the potential response graph is perfect. Therefore, both conditions in Assumption \ref{A:support_regularity} hold, and $\mathcal{R}^*$ is regular.
	\end{proof}

\section{Computation and Inference} \label{app:computation}
\subsection{Master Algorithm}

A pseudo-code for constructing a potential response graph, computing MISs, and checking sharpness is given in Algorithm \ref{alg:master}.

\begin{myalgorithm}[Graph-Based Procedure for Testing Support Restrictions]\rm
\label{alg:master}
\hspace{1mm}\\
\textbf{Inputs:} Support restriction $\mathcal R^*$; The set of vertices: $V_G = \{ v_{r,z} : r \in \mathcal R_z,\ z \in \mathcal Z \}$

\medskip \noindent
\textbf{Step 1: Graph construction}

\begin{itemize}
\item \textbf{If} Pairwise Incompatibility holds (Method 1):
\begin{itemize}
    \item \textbf{For each} pair $(v_{r,z}, v_{r',z'}) \in V_G \times V_G$:
    \begin{itemize}
        \item \textbf{If} $(r,z)$ and $(r',z')$ are consistent under $\mathcal R^*$, add an edge between them.
    \end{itemize}
\end{itemize}

\item \textbf{Otherwise} (Method 2):
\begin{itemize}
    \item Enumerate all $r^* \in \mathcal R^*$ and construct $A^*$ as described in Sec. 3.3.
    \item \textbf{For each} pair $(v_{r,z}, v_{r',z'}) \in V_G \times V_G$:
    \begin{itemize}
        \item \textbf{If} there exists a column of $A^*$ with entries equal to $1$ in rows $(r,z)$ and $(r',z')$, add an edge between them.
    \end{itemize}
\end{itemize}
\end{itemize}

\medskip \noindent
\textbf{Step 2: Enumerate Maximal Independent Sets}

\begin{itemize}
\item Compute the collection
\[
\mathcal I_G = \{ I \subseteq V_G : I \text{ is a maximal independent set of } G \}
\]
using, for example, the algorithm of \citet{tsukiyama1977new}.
\end{itemize}

\medskip \noindent
\textbf{Step 3: Check Sharpness (optional)}

\noindent Initialize indicators $\mathbb{I}_1 = 0$ and $\mathbb{I}_2 = 0$; 

\begin{itemize}
\item Check whether $G$ is a perfect graph using, e.g., Algorithm \ref{alg:det_odd_holes}; \textbf{If TRUE}: Set $\mathbb{I}_1 = 1$.
\item \textbf{If} Method 1 was used, set $\mathbb{I}_2 = 1$.

\textbf{Otherwise}:
\begin{itemize}
    \item Enumerate all maximal cliques $C \in \mathcal C_G$.
    \item \textbf{For each} $C \in \mathcal C_G$: Check whether $C$ corresponds to $r^* \in \mathcal R^*$.
    
 \textbf{If} every $C$ corresponds to some support point:  Set $\mathbb{I}_2 = 1$.
\end{itemize}
\end{itemize}

\begin{itemize}
\item \textbf{If} $\mathbb{I}_1 = 1$ and $\mathbb{I}_2 = 1$, claim the sharpness of the MIS inequalities.

\textbf{Otherwise}: Cannot ensure sharpness of the MIS inequalities. \qed 
\end{itemize}
\end{myalgorithm}

\subsection{Computational costs of Methods 1 and 2}\label{ssec:method1_vs_method2}
As discussed in Section \ref{ssec:computation}, $G$ can be constructed either by checking pairwise (in)consistency of vertices or by enumerating all support points. We compare the computational cost of the two approaches using Example \ref{ex:iv} that imposes the exclusion restriction. 

Method 1 constructs $G$ by checking compatibility for each pair of nodes. 
Recall that
$|V_G| = |\mathcal Y|\, |\mathcal D|\, |\mathcal Z|.$
Checking the compatibility of two vertices $v=(y,d,z)$ and $v'=(y',d',z')$ 
 can be performed in $O(1)$ time. To construct $G$, we check every unordered pair of vertices in $V_G$.
There are
\begin{align*}
\binom{|V_G|}{2}
= 
\binom{|\mathcal Y|\,|\mathcal D|\,|\mathcal Z|}{2}
=
O\left(|\mathcal Y|^2\, |\mathcal D|^2\, |\mathcal Z|^2\right)
\end{align*}
such pairs. Hence, the total time complexity is
\begin{align*}
T_{\mathrm{graph}}
=
O\left(|\mathcal Y|^2\, |\mathcal D|^2\, |\mathcal Z|^2\right).
\end{align*}

Method 2 enumerates all types. Under exclusion, a latent type $R^*\in\mathcal R^*$ is fully determined by
\begin{align*}
\{Y(d)\}_{d\in\mathcal D} \in \mathcal Y^{\mathcal D},
\qquad
\{D(z)\}_{z\in\mathcal Z} \in \mathcal D^{\mathcal Z}.
\end{align*}
There are $|\mathcal Y|$ possible values for each $Y(d)$ and 
$|\mathcal D|$ possible values for each $D(z)$.
Thus, the total number of latent types is $M
= |\mathcal Y|^{|\mathcal D|}\, |\mathcal D|^{|\mathcal Z|}.$
Any algorithm that enumerates all latent types must output each element of
$\mathcal R^*$, so the time required is
\begin{align*}
T_{\mathrm{enumerate}}
= O\left(|\mathcal Y|^{|\mathcal D|}\, |\mathcal D|^{|\mathcal Z|}\right).
\end{align*}

Thus, $T_{\mathrm{enumerate}}$ grows exponentially in 
$|\mathcal D|$ and $|\mathcal Z|$, while $T_{\mathrm{graph}}$ grows only
polynomially.  
In sum, when the support restriction satisfies the 
pairwise incompatibility criterion, the graph 
$G$ can be constructed without enumerating latent types, even in settings where 
$M=|\mathcal R^*|$ is intractably large.

\subsection{Testing Graph Perfectness}

For any vertex $v \in V_G$, we denote the set of its neighbors by $N(v) \subseteq V_G$. For an arbitrary set of vertices $S \subseteq V_G$, we denote $N_{S}(v) = N(v) \cap S$ and $\overline{N}_{S}(v) = N_{S}(v) \cup \{v\}$. Similarly, for a set of vertices $C \subseteq V_G$, we denote $N_S(C) = \bigcup_{v \in C}N_S(v)$ and $\overline{N}_S(C) = N_S(C) \cup C$. For any subgraph $G' \subseteq G$, let $\omega(G')$ denote its clique number (size of the maximal clique).  Besides the potential response graph $G$ (and its graph complement), the two important inputs are: (i) An adjacency list: $\{N(v): v \in V_G\}$; and (ii) A list of parts, $\{V_k: k \in \{1, \dots, K\}\}$. 

\medskip

Odd holes of length five or more consist of wedges, i.e., subgraphs with three vertices and two edges, so the first step is to list all of them. 
\begin{myalgorithm}[Wedges] \label{alg:triads} \hspace{1mm}\\
\textbf{Inputs:} An adjacency list of $G$; A list of parts $V_k$ of $G$. \\
Initialize $\mathcal{O} = \varnothing$. \textbf{For each} $k \in \{1, \dots, K\}$, \textbf{for each} $v \in V_k$:
\begin{itemize}
 		\item Compute $N_{V_l}(v)$, for $l \ne k$. 
 		\item \textbf{For each} $l \ne k$: \textbf{for each} pair $v', v'' \in N_{V_l}(v)$, add $(v', v, v'')$ to $\mathcal{O}$.
 		\item \textbf{For each} $l< m \ne k$: \textbf{for each} pair of vertices $(v', v'') \in  N_{V_l}(v) \times (N_{V_m}(v) \backslash N_{V_m}(v'))$, add $(v', v, v'')$ to $\mathcal{O}$. 
\end{itemize}
\textbf{Return} $\mathcal{O}$.  \qed 	\\[2mm]
A simple calculation suggests that the worst-case complexity of Algorithm 1 is is $O(N^3)$.  For future reference, we let $N_{W}$ denote the number of wedges in $G$.  

\end{myalgorithm}

Let $(v', v, v'')$ be a wedge. Consider a pair of vertices $\tilde{v}' \in N(v') \backslash N(v'')\backslash N(v)$ and $\tilde{v}'' \in N(v'') \backslash N(v') \backslash N(v)$. If $\tilde{v}'$ and $\tilde{v}''$ are disconnected, the set of vertices $(\tilde{v}', v', v, v'', \tilde{v}'')$ induces a line graph with five vertices linked by four edges, which can potentially be expanded into a longer odd hole by further adding some neighbors of $\tilde{v}'$ and $\tilde{v}''$. If $\tilde{v}'$ and $\tilde{v}''$ are connected, this subgraph is an odd hole of length five, which cannot be expanded into a longer odd hole. Next, suppose a set of vertices $(v', V, v'')$ with $|V|$ odd induces a line graph with $|V| + 2$ vertices and $|V| + 1$ edges. Consider the vertices $\tilde{v}' \in N(v') \backslash N(v'')\backslash N(V)$ and $\tilde{v}'' \in N(v'') \backslash N(v') \backslash N(V)$. Similar to the above, if $\tilde{v}'$ and $\tilde{v}''$ are disconnected, the set of vertices $(\tilde{v}', v', V, v'', \tilde{v}'')$ induces a line graph with $|V| + 4$ vertices and $|V| + 3$ edges, which can potentially be expanded further, and if $\tilde{v}'$ and $\tilde{v}''$ are connected, it induces an odd hole which cannot be expanded further. These observations suggest the following algorithm.

\begin{myalgorithm}[Odd Holes from a Wedge]\label{alg:odd_holes} \hspace{1mm} \\
\textbf{Inputs:} An wedge $(v', v, v'')$; An adjacency list of  $G$.\\[2mm]
  Initialize $\mathcal{O} = \varnothing$ and a stack $\mathcal{P} = \{(v', V, v'')\}$ with $V = \{v\}$. \textbf{While} $\mathcal{P} \ne \varnothing$, do:

	\begin{itemize}
	\item Pop the first element from $\mathcal{P}$, denoted $(v', V, v'')$.
	\item \textbf{If} $N(v') \backslash N(v'') \backslash N(V) = \varnothing$ or $N(v'') \backslash N(v') \backslash N(V) = \varnothing$: \textbf{Next.} 
	\item \textbf{For each} $\tilde{v}' \in N(v') \backslash N(v'') \backslash N(V)$ and $\tilde{v}'' \in N(v'') \backslash N(v') \backslash N(V)$:
	\begin{itemize}
		\item \textbf{If} $\tilde{v}'$ and $\tilde{v}''$ are connected: add $(\tilde{v}', v', V, v'', \tilde{v}'')$ to $\mathcal{O}$; 
		\item \textbf{Else}: add $(\tilde{v}', \tilde{V}, \tilde{v}'')$, where $\tilde{V} = V \cup \{v', v''\}$, to $\mathcal{P}$.
	\end{itemize}
		\end{itemize}
		\item \textbf{Return} $\mathcal{O}$. \qed 
\end{myalgorithm}

To describe the complexity of this algorithm, let $\mathcal{L}_G^{\text{odd}}$ denote the class of all line subgraphs of $G$ with an odd number of vertices, and
\[
M = \max \limits_{(v', V, v'') \subseteq V_G} \left\{\frac{}{} \max(N(v') \backslash N(v'') \backslash N(V),\; N(v'') \backslash N(v') \backslash N(V)  ) \,\bigg | \, G[(v', V, v'')] \in \mathcal{L}_G^{\text{odd}} \frac{}{}\right\} 
\]
denote the maximum size of the ``exclusive'' neighborhood of one of the endpoints of a line subgraph in $G$. One can bound $M$ by the maximum degree, but it is typically smaller. If the longest odd hole in $G$ has length $2L + 1$, the worst-case complexity is  $O( M^{2L} )$.  Notably, the number of even holes in a graph does not matter. 

\medskip 
To find all odd holes in $G$, we may run Algorithm \ref{alg:odd_holes} starting from each wedge independently. The downside of this approach is that it recovers each odd hole of length $2l + 1$ exactly $2l + 1$ times; The upside is that it allows for parallelization. Let $N_{OH}$ denote the number of odd holes. An upper bound on the resulting complexity is $O(N_{W}M^{2L} + N_{OH}L)$, where the first summand is due to running Algorithm \ref{alg:odd_holes} for each wedge and the second one is due to removing duplicates from the resulting list. Although there is room for improvement, it is reassuring that the worst-case complexity scales linearly with the number of odd holes, so the algorithm is output-sensitive.

Moreover, the above algorithm can be adapted to detect an odd hole. 

\begin{myalgorithm}[Detecting an Odd Hole]\label{alg:det_odd_holes} \hspace{1mm} \\
\textbf{Inputs:} An adjacency list of  $G$.

\begin{enumerate}
	\item Find all wedges in $G$ using Algorithm 1 and add them to a stack $\mathcal{P}$. 
	\item \textbf{While} $\mathcal{P} \ne \varnothing$, do:
	\begin{itemize}
	\item Pop the first element from $\mathcal{P}$, denoted $(v', V, v'')$.
	\item \textbf{If} $N(v') \backslash N(v'') \backslash N(V) = \varnothing$ or $N(v'') \backslash N(v') \backslash N(V) = \varnothing$: \textbf{Next.} 
	\item \textbf{For each} $\tilde{v}' \in N(v') \backslash N(v'') \backslash N(V)$ and $\tilde{v}'' \in N(v'') \backslash N(v') \backslash N(V)$:
	\begin{itemize}
		\item \textbf{If} $\tilde{v}'$ and $\tilde{v}''$ are connected: \textbf{Return: TRUE} 
		\item \textbf{Else}: add $(\tilde{v}', \tilde{V}, \tilde{v}'')$, where $\tilde{V} = V \cup \{v', v''\}$, to $\mathcal{P}$.
	\end{itemize} 
		\end{itemize} 
	\item If no odd holes have been found, \textbf{Return: FALSE}\qed 
\end{enumerate}

\end{myalgorithm}
To verify that $G$ is perfect, one may run Algorithm \ref{alg:det_odd_holes} starting from $G$ and its graph complement $\overline{G}$.

\subsection{Inference with Continuous Covariates}\label{ssec:clr_details}
An estimator of $\sigma_n(v)$ can be obtained as follows. Let $K_n=|m_n||\cR||\cZ|$. Let $\bm{b}_n$ be a $K_n$-dimensional vector such that 
\begin{align}
\bm{b}_n(\x,r,z)=[\bm{0}'_{m_n},\dots,\bm{0}'_{m_n},b'_n(\x),\bm{0}'_{m_n},\dots,\bm{0}'_{m_n}]'	
\end{align}
with $b_n(\x)'$ appearing in the block corresponding to the specific $(r,z)$ value, and $\bm{0}'_{m_n}$ is an $m_n$-dimensional vector of zeroes. 
Let  $\bm{\chi}_n=(\chi_n(r,z)',r\in\cR,z\in\cZ)'$ be a $K_n$-dimensional vector.
We write
  \begin{align}
	\beta(r|z,\x)=b_n(\x)'\chi_n(r,z)+A_n(\x,r,z)=\bm{b}_n(\x,r,z)'\bm{\chi}_n+A_n(\x,r,z),
\end{align}
where $A_n(\x,r,z)$ is represents an approximation error.
Let $Q_n=E[b_n(\X_i)b_n(\X_i)']$ and 
\begin{align}
	\epsilon_i=\left(\left(\frac{\pi(z) \bm{1}(R_i = r, Z_i = z) - p(r, z) \bm{1}(Z_i = z)}{\pi(z)^2} \right)_{r \in \mathcal{R}} \right)_{z \in \mathcal{Z}}.
\end{align}
Then, the following asymptotic linear representation holds \citep[see][Section 4.2]{CLR13}:
  $$\sqrt n(\hat{\bm{\chi}}_n-\bm{\chi}_n)=(I_{|\cR||\cZ|}\otimes Q_n)^{-1}\frac{1}{\sqrt n}\sum_i \underbrace{(I_{|\cR||\cZ|}\otimes b_n(\X_i))\epsilon_i}_{u_i}+o_{P_n}(1/\ln n).$$ 
  Let $\bm{B}_n(\x)$ be a $|\cR||\cZ|$-by-$K_n$ matrix that stacks $\bm{b}_n(\x,r,z)'$ as its rows.
Let $\Omega_n=(I_{|\cR||\cZ|}\otimes Q_n)^{-1}E[u_iu_i'](I_{|\cR||\cZ|}\otimes Q_n)^{-1}$ and define
  \begin{align}
  \sigma_n(v)=\|n^{-1/2}a_I' \bm{B}_n(\x)\Omega_n^{1/2}\|.\label{eq:sigma_n}
\end{align}
Let  $\hat g(v)=a_I' \bm{B}_n(\x)\hat\Omega_n^{1/2}$, where $\hat\Omega_n$ is a consistent estimator of $\Omega_n$. The standard error $\sigma_n(v)$ can be estimated by $\hat\sigma_n(v)=\|\hat g(v)\|/\sqrt n$.
Given these objects, we may also compute CLR's critical value $\hat c_{n,\alpha}$ using their Algorithm 1, where we use $\hat g(v)=a_I' \bm{B}_n(\x)\hat\Omega_n^{1/2}$. 

\subsection{A Test Based on the Vertex Representation}\label{sec:fsst}
This section discusses how to conduct a test based on the vertex representation of $\mathbf B^*$ in \eqref{E:null}. Such a test is considered in \cite{fang2023inference} for the model of \cite{Imbens:1994aa} in their Example 2.2 and has recently beenextended to a more general setting by \cite{bai2025inference}.
Below, we outline how to apply the test to general potential outcome models, possibly with discrete covariates.

We first consider the setting without any covariates. Let $p=(N+1)$ and $d=M$, and construct a $p\times d$ matrix $A$ and a $d$-dimensional vector $\beta(P)$ as follows:
\begin{align}
A = \begin{bmatrix}
	A^* \\
	\bm{1}_{d}'
\end{bmatrix};
\;\;\;\;
\beta(P) = \begin{bmatrix}
	\beta_0(P)\\
	1
\end{bmatrix},\;\;\label{eq:A_beta}
\end{align}
where $A^*$ is the support matrix.\footnote{One can also accommodate continuous outcomes by introducing partitions of $\cY$ and redefining $A^*$ \citep[see][Sec. 5]{bai2025inference}.}
 We let $\hat\beta_n=(\hat\beta_{0,n})',1)'$ and let $\hat\Omega^e_n$ be 
\[
\hat\Omega^e_n=\begin{bmatrix}
	\hat V_n^{1/2} & 0\\
	0&0
\end{bmatrix},	
\]
where $\hat V_n$ is defined as in \eqref{eq:Vn}.  Let $\hat x^\star_n$ be a solution to the following quadratic program:
\begin{align}
\min_{x\in \mathbb R^d}\big(\hat\beta_{0,n}-A^*x\big)'\hat V_n^{-1}	\big(\hat\beta_{0,n}-A^*x\big)  \;\; s.t.\;\; \mathbf 1_d'x=1.
\end{align}
Following the derivation in Appendix M of \cite{fang2023inference}, let
\begin{align}
T_n^{\text{e}}=\big\|\sqrt n \hat V_n^{-1/2}\big(\hat\beta_{0,n}-A^*\hat x^\star_n\big)\big\|_\infty.
\end{align}
This statistic is zero if $\hat\beta_{0,n}$ is in the range of $A^*$. 
Again, following the derivation in Appendix M of \cite{fang2023inference}, let
\begin{align}
T_n^{\text{i}} = \sup_{\substack{s \in \mathbf{R}^p \\ x \in \mathbf{R}^d \\ \phi^+ \in \mathbf{R}^p_{+} \\ \phi^{-} \in \mathbf{R}^p_{+} }} \langle s, \sqrt{n} \hat{\beta}_n \rangle \;\; s.t.\;\; \left\{
\begin{array}{ll}
	Ax = s; &  \langle \bm{1}, \phi^{+} \rangle + \langle \bm{1}, \phi^{-} \rangle \leq 1;\\
	A's \leq 0;   &
	\phi^{+} - \phi^{-} = \hat{\Omega}_n^{\text{i}} s,
\end{array} \right.
\end{align}
where Let $\hat\Omega^i_n$ be the sample standard deviation of $\sqrt n A\hat x^\star_n$.
This is a linear program. Let $T_n=\max\{T_n^{\text{e}},T_n^{\text{i}}\}$. 

Computing the restricted estimator is delicate. We suggest
\[
\hat{\beta}_n^{r} = \argmin \limits_{\substack{s \in \mathbf{R}^p \\ x \in \mathbf{R}^d}}(\hat{\beta}_{0, n} - s)' \hat{W}_n(\hat{\beta}_{0, n} - s) \;\; s.t.\;\; \left\{ 
\begin{array}{c}
A_0x = s; \\
	x\geq 0;\\
	\langle \bm{1}, x\rangle = 1.
\end{array}
\right.
\]
Natural options of $\hat W_n$ include
\[
\hat{W}_n = I_p; \;\;\;\;\;\;\;\;\;\;\; \hat{W}_n = \text{diag}(\{\hat{V}^{-1}_{jj}\}_{j = 1}^p); \;\;\;\;\;\;\;\; \hat{W}_n = (\hat{V}_n + \epsilon I_{p})^{-1}. 
\]

Finally, we compute a critical value.
Let $(\hat\beta_{b,n},\hat x^\star_{b,n}),b=1,\dots,B$ be nonparametric bootstrap analogs of $(\hat\beta_n,\hat x^\star_n).$ Let
\begin{align}
\hat{\mathbb G}^e_n=\sqrt n\{(\hat\beta_{b,n}-A \hat x^\star_{b,n})-(\hat\beta_n-A\hat x^\star_n)\},~\hat{\mathbb G}^i_n=\sqrt nA(\hat x^\star_{b,n}-\hat x^\star_{n}).	
\end{align}
Define
\begin{align}
\hat T^{*,e}_n=	\big\|\sqrt n \hat V_n^{-1/2}\hat{\mathbb G}^e_n\big\|_\infty,
\end{align}
and 
\begin{align}
\hat T^{*,i}_n=	\sup_{\substack{s \in \mathbf{R}^p \\ x \in \mathbf{R}^d \\ \phi^+ \in \mathbf{R}^p_{+} \\ \phi^{-} \in \mathbf{R}^p_{+} }} \langle s, \hat{\mathbb G}^i_n+\sqrt{n} \lambda_n\hat{\beta}^r_n \rangle \;\; s.t.\;\; \left\{
\begin{array}{ll}
	Ax = s; &  \langle \bm{1}, \phi^{+} \rangle + \langle \bm{1}, \phi^{-} \rangle \leq 1;\\
	A's \leq 0;   &
	\phi^{+} - \phi^{-} = \hat{\Omega}_n^{\text{i}} s.
\end{array} \right.
\end{align}
The regularization parameter $\lambda_n$ can be chosen following Remark 4.2 of \cite{fang2023inference}. The critical value $\hat c_n(1-\alpha)$ is then defined as the $1-\alpha$ quantile of $\hat T^*_n=\max\{\hat T^{*,e}_n,\hat T^{*,i}_n\}$. The test rejects the null hypothesis when $T_n>\hat c_n(1-\alpha)$.

Now suppose one conditions on discrete covariates $W\in \cW=\{\x_1,\dots,\x_L\}.$ For simplicity, suppose $\cR^*$ does not depend on $W$. Let  $\beta_0(P,\x) = (\beta(r \,|\,z,\x)_{r \in \mathcal{R}})_{z \in \mathcal{Z}}\in \mathbb R^{N}$, and let $\beta_0(P)=(\beta_0(P,\x_1)',\dots,\beta_0(P,\x_L)')'\in\mathbb R^{NL}$. Let  $p=(N+1)L$, $d=ML$, and define $A\in\{0,1\}^{p\times d}$ and $\beta(P)$ by
\[
A = \begin{bmatrix}[0.8]
	A^*&0&\cdots&0 \\
	0&A^*&\cdots&0 \\
	\vdots&\cdots&\ddots&\vdots\\
	0&\cdots&\cdots&A^*\\
	\bm{1}_{M}'&0&\cdots&0\\
	\vdots&\vdots&\vdots&\vdots\\
	0&0&\cdots&\bm{1}_{M}'
\end{bmatrix};
\;\;\;\;
\beta(P) = \begin{bmatrix}[0.8]
	\beta_0(P)\\
	\bm{1}_{L}
\end{bmatrix}.\;\;
\]
If $\cR^*$ varies with $w$, let $A^*_{w}$ be the support matrix of $\cR^*_{w}$ and define
\[
A = \begin{bmatrix}[0.8]
	A^*_{w_1}&0&\cdots&0 \\
	0&A^*_{w_2}&\cdots&0 \\
	\vdots&\cdots&\ddots&\vdots\\
	0&\cdots&\cdots&A^*_{w_L}\\
	\bm{1}_{M}'&0&\cdots&0\\
	\vdots&\vdots&\vdots&\vdots\\
	0&0&\cdots&\bm{1}_{M}'
\end{bmatrix};
\;\;\;\;
\beta(P) = \begin{bmatrix}[0.8]
	\beta_0(P)\\
	\bm{1}_{L}
\end{bmatrix}.\;\;
\]
The rest of the analysis is then the same as before. When $W$ contains continuous components, the test of \cite{fang2023inference} does not readily apply. A possible way forward is to discretize $W$ into $L$ bins and apply the argument above.

\section{Additional Figures}

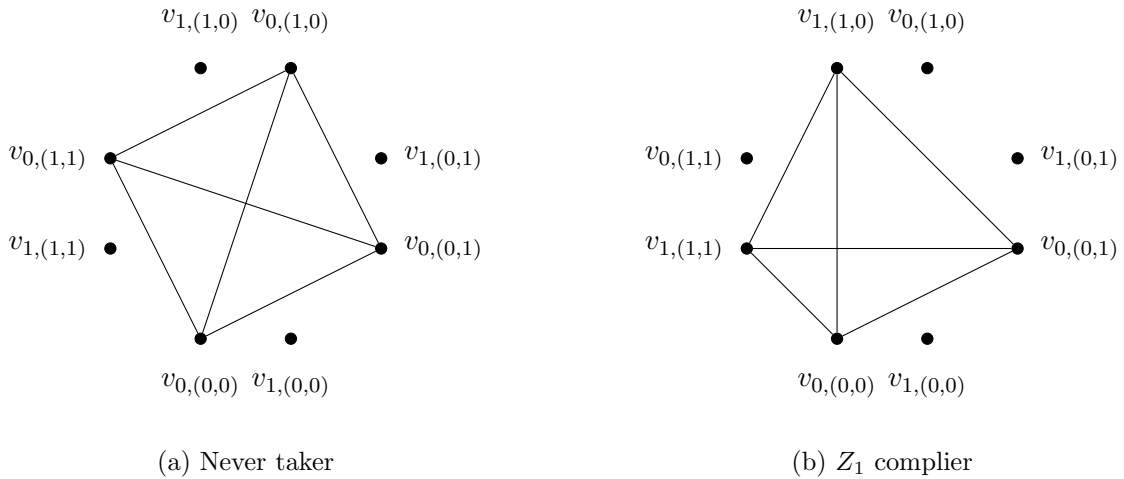
\begin{figure}[h]
\centering
\begin{subfigure}[t]{0.49\textwidth}
\centering
\begin{tikzpicture}[scale=1.2, every node/.style={draw, circle, fill=black, inner sep=1.5pt}]
\node (a0) at (0, 0) {};
\node (a1) at (1, 0) {};
\node (b0) at (2, 1) {};
\node (b1) at (2, 2) {};
\node (c0) at (1, 3) {};
\node (c1) at (0, 3) {};
\node (d0) at (-1, 2) {};
\node (d1) at (-1, 1) {};

\node[draw=none, fill=none, below] at (a0) {$v_{0,(0,0)}$};
\node[draw=none, fill=none, below] at (a1) {$v_{1,(0,0)}$};
\node[draw=none, fill=none, right] at (b0) {\;\;$v_{0,(0,1)}$};
\node[draw=none, fill=none, right] at (b1) {\;\;$v_{1,(0,1)}$};
\node[draw=none, fill=none, above] at (c0) {$v_{0,(1,0)}$};
\node[draw=none, fill=none, above] at (c1) {$v_{1,(1,0)}$};
\node[draw=none, fill=none, left]  at (d0) {$v_{0,(1,1)}$\;\;};
\node[draw=none, fill=none, left]  at (d1) {$v_{1,(1,1)}$\;\;};

\draw (a0) -- (b0) -- (c0) -- (d0) -- (a0);
\draw (a0) -- (c0);
\draw (d0) -- (b0);
\end{tikzpicture}
\caption{Never taker}
\end{subfigure}
\hfill
\centering
\begin{subfigure}[t]{0.49\textwidth}
\centering
\begin{tikzpicture}[scale=1.2, every node/.style={draw, circle, fill=black, inner sep=1.5pt}]
\node (a0) at (0, 0) {};
\node (a1) at (1, 0) {};
\node (b0) at (2, 1) {};
\node (b1) at (2, 2) {};
\node (c0) at (1, 3) {};
\node (c1) at (0, 3) {};
\node (d0) at (-1, 2) {};
\node (d1) at (-1, 1) {};

\node[draw=none, fill=none, below] at (a0) {$v_{0,(0,0)}$};
\node[draw=none, fill=none, below] at (a1) {$v_{1,(0,0)}$};
\node[draw=none, fill=none, right] at (b0) {\;\;$v_{0,(0,1)}$};
\node[draw=none, fill=none, right] at (b1) {\;\;$v_{1,(0,1)}$};
\node[draw=none, fill=none, above] at (c0) {$v_{0,(1,0)}$};
\node[draw=none, fill=none, above] at (c1) {$v_{1,(1,0)}$};
\node[draw=none, fill=none, left]  at (d0) {$v_{0,(1,1)}$\;\;};
\node[draw=none, fill=none, left]  at (d1) {$v_{1,(1,1)}$\;\;};

\draw (a0) -- (b0) -- (c1) -- (d1) -- (a0);
\draw (a0) -- (c1);
\draw (d1) -- (b0);
\end{tikzpicture}
\caption{$Z_1$ complier}
\end{subfigure}
\caption{Examples of Maximal Cliques representing Support Points in  Example \ref{ex:partial_monotonicity}.}
\label{fig:graph_partial_monotonicity_1}
\end{figure}

\begin{figure}[htbp]
	\centering 
    \begin{subfigure}{0.35\textwidth} \centering
		\includegraphics[width=\textwidth]{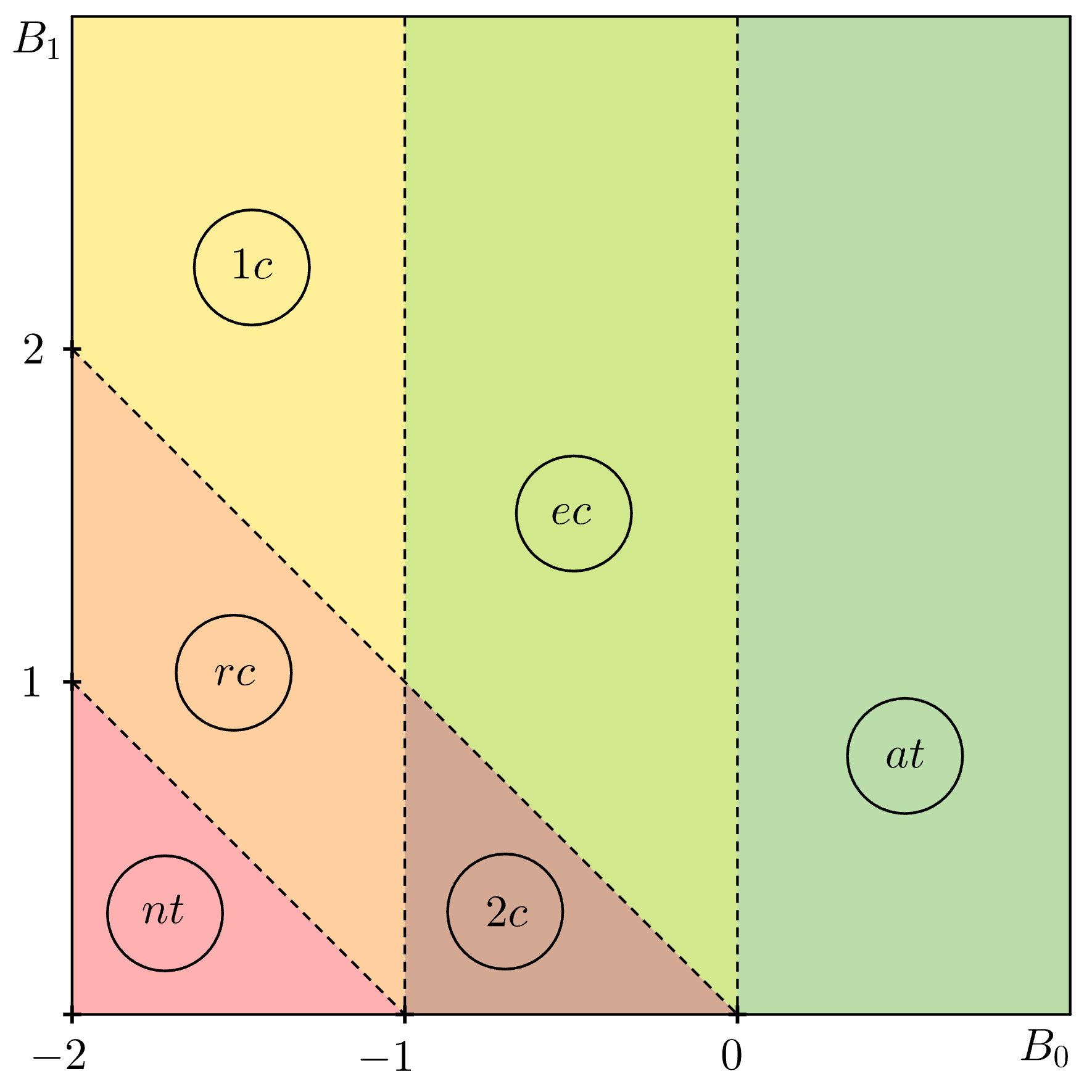}
		\subcaption{Original Partition}
	\end{subfigure} \hspace{1cm}
	\begin{subfigure}{0.35\textwidth} \centering
		\includegraphics[width=\textwidth]{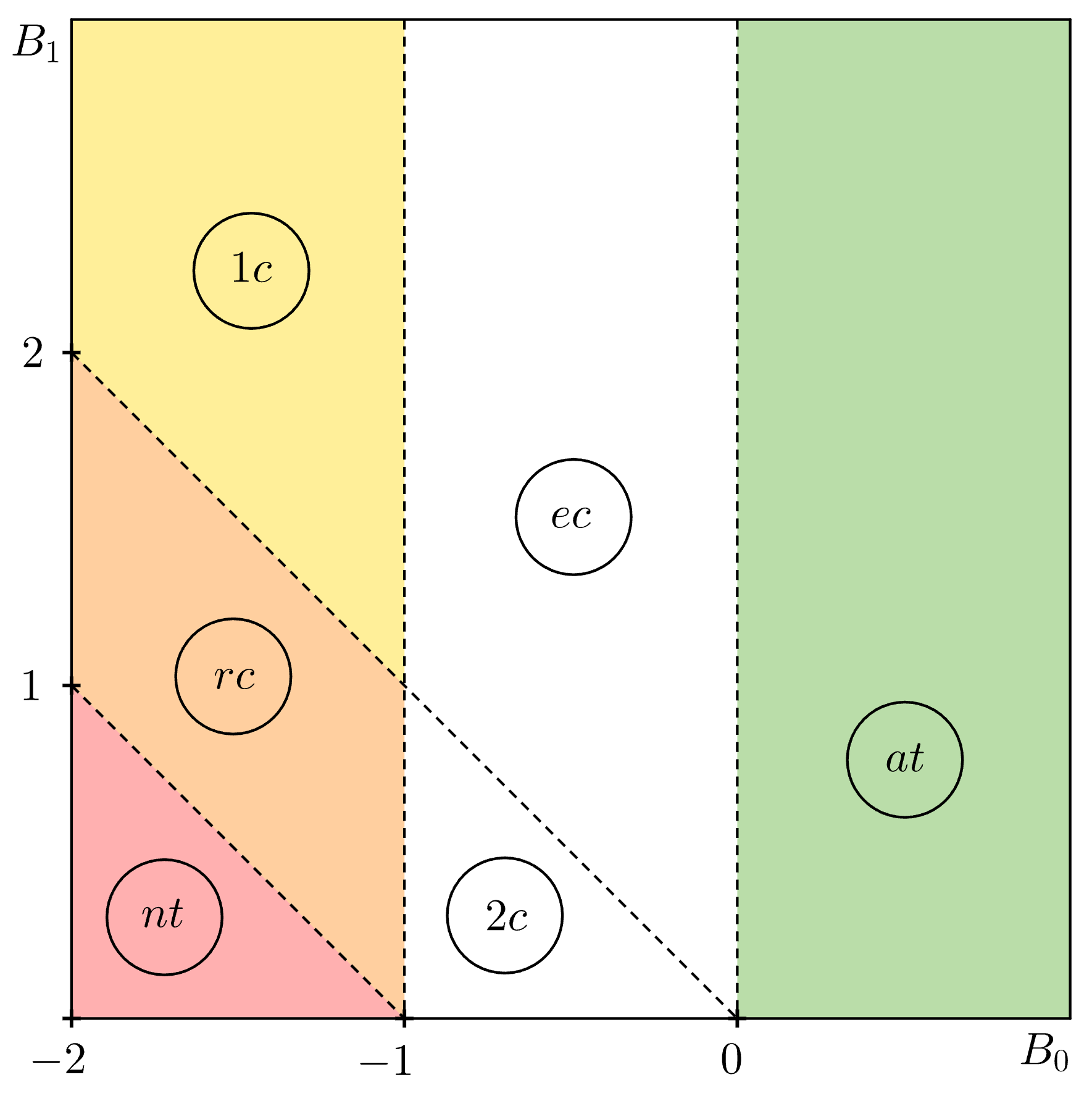}
		\subcaption{MIS 1}
	\end{subfigure}\\ \vspace{5mm}
	\begin{subfigure}{0.35\textwidth} \centering
		\includegraphics[width=\textwidth]{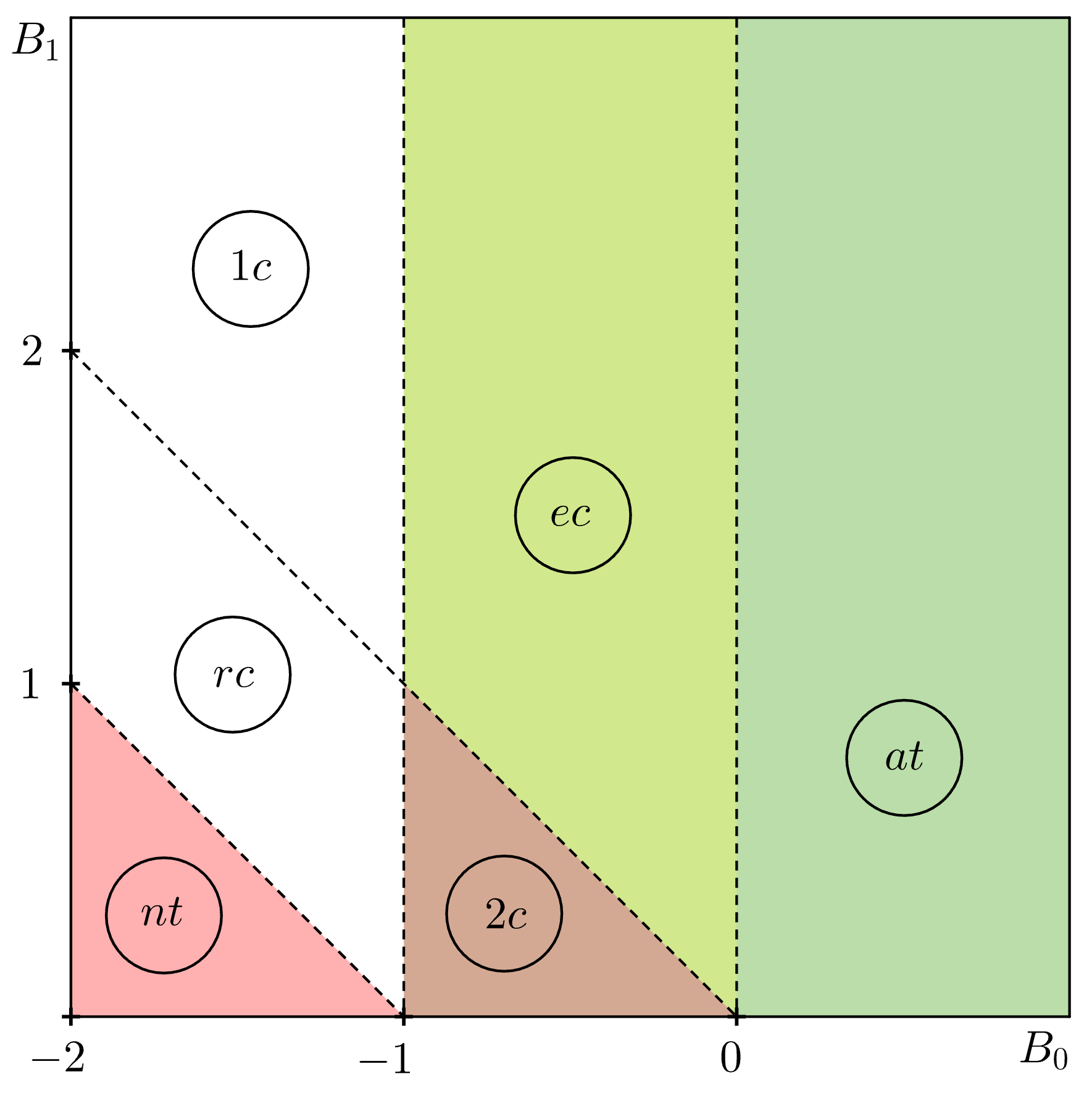}
		\subcaption{MIS 2}
	\end{subfigure} \hspace{1cm}
	\begin{subfigure}{0.35\textwidth} \centering
		\includegraphics[width=\textwidth]{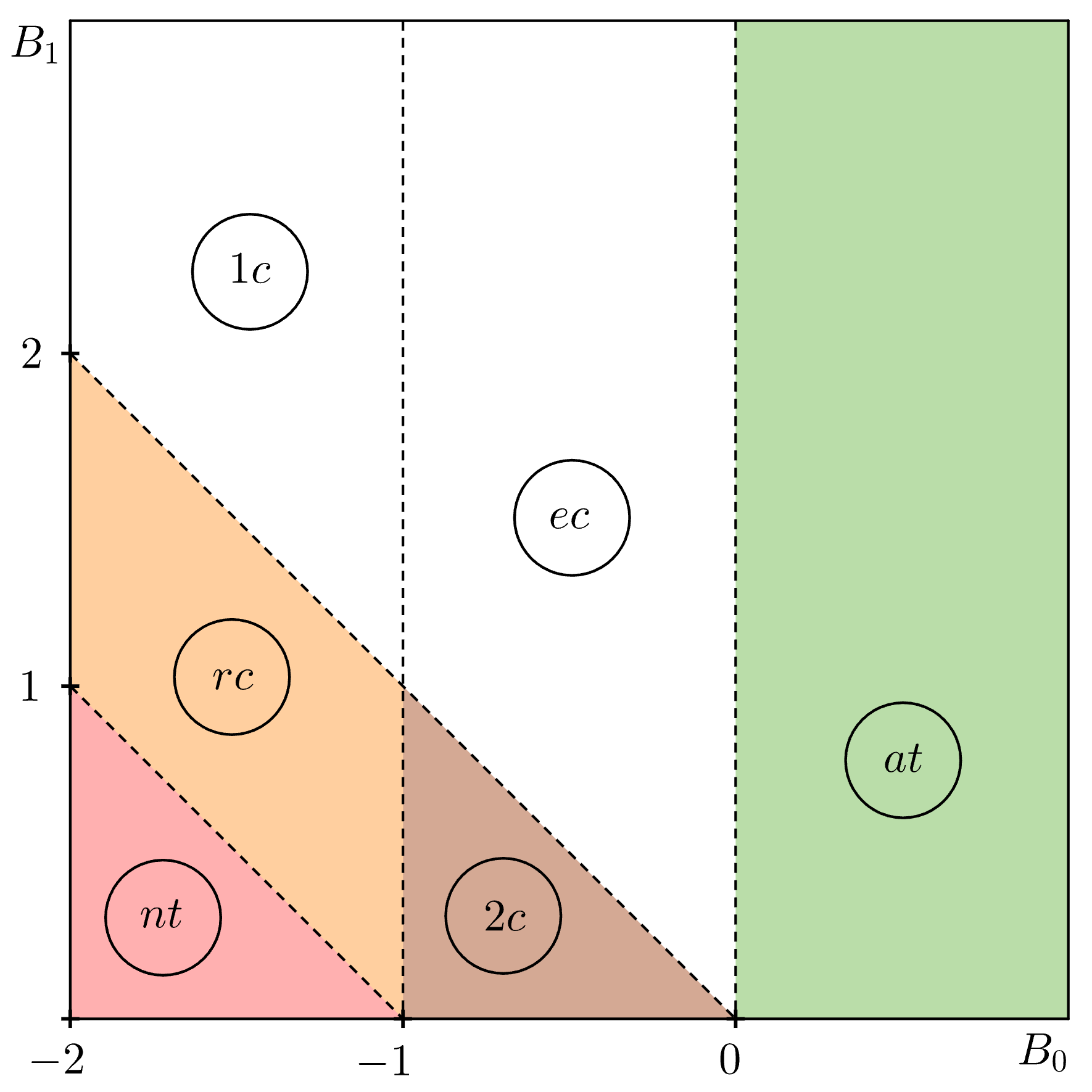}
		\subcaption{MIS 3}
	\end{subfigure} \\
    \vspace{5mm}
	\begin{subfigure}{0.35\textwidth} \centering
		\includegraphics[width=\textwidth]{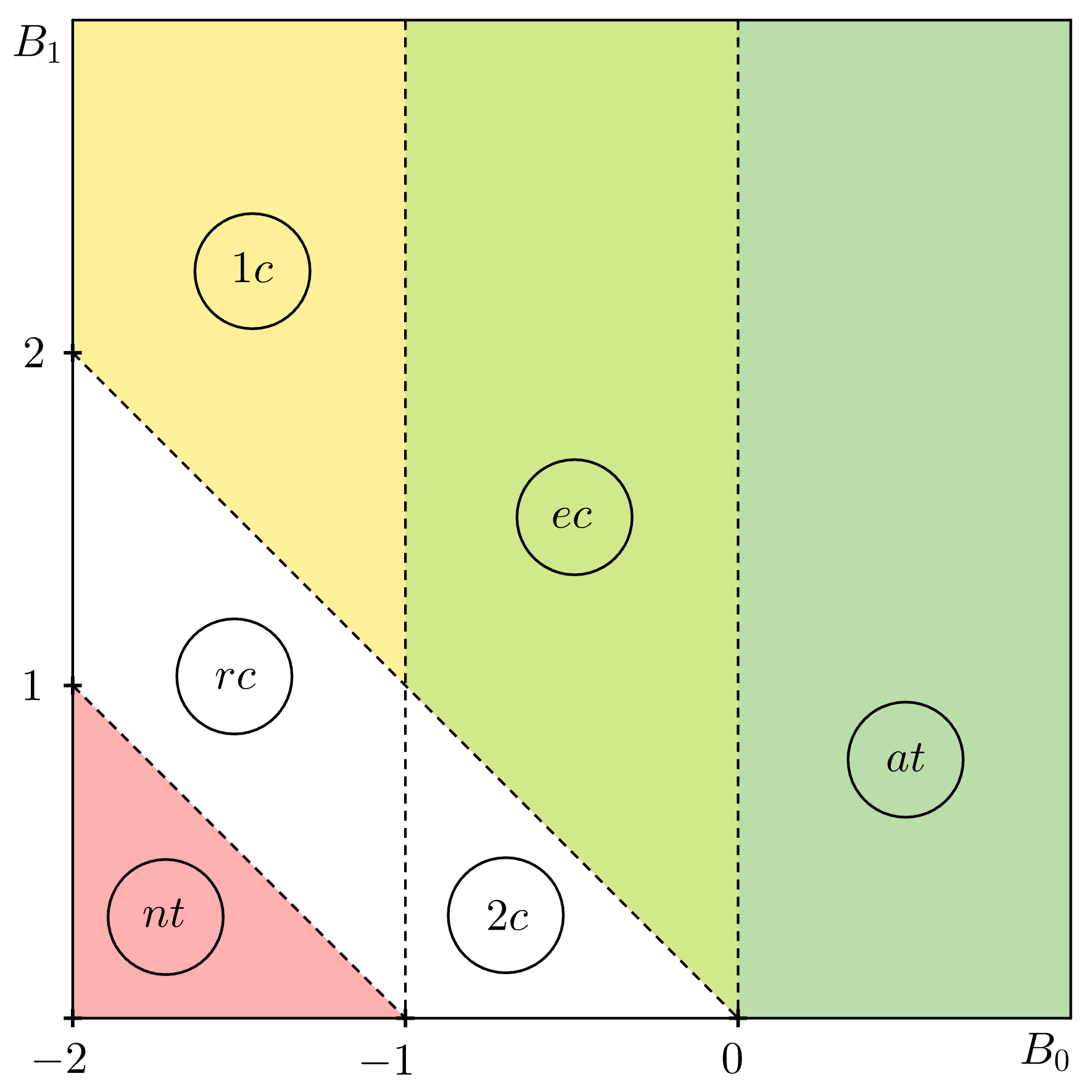}
		\subcaption{MIS 4}
	\end{subfigure} \hspace{1cm}
	\begin{subfigure}{0.35\textwidth} \centering
		\includegraphics[width=\textwidth]{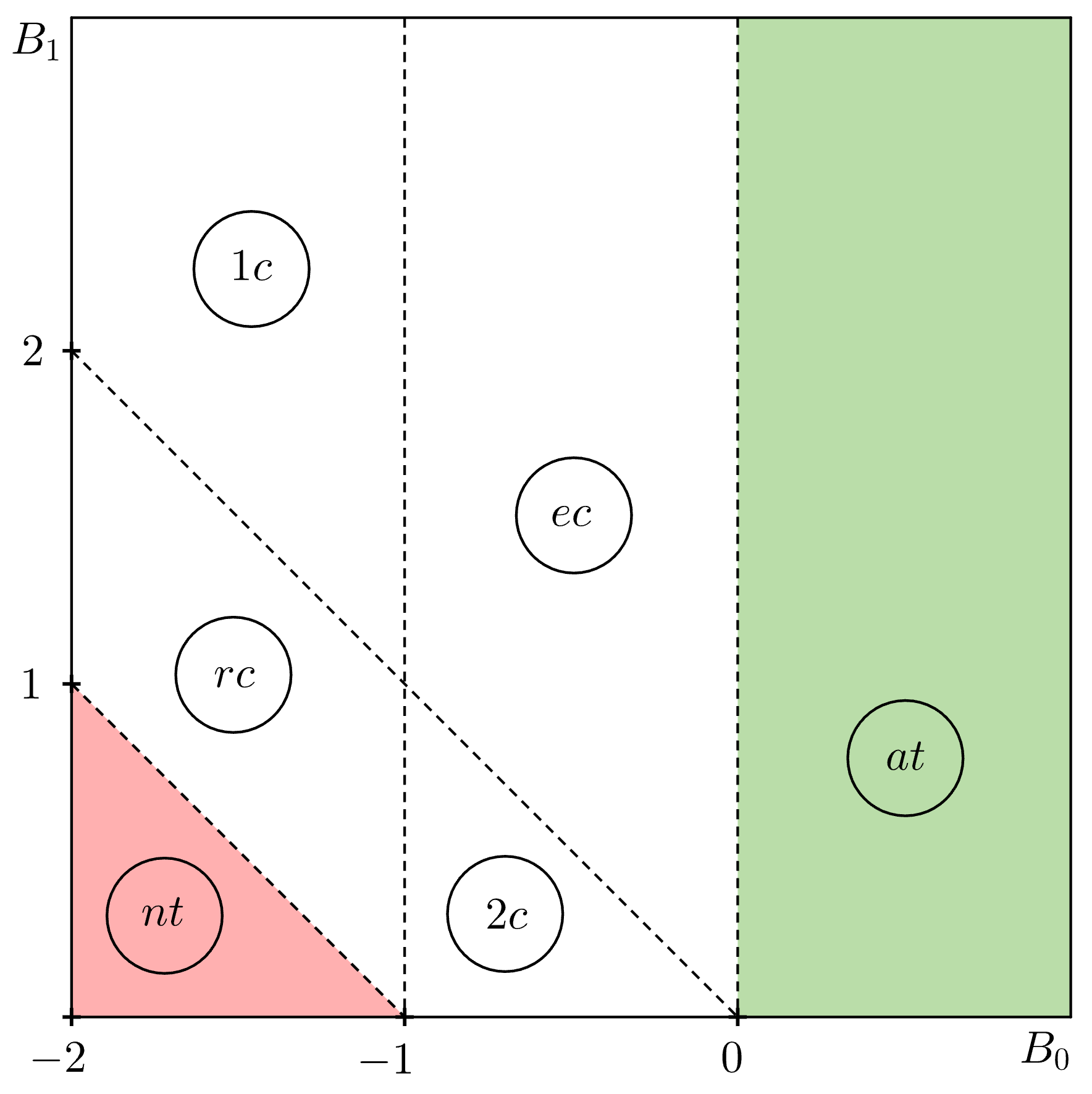}
		\subcaption{MIS 5}
	\end{subfigure}
	
	\caption{Inequalities for non-redundant Maximal Independent Sets in Example \ref{ex:partial_monotonicity}.} \label{fig:partial_monotonicity}
\end{figure}

\newpage 
\bibliography{HK_references}

\end{document}